\documentclass[english]{article}
\usepackage[T1]{fontenc}
\usepackage[latin9]{inputenc}
\usepackage{geometry}
\usepackage{color}
\usepackage{babel}
\usepackage{enumitem}
\usepackage{amsmath}
\usepackage{amsthm}
\usepackage{amssymb}
\usepackage[unicode=true,pdfusetitle,
 bookmarks=true,bookmarksnumbered=false,bookmarksopen=false,
 breaklinks=false,pdfborder={0 0 0},pdfborderstyle={},backref=false,colorlinks=true]
 {hyperref}
\hypersetup{
 linkcolor=blue,citecolor=blue,urlcolor=black}

\makeatletter

\providecommand{\tabularnewline}{\\}

\numberwithin{equation}{section}
\newcommand{\lyxaddress}[1]{
	\par {\raggedright #1
	\vspace{1.4em}
	\noindent\par}
}
\theoremstyle{plain}
\newtheorem{thm}{\protect\theoremname}
\theoremstyle{plain}
\newtheorem{prop}[thm]{\protect\propositionname}
\ifx\proof\undefined
\newenvironment{proof}[1][\protect\proofname]{\par
	\normalfont\topsep6\p@\@plus6\p@\relax
	\trivlist
	\itemindent\parindent
	\item[\hskip\labelsep\scshape #1]\ignorespaces
}{%
	\endtrivlist\@endpefalse
}
\providecommand{\proofname}{Proof}
\fi
\theoremstyle{definition}
\newtheorem{defn}[thm]{\protect\definitionname}
\theoremstyle{plain}
\newtheorem{cor}[thm]{\protect\corollaryname}
\theoremstyle{plain}
\newtheorem{lem}[thm]{\protect\lemmaname}

\date{}

\makeatother

\providecommand{\corollaryname}{Corollary}
\providecommand{\definitionname}{Definition}
\providecommand{\lemmaname}{Lemma}
\providecommand{\propositionname}{Proposition}
\providecommand{\theoremname}{Theorem}

\begin{document}
\title{Derivations for the MPS overlap formulas of rational spin chains}
\author{Tamas Gombor$^{a,b}$}
\maketitle

\lyxaddress{\begin{center}
$^{a}$MTA-ELTE \textquotedblleft Momentum\textquotedblright{} Integrable
Quantum Dynamics Research Group,\\
 Eötvös Loránd University, Pázmány Péter sétány 1/A, 1117 Budapest,
Hungary\\
$^{b}$HUN-REN Wigner Research Centre for Physics, \\
Konkoly-Thege Miklós u. 29-33, 1121 Budapest, Hungary
\par\end{center}}
\begin{abstract}
We derive a universal formula for the overlaps between integrable
matrix product states (MPS) and Bethe eigenstates in $\mathfrak{gl}_{N}$
symmetric spin chains. This formula expresses the normalized overlap
as a product of a MPS-independent Gaudin-determinant ratio and a MPS-dependent
scalar factor constructed from eigenvalues of commuting operators,
defined via the $K$-matrix associated with the MPS. Our proof is
fully representation-independent and relies solely on algebraic Bethe
Ansatz techniques and the $KT$-relation. We also propose a generalization
of the overlap formula to $\mathfrak{so}_{N}$ and $\mathfrak{sp}_{N}$
spin chains, supported by algebra embeddings and low-rank isomorphisms.
These results significantly broaden the class of integrable initial
states for which exact overlap formulas are available, with implications
for quantum quenches and defect CFTs.
\end{abstract}

\section{Introduction}

Overlaps between integrable boundary states and Bethe eigenstates
in quantum spin chains have emerged as an important building block
in the study of integrable systems. The integrable boundary states
allow exact computations of overlaps with Bethe eigenstates, often
yielding compact expressions involving Gaudin-like determinant ratios
and scalar prefactors. The structure of these overlaps is universal
across a wide class of models, including those with $\mathfrak{gl}_{N}$,
$\mathfrak{so}_{N}$, and $\mathfrak{sp}_{N}$ symmetries, and they
are relevant in both condensed matter and high-energy contexts.

In recent years, the study of non-equilibrium dynamics in integrable
models has gained significant momentum, driven both by experimental
advances \cite{Hackerm_ller_2010,Schneider2012FermionicTA,PhysRevLett.122.090601}
and theoretical developments \cite{Calabrese_2016,Bastianello_2022}.
A key focus has been on quantum quenches, where a system is initialized
in a specific state and its subsequent evolution is tracked. The Quench
Action method has emerged as a powerful framework for analyzing such
dynamics, particularly in determining the long-time steady states
\cite{Wouters_2014,Caux:2013ra,Essler:2016ufo}. Central to this approach
is the computation of overlaps between the initial state and the eigenstates
of the post-quench Hamiltonian. While the method has been successfully
applied to various integrable systems \cite{PhysRevLett.116.070408,Piroli:2018ksf,Rylands:2022naf,Rylands:2022gev},
much of the progress has relied on initial states with simple entanglement
structures, such as two-site product states, due to the tractability
of their overlap calculations.

In the context of the AdS/CFT correspondence, integrability methods
of the boundary states have provided powerful tools for analyzing
correlation functions in the presence of defects \cite{deLeeuw:2017cop,Kristjansen:2024dnm}.
A central development has been the realization that one-point functions
in defect configurations of the $\mathcal{N}=4$ super Yang-Mills
(SYM) and ABJM theories can be expressed as overlaps between multi-particle
Bethe states and special boundary or matrix product states (MPS) \cite{Buhl-Mortensen:2015gfd,deLeeuw:2015hxa,Kristjansen:2021abc,Gombor:2022aqj}.
This approach was initially applied to domain wall defects, such as
the D3-D5 brane setup, where tree-level one-point functions were linked
to MPS overlaps \cite{DeLeeuw:2018cal} and later extended to finite
coupling via integrable bootstrap \cite{Gombor:2020kgu,Gombor:2020auk,Komatsu:2020sup}.
These techniques have also been generalized to other configurations,
Wilson and \textquoteright t Hooft lines, where similar overlap structures
emerge \cite{Kristjansen:2023ysz,Kristjansen:2024map,Jiang:2023cdm,Gombor:2024api}.
Moreover, overlaps have proven essential in computing certain three-point
functions involving determinant and single trace operators, particularly
in the $AdS_{5}/CFT_{4}$ and $AdS_{4}/CFT_{3}$ dualities \cite{Jiang:2019zig,Jiang:2019xdz,Yang:2021hrl}.
More recent applications include the study of correlation functions
on the Coulomb branch of $\mathcal{N}=4$ SYM \cite{Ivanovskiy:2024vel},
and the surface defects of Gukov-Witten type in $\mathcal{N}=4$ SYM
\cite{Holguin:2025bfe,Chalabi:2025nbg}. The integrable conformal
defects in $\mathcal{N}=4$ SYM was also classified \cite{deLeeuw:2024qki}.
These developments underscore the universality of overlap structures
and their growing relevance in holographic context.

We can distinguish two types of boundary states: the simpler tensor
product states and the matrix product states. In the latter, an extra
vector space (boundary space) is introduced, and the boundary state
is constructed from the product of matrices acting on this space.
If this vector space is trivial (one-dimensional), then the MPS reduces
to a simple tensor product state. Boundary states can be related to
integrable boundary conditions \cite{Ghoshal:1993tm,Piroli:2017sei,Pozsgay:2018dzs}.
It has been observed that the overlaps of tensor product states with
Bethe states are proportional to the ratio of so-called Gaudin determinants
\cite{Brockmann_2014}. In early works, such overlap formulas were
proven for XXX and XXZ spin chains \cite{Foda:2015nfk,Jiang:2020sdw}.
These derivations heavily depended on the representation of the Hilbert
space and the specific form of the tensor product state, and it was
not clear how to generalize them to systems with multiple quasi-particles
(nested systems). In \cite{Gombor:2021uxz}, the so-called $KT$-relation
was introduced, whose main advantage is that it enables the calculation
of overlaps using purely algebraic Bethe Ansatz tools. As a result,
the derivations are universal, meaning they do not depend on the specific
representations of the Hilbert space. The $KT$-relation and the corresponding
overlap proof can be generalized to $\mathfrak{gl}_{N}$ symmetric
spin chains \cite{Gombor:2021hmj,Gombor:2023bez}. These papers proved
the overlap functions for all $\mathfrak{gl}_{N}$ symmetric spin
chains (independent of representation), for all integrable tensor
product states, which is a rather general result.

In this paper, we set an even more ambitious goal. We aim to determine
a universal overlap function for all integrable MPSs of all rational
spin chains. We prove these formulas for a broad class of integrable
MPSs of $\mathfrak{gl}_{N}$ spin chains. In the remaining cases,
we provide strong arguments supporting the validity of our formula.
The universal overlap formula takes the following form
\begin{equation}
\frac{\langle\mathrm{MPS}|\bar{u}\rangle}{\sqrt{\langle\bar{u}|\bar{u}\rangle}}=\underbrace{{\color{blue}\sum_{\ell=1}^{d_{B}}\beta_{\ell}\prod_{\nu=1}^{n_{+}}\prod_{k=1}^{r_{\nu}^{+}}\tilde{\mathcal{F}}_{\ell}^{(\nu)}(u_{k}^{\nu})}}_{\langle\mathrm{MPS}|\text{ dependent}}\times\underbrace{{\color{red}\sqrt{\frac{\det G^{+}}{\det G^{-}}}}}_{|\bar{u}\rangle\text{ dependent}},\label{eq:onOV-1}
\end{equation}
where $|\bar{u}\rangle$ is the Bethe state and $u_{k}^{\nu}$ are
the Bethe roots. The Gaudin determinants $\det G^{\pm}$ depend only
on the Bethe Ansatz equations. The quantities $\beta_{\ell}$ and
$\tilde{\mathcal{F}}_{\ell}^{(\nu)}(u)$ are eigenvalues of commuting
operators $\mathbf{B}$ and $\mathbf{F}^{(\nu)}(u)$. These operators
can be defined from the elements of the $K$-matrix associated with
$\langle\mathrm{MPS}|$. Assuming that a similar universal formula
exists for $\mathfrak{so}_{N}$ and $\mathfrak{sp}_{N}$ spin chains,
we determine the form of the operators $\mathbf{B}$ and $\mathbf{F}^{(\nu)}(u)$.
The calculation is based on the algebra embeddings $\mathfrak{sl}_{\left\lfloor \frac{N}{2}\right\rfloor }\subset\mathfrak{so}_{N}$,
$\mathfrak{sl}_{n}\subset\mathfrak{sp}_{2n}$ and the isometries of
low-rank cases: $\mathfrak{sp}_{2}\cong\mathfrak{sl}_{2}$, $\mathfrak{so}_{3}\cong\mathfrak{sl}_{2}$,
$\mathfrak{so}_{4}\cong\mathfrak{sl}_{2}\oplus\mathfrak{sl}_{2}$.
The universal overlap formula and the definitions of the corresponding
$\mathbf{B}$ and $\mathbf{F}^{(\nu)}(u)$ operators were already
published in a previous letter \cite{Gombor:2024iix} without proof.
In the current paper, we provide the proof for the $\mathfrak{gl}_{N}$
case and argue for the orthogonal and symplectic cases.

The paper follows the structure below. In the next section, we summarize
the necessary definitions of $\mathfrak{gl}_{N}$ symmetric spin chains.
In Section \ref{sec:Integrable-matrix-product}, we define the integrable
MPSs and the corresponding $K$-matrices for $\mathfrak{gl}_{N}$
spin chains. We show how the $KT$-relation can be used for the systematic
calculation of overlaps and introduce the nested $K$-matrices required
for this. In Section \ref{sec:Derivations-of-the-overlaps}, we present
theorems concerning the nested $K$-matrices and overlaps, including
the formula (\ref{eq:onOV-1}), which is the main result. The proofs
are found in the Appendix. In Section \ref{sec:Other-reflection-algebras},
we extend the definition of the operators $\mathbf{B}$ and $\mathbf{F}^{(\nu)}(u)$
to all $\mathfrak{gl}_{N}$ MPSs for which the formulas in Section
\ref{sec:Derivations-of-the-overlaps} are not applicable. In Section
\ref{sec:Orthogonal}, we extend the definition of the operators $\mathbf{B}$
and $\mathbf{F}^{(\nu)}(u)$ to $\mathfrak{so}_{N}$ and $\mathfrak{sp}_{N}$
symmetric spin chains as well.

\section{Definitions of the $\mathfrak{gl}_{N}$ symmetric spin chains\label{sec:Definitions}}

Let $T_{i,j}(z)$ be the generators of the Yangian $Y(N)$ algebra
\cite{Molev:1994rs}. They satisfy the $RTT$-relation
\begin{equation}
R_{1,2}(u-v)T_{1}(u)T_{2}(v)=T_{2}(v)T_{1}(u)R_{1,2}(u-v).\label{eq:RTT}
\end{equation}
The $\mathfrak{gl}_{N}$ symmetric $R$ matrix is
\begin{equation}
R(u)=1+\frac{1}{u}P,\qquad P=\sum_{i,j=1}^{N}e_{i,j}\otimes e_{j,i},\label{eq:Rm}
\end{equation}
where $e_{i,j}\in\mathrm{End}(\mathbb{C}^{N})$-s are the unit matrices
with the only nonzero entry equals to 1 at the intersection of the
$i$-th row and $j$-th column. The $RTT$-relation can be expressed
with the entries
\begin{equation}
\left[T_{i,j}(u),T_{k,l}(v)\right]=\frac{1}{u-v}\left(T_{k,j}(v)T_{i,l}(u)-T_{k,j}(u)T_{i,l}(v)\right).
\end{equation}
In this paper we will assume the following dependence of the Yangian
generators on the spectral parameter
\begin{equation}
T_{i,j}(u)=\delta_{i,j}+\sum_{\ell=0}^{\infty}u^{-\ell-1}T_{i,j}[\ell].
\end{equation}
The generators $E_{i,j}\equiv T_{j,i}[0]$ satisfy the $\mathfrak{gl}_{N}$
Lie-algebra relation
\begin{equation}
\left[E_{i,j},E_{k,l}\right]=\delta_{k,j}E_{i,l}-\delta_{i,l}E_{k,j}.
\end{equation}

Let $\mathcal{V}_{\Lambda}$ and $E_{i,j}^{\Lambda}\in\mathrm{End}(\mathcal{V}_{\Lambda})$
be vector space and a corresponding irreducible highest weight representation
of $\mathfrak{gl}_{N}$ for which
\begin{equation}
\begin{split}E_{i,j}^{\Lambda}|0_{\Lambda}\rangle & =0,\quad i<j,\\
E_{i,i}^{\Lambda}|0_{\Lambda}\rangle & =\Lambda_{i}|0_{\Lambda}\rangle,
\end{split}
\end{equation}
where $\Lambda=(\Lambda_{1},\dots,\Lambda_{N})$ is an $N$-tuple
of scalars and $|0_{\Lambda}\rangle\in\mathcal{V}_{\Lambda}$ is the
highest weight state. For finite dimensional representations $\Lambda_{i}-\Lambda_{i+1}\in\mathbb{N}$
for $i=1,\dots,N-1$. The Lax operators $L_{i,j}^{\Lambda}(u)\in\mathrm{End}(\mathcal{V}_{\Lambda})$
are the evaluation representations of $Y(N)$
\begin{equation}
L_{i,j}^{\Lambda}(u)=\delta_{i,j}+\frac{1}{u}E_{j,i}^{\Lambda}.\label{eq:Lax}
\end{equation}
The Lax-operators satisfy the $RTT$-relation. From the co-product
property of the Yangian, we can define tensor product representations
on $\mathcal{H}=\mathcal{V}_{\Lambda^{(1)}}\otimes\dots\otimes\mathcal{V}_{\Lambda^{(J)}}$.
The $\mathcal{H}$ is the quantum space. The monodromy matrix is the
corresponding representation of the Yangian
\begin{equation}
T_{0}^{\bar{\Lambda},\bar{\xi}}(u)=L_{0,J}^{\Lambda^{(J)}}(u-\xi_{J})\dots L_{0,1}^{\Lambda^{(1)}}(u-\xi_{1}).
\end{equation}
 The monodromy matrix is a \emph{highest weight representation} of
the Yangian, i.e.,
\begin{equation}
\begin{split}T_{i,j}^{\bar{\Lambda},\bar{\xi}}(u)|0\rangle & =0,\quad i>j,\\
T_{i,i}^{\bar{\Lambda},\bar{\xi}}(u)|0\rangle & =\lambda_{i}(u)|0\rangle,
\end{split}
\end{equation}
where $|0\rangle=|0_{\Lambda^{(1)}}\rangle\otimes\dots\otimes|0_{\Lambda^{(J)}}\rangle$
is the pseudo-vacuum and $\lambda_{i}(u)$-s are the pseudo-vacuum
eigenvalues which can be expressed as
\begin{equation}
\lambda_{i}(u)=\prod_{j=1}^{J}\frac{u-\xi_{j}+\Lambda_{i}^{(j)}}{u-\xi_{j}}.
\end{equation}
For simplicity, we will omit the upper index from now on and use the
notation $T_{i,j}(u)$ for arbitrary highest-weight representations
of the Yangian algebra.

The crossed monodromy matrix is defined as an inverse
\begin{equation}
\sum_{k=1}^{N}\widehat{T}_{k,i}(z)T_{k,j}(z)=\lambda_{1}(z)\lambda_{1}(-z)\delta_{i,j}.\label{eq:Thatdef}
\end{equation}
The crossed monodromy matrices satisfy the following algebra
\begin{equation}
\begin{split}R_{1,2}(u-v)\widehat{T}_{1}(u)\widehat{T}_{2}(v) & =\widehat{T}_{2}(v)\widehat{T}_{1}(u)R_{1,2}(u-v),\\
\widehat{R}_{1,2}(u-v)\widehat{T}_{1}(u)T_{2}(v) & =T_{2}(v)\widehat{T}_{1}(u)\widehat{R}_{1,2}(u-v).
\end{split}
\end{equation}
It follows from the first equation that the crossed monodromy matrix
is also a representation of the Yangian algebra. In the second equation
we defined the crossed $R$-matrix
\begin{equation}
\widehat{R}_{1,2}(u)=R_{1,2}^{t_{1}}(-u),\label{eq:Rb}
\end{equation}
where $t_{1}$ is the partial transposition corresponding to the space
$1$. The crossed monodromy matrices are \emph{lowest weight representations}
of the Yangians, i.e.
\begin{equation}
\begin{split}\widehat{T}_{i,j}(u)|0\rangle & =0,\quad i<j,\\
\widehat{T}_{i,i}(u)|0\rangle & =\hat{\lambda}_{i}(u)|0\rangle,
\end{split}
\end{equation}
where $|0\rangle$ is the highest weight vector (pseudo-vacuum) of
the monodromy matrix $T_{i,j}$ and the crossed eigenvalues are \cite{Liashyk:2018egk}\footnote{In this paper, we use a different convention for the crossed monodromy
matrix than \cite{Liashyk:2018egk}. The connection is $\widehat{T}_{i,j}\leftrightarrow\widehat{T}_{N+1-i,N+1-j}$.}
\begin{equation}
\hat{\lambda}_{i}(u)=\frac{\lambda_{1}(u)\lambda_{1}(-u)}{\lambda_{i}(u-(i-1))}\prod_{k=1}^{i-1}\frac{\lambda_{k}(u-k)}{\lambda_{k}(u-(k-1))}.\label{eq:lamhat}
\end{equation}
We also introduce the $\alpha$-functions
\begin{equation}
\alpha_{i}(u)=\frac{\lambda_{i}(u)}{\lambda_{i+1}(u)},\qquad\hat{\alpha}_{i}(u)=\frac{\hat{\lambda}_{i}(u)}{\hat{\lambda}_{i+1}(u)}.
\end{equation}
The crossed $\alpha$-functions can expressed from the original ones
as
\begin{equation}
\hat{\alpha}_{i}(u)=1/\alpha_{i}(u-i).
\end{equation}

The transfer matrices can be defined in the usual way 
\begin{equation}
\mathcal{T}(u)=\sum_{i=1}^{N}T_{i,i}(u),\qquad\widehat{\mathcal{T}}(u)=\sum_{i=1}^{N}\widehat{T}_{i,i}(u).
\end{equation}
These are commutative operators
\begin{equation}
\left[\mathcal{T}(u),\mathcal{T}(v)\right]=\left[\mathcal{T}(u),\widehat{\mathcal{T}}(v)\right]=\left[\widehat{\mathcal{T}}(u),\widehat{\mathcal{T}}(v)\right]=0.
\end{equation}
The eigenvectors $\mathbb{B}(\bar{t})$ are called Bethe vectors.
The Bethe vectors depend on the Bethe roots $t_{k}^{\nu}$ for $\nu=1,\dots,N-1$
and $k=1,\dots,r_{\nu}$. We define the sets $\bar{t}^{\nu}=\{t_{k}^{\nu}\}_{k=1}^{r_{\nu}}$
and the set of sets $\bar{t}=\{\bar{t}^{\nu}\}_{\nu=1}^{N-1}$. The
$r_{\nu}$ quantum numbers are the cardinalities of the sets $\bar{t}^{\nu}$.
The Bethe roots satisfy the Bethe equations
\begin{equation}
\alpha_{\mu}(t_{k}^{\mu}):=\frac{\lambda_{\mu}(t_{k}^{\mu})}{\lambda_{\mu+1}(t_{k}^{\mu})}=\frac{f(t_{k}^{\mu},\bar{t}_{k}^{\mu})}{f(\bar{t}_{k}^{\mu},t_{k}^{\mu})}\frac{f(\bar{t}^{\mu+1},t_{k}^{\mu})}{f(t_{k}^{\mu},\bar{t}^{\mu-1})},
\end{equation}
where we introduced the following shorthand notations
\begin{equation}
\begin{split}f(u,v) & =\frac{u-v+1}{u-v},\\
f(u,\bar{t}^{i}) & =\prod_{k=1}^{r_{i}}f(u,t_{k}^{i}),\quad f(\bar{t}^{i},u)=\prod_{k=1}^{r_{i}}f(t_{k}^{i},u),\quad f(\bar{t}^{i},\bar{t}^{j})=\prod_{k=1}^{r_{i}}f(t_{k}^{i},\bar{t}^{j}).
\end{split}
\end{equation}
The transfer matrix eigenvalues are \cite{Liashyk:2018egk}
\begin{align}
\tau(u|\bar{t}) & =\sum_{i=1}^{N}\lambda_{i}(u)f(\bar{t}^{i},u)f(u,\bar{t}^{i-1}),\label{eq:eig}\\
\hat{\tau}(u|\bar{t}) & =\sum_{i=1}^{N}\hat{\lambda}_{i}(u)f(\bar{t}^{i-1}+(i-1),u)f(u,\bar{t}^{i}+i).
\end{align}
We call the Bethe vector on-shell if the Bethe roots satisfy the Bethe
equations (when the Bethe vector is an eigenvector), and off-shell
for generic Bethe roots. In the literature recursive definitions are
available for the off-shell Bethe vectors. For the overlap calculations
we also have to know how the monodromy matrix entries act on the off-shell
Bethe vectors. Fortunately it was also previously derived in \cite{Hutsalyuk:2020dlw}. 

One can also define the left eigenvectors $\mathbb{C}(\bar{t})$ of
the transfer matrix and the square of the norm of the on-shell Bethe
states satisfies the Gaudin hypothesis \cite{Hutsalyuk:2017way}
\begin{equation}
\mathbb{C}(\bar{t})\mathbb{B}(\bar{t})=\frac{\prod_{\nu=1}^{N-1}\prod_{k\neq l}f(t_{l}^{\nu},t_{k}^{\nu})}{\prod_{\nu=1}^{N-2}f(\bar{t}^{\nu+1},\bar{t}^{\nu})}\det G,\label{eq:norm}
\end{equation}
where $G$ is the Gaudin matrix given by
\begin{equation}
G_{j,k}^{(\mu,\nu)}=-\frac{\partial\log\Phi_{j}^{(\mu)}}{\partial t_{k}^{\nu}}.
\end{equation}

\section{Integrable matrix product states\label{sec:Integrable-matrix-product}}

In this section, we generalize the previously introduced $KT$-relation
to MPSs. This essentially means that certain quantities which were
previously scalars will now become matrices in the so-called boundary
space.

\subsection{KT-relations}

We generalize the previously defined $KT$-relation \cite{Gombor:2021hmj,Gombor:2023bez}
with boundary degrees of freedom as
\begin{equation}
\sum_{k=1}^{N}\sum_{c=1}^{d_{B}}K_{i,k}^{a,c}(z)\langle\psi_{c,b}|T_{k,j}(z)=\sum_{k=1}^{N}\sum_{c=1}^{d_{B}}\langle\psi_{a,c}|\bar{T}_{i,k}(-z)K_{k,j}^{c,b}(z).
\end{equation}
The $\langle\psi_{a,b}|$ are elements of the dual of the quantum
space $\mathcal{H}$, i.e., $\langle\psi_{a,b}|\in\mathcal{H}^{*}$
for $a,b=1,\dots,d_{B}$. The $KT$-relations have two types: for
the \textbf{uncrossed} $KT$-relation $\bar{T}\equiv T$ and for the
\textbf{crossed} $KT$-relation $\bar{T}\equiv\widehat{T}$. 

We can introduce a boundary vector space $\mathcal{H}_{B}=\mathbb{C}^{d_{B}}$
and we can collect the states and the coefficients to matrices in
the boundary space 
\begin{equation}
\begin{split}\langle\Psi| & :=\sum_{a,b=1}^{d_{B}}\langle\psi_{a,b}|\otimes e_{a,b}^{B}\in\mathcal{H}^{*}\otimes\mathrm{End}(\mathcal{H}_{B}),\\
\mathbf{K}_{i,j}(z) & :=\sum_{a,b=1}^{d_{B}}K_{i,j}^{a,b}(z)e_{a,b}^{B}\in\mathrm{End}(\mathcal{H}_{B}),
\end{split}
\end{equation}
where $e_{a,b}^{B}\in\mathrm{End}(\mathcal{H}_{B})$ are the unit
matrices of the boundary space for $a,b=1,\dots,d_{B}$. Using these
notations the $KT$-relation simplifies as
\begin{equation}
\mathbf{K}_{0}(z)\langle\Psi|T_{0}(z)=\langle\Psi|\bar{T}_{0}(-z)\mathbf{K}_{0}(z),
\end{equation}
where $T(z),\bar{T}(z)\in\mathrm{End}(\mathbb{C}^{N})\otimes\mathrm{End}(\mathcal{H})$
are the monodromy matrices, the $\mathbf{K}(z)=\sum_{i,j=1}^{N}e_{i,j}\otimes\in\mathbf{K}_{i,j}(z)\in\mathrm{End}(\mathbb{C}^{N})\otimes\mathrm{End}(\mathcal{H}_{B})$
is the $K$-matrix and $\langle\Psi|\in\mathcal{H}^{*}\otimes\mathrm{End}(\mathcal{H}_{B})$
is the boundary state.

We also define the matrix product state $\langle\mathrm{MPS}|\in\mathcal{H}^{*}$
as
\begin{equation}
\langle\mathrm{MPS}|=\sum_{\ell=1}^{d_{B}}\langle\psi_{\ell,\ell}|=\mathrm{Tr}_{\mathcal{H}_{B}}\left(\langle\Psi|\right).
\end{equation}
Our goal is to calculate on-shell overlaps $\langle\mathrm{MPS}|\mathbb{B}(\bar{t})$.

\subsection{Reflection algebras}

The compatibility of the $KT$-relation with the $RTT$-relation (we
do the calculation $\langle\Psi|T_{1}(u)T_{2}(v)\to\langle\Psi|\bar{T}_{1}(-u)\bar{T}_{2}(-v)$
in two different orders) requires the equation \cite{Gombor:2021hmj}
\begin{equation}
R_{1,2}(v-u)\mathbf{K}_{1}(u)\bar{R}_{1,2}(-u-v)\mathbf{K}_{2}(v)=\mathbf{K}_{2}(v)\bar{R}_{1,2}(-u-v)\mathbf{K}_{1}(u)R_{1,2}(v-u).\label{eq:refl}
\end{equation}
where we used the notation
\begin{equation}
\bar{R}(u)\equiv\begin{cases}
R(u), & \text{for the uncrossed case},\\
\widehat{R}(u), & \text{for the crossed case}.
\end{cases}
\end{equation}
The $KT$-relation and the reflection equation is invariant under
the renormalization of the $K$-matrix $\mathbf{K}_{1}(u)\to f(u)\mathbf{K}_{1}(u)$
for any scalar function $f(u)$.

In the crossed case the reflection equation with components is
\begin{align}
\left[\mathbf{K}_{i,j}(u),\mathbf{K}_{k,l}(v)\right] & =\frac{1}{u-v}\left(\mathbf{K}_{k,j}(u)\mathbf{K}_{i,l}(v)-\mathbf{K}_{k,j}(v)\mathbf{K}_{i,l}(u)\right)\nonumber \\
 & -\frac{1}{u+v}\left(\mathbf{K}_{i,k}(u)\mathbf{K}_{j,l}(v)-\mathbf{K}_{k,i}(v)\mathbf{K}_{l,j}(u)\right)\label{eq:crBY}\\
 & +\frac{1}{u^{2}-v^{2}}\left(\mathbf{K}_{k,i}(u)\mathbf{K}_{j,l}(v)-\mathbf{K}_{k,i}(v)\mathbf{K}_{j,l}(u)\right).\nonumber 
\end{align}
With a proper normalization, the $K$-matrix has the asymptotic limit
\begin{equation}
\mathbf{K}_{i,j}(u)=\mathbf{A}_{i,j}+\mathcal{O}(u^{-1}).
\end{equation}
Taking the $u\to\infty$ limit of the reflection equation becomes
\begin{equation}
\left[\mathbf{A}_{i,j},\mathbf{K}_{k,l}(v)\right]=0.
\end{equation}
We also assume that the $K$-matrix is irreducible, i.e. the set of
matrices in the boundary space $\{\mathbf{K}_{i,j}(u)\}_{i,j=1}^{N}$
do not have non-trivial invariant subspace. It leads to $\mathbf{A}_{i,j}=\mathcal{U}_{i,j}\mathbf{1}$,
where $\mathcal{U}_{i,j}\in\mathbb{C}$ for $i,j=1,\dots,N$. Multiplying
the reflection equation with $(u+v)$ and taking the $v=-u$ limit,
we obtain 
\begin{equation}
\begin{split}0 & =-\left(\mathbf{K}_{i,k}(u)\mathbf{K}_{j,l}(-u)-\mathbf{K}_{k,i}(-u)\mathbf{K}_{l,j}(u)\right)\\
 & +\frac{1}{2u}\left(\mathbf{K}_{k,i}(u)\mathbf{K}_{j,l}(-u)-\mathbf{K}_{k,i}(-u)\mathbf{K}_{j,l}(u)\right).
\end{split}
\end{equation}
Taking the $u\to\infty$ limit
\begin{equation}
\mathcal{U}_{i,k}\mathcal{U}_{j,l}=\mathcal{U}_{k,i}\mathcal{U}_{l,j},
\end{equation}
therefore $\mathcal{U}_{i,k}=\pm\mathcal{U}_{k,i}$. In summary, the
asymptotic limit of the $K$-matrix is 
\begin{equation}
\mathbf{K}_{i,j}(u)=\mathcal{U}_{i,j}\mathbf{1}+\mathcal{O}(u^{-1}).\label{eq:seriestwY}
\end{equation}
The reflection equation with the series expansion (\ref{eq:seriestwY})
defines the twisted Yangian algebras $Y^{\pm}(N)$ ($Y^{+}(N)$ for
symmetric and $Y^{-}(N)$ for anti-symmetric $\mathcal{U}$). \footnote{In the literature (e.g., \cite{MolevBook}), these algebras are often
referred to as extended twisted Yangians, and the definitions of twisted
Yangians include an additional symmetry property. This extra condition
partially fixes the normalization of the $K$-matrix, meaning that
the symmetry equation is not invariant under the rescaling $\mathbf{K}(u)\to\mu(u)\mathbf{K}(u)$
for an arbitrary function $\mu(u)$. However, our $KT$-relation is
invariant under such rescaling, and we do not use the symmetry property
in our calculations. Therefore, in this paper, we work with the extended
versions, and for simplicity, we omit the term \textquotedblleft extended.\textquotedblright{}}

For the uncrossed case the explicit form of the reflection equation
is
\begin{align}
\left[\mathbf{K}_{i,j}(u),\mathbf{K}_{k,l}(v)\right] & =\frac{1}{u-v}\left(\mathbf{K}_{k,j}(u)\mathbf{K}_{i,l}(v)-\mathbf{K}_{k,j}(v)\mathbf{K}_{i,l}(u)\right)\nonumber \\
 & +\frac{1}{u+v}\sum_{n=1}^{N}\left(\delta_{j,k}\mathbf{K}_{i,n}(u)\mathbf{K}_{n,l}(v)-\delta_{i,l}\mathbf{K}_{k,n}(v)\mathbf{K}_{n,j}(u)\right)\\
 & -\frac{1}{u^{2}-v^{2}}\delta_{i,j}\sum_{n=1}^{N}\left(\mathbf{K}_{k,n}(u)\mathbf{K}_{n,l}(v)-\mathbf{K}_{k,n}(v)\mathbf{K}_{n,l}(u)\right).\nonumber 
\end{align}
Taking the $u\to\infty$ limit of the reflection equation we can easily
show that the the asymptotic limit of the $K$-matrix is 
\begin{equation}
\mathbf{K}_{i,j}(u)=\mathcal{U}_{i,j}\mathbf{1}+\mathcal{O}(u^{-1}).
\end{equation}
Taking first the $v=-u$ limit and later the $u\to\infty$ limit we
obtain that
\begin{equation}
\sum_{n=1}^{N}\left(\delta_{j,k}\mathcal{U}_{i,n}\mathcal{U}_{n,l}-\delta_{i,l}\mathcal{U}_{k,n}\mathcal{U}_{n,j}\right)=0.
\end{equation}
The solution of this equation is 
\begin{equation}
\sum_{n=1}^{N}\mathcal{U}_{i,n}\mathcal{U}_{n,l}=a\delta_{i,l},\label{eq:UUa}
\end{equation}
where $a\in\mathbb{C}$. For $a=0$ we call the $K$-matrix singular.
For non-singular $K$-matrices we can always choose the normalization
as $a=1$. The possible eigenvalues of the matrix $\mathcal{U}$ are
$\pm1$. Let the number of $-1$ be $M$. The reflection equation
and the asymptotic expansion with $M$ number of $-1$ eigenvalues
defines the reflection algebra $\mathcal{B}(N,M)$. \footnote{Similar to twisted Yangians, these algebras are commonly referred
to as extended reflection algebras \cite{MOLEV:2002}, and the definition
of reflection algebras includes an additional unitarity condition.
However, since we do not use this condition, we also omit the term
\textquotedblleft extended\textquotedblright{} here for the sake of
simplicity.} 

\subsection{Matrix product states from co-product}

Let us introduce a tensor product quantum space as $\mathcal{H}=\mathcal{H}^{(1)}\otimes\mathcal{H}^{(2)}$.
The corresponding monodromy matrices are $\bar{T}_{i,j}^{(1)}(u)\in\mathrm{End}(\mathcal{H}^{(1)}),\bar{T}_{i,j}^{(2)}(u)\in\mathrm{End}(\mathcal{H}^{(2)})$
and
\begin{equation}
\bar{T}_{0}(u)=\bar{T}_{0}^{(2)}(u)\bar{T}_{0}^{(1)}(u)\in\mathrm{End}(\mathbb{C}^{N})\otimes\mathrm{End}(\mathcal{H}),\label{eq:comon}
\end{equation}
Let $\left\langle \Psi^{(1)}\right|\in\mathcal{H}^{(1)}\otimes\mathcal{H}_{B}$
and $\left\langle \Psi^{(2)}\right|\in\mathcal{H}^{(2)}\otimes\mathcal{H}_{B}$
be boundary states with the same $K$-matrix i.e. they satisfy the
same $KT$-relation
\begin{equation}
\mathbf{K}_{0}(u)\langle\Psi^{(i)}|T_{0}^{(i)}(u)=\langle\Psi^{(i)}|\bar{T}_{0}^{(i)}(-u)\mathbf{K}_{0}(u),
\end{equation}
for $i=1,2$. 
\begin{prop}
\label{lem:co-prod}The co-vector 
\begin{equation}
\langle\Psi|=\langle\Psi^{(2)}|\langle\Psi^{(1)}|\label{eq:Psi_Factor}
\end{equation}
is a boundary state in the tensor product quantum space $\mathcal{H}$
with the same $K$-matrix i.e. it satisfies the $KT$-relation
\begin{equation}
\mathbf{K}_{0}(u)\langle\Psi|T_{0}(u)=\langle\Psi|\bar{T}_{0}(-u)\mathbf{K}_{0}(u).\label{eq:coKT}
\end{equation}
\end{prop}
\begin{proof}
Let us start with the lhs
\begin{equation}
\mathbf{K}_{0}(u)\langle\Psi|T_{0}(u)=\mathbf{K}_{0}(u)\left(\langle\Psi^{(2)}|\langle\Psi^{(1)}|\right)\left(T_{0}^{(2)}(u)T_{0}^{(1)}(u)\right)=\mathbf{K}_{0}(u)\left(\langle\Psi^{(2)}|T_{0}^{(2)}(u)\right)\left(\langle\Psi^{(1)}|T_{0}^{(1)}(u)\right),
\end{equation}
where we used the definition (\ref{eq:Psi_Factor}) and the co-product
property of the monodromy matrix (\ref{eq:comon}). Now we can use
the $KT$-relation of space $\mathcal{H}^{(2)}:$
\begin{equation}
\mathbf{K}_{0}(u)\langle\Psi|T_{0}(u)=\left(\langle\Psi^{(2)}|\bar{T}_{0}^{(2)}(-u)\right)\mathbf{K}_{0}(u)\left(\langle\Psi^{(1)}|T_{0}^{(1)}(u)\right).
\end{equation}
Now on $\mathcal{H}^{(1)}:$
\begin{equation}
\begin{split}\mathbf{K}_{0}(u)\langle\Psi|T_{0}(u) & =\left(\langle\Psi^{(2)}|\bar{T}_{0}^{(2)}(-u)\right)\left(\langle\Psi^{(1)}|\bar{T}_{0}^{(1)}(-u)\right)\mathbf{K}_{0}(u)\\
 & =\left(\langle\Psi^{(2)}|\langle\Psi^{(1)}|\right)\left(\bar{T}_{0}^{(2)}(-u)\bar{T}_{0}^{(1)}(-u)\right)\mathbf{K}_{0}(u).
\end{split}
\end{equation}
Using (\ref{eq:comon}) we just proved (\ref{eq:coKT}).
\end{proof}
Using the co-product property, we can construct integrable MPSs of
arbitrary length from the \textquotedblleft elementary\textquotedblright{}
solutions of the $KT$-equation. In the following, we present an example
where only the defining representations are used. We begin with the
reflection equation (\ref{eq:refl}). After change of variables it
is equivalent to
\begin{equation}
\mathbf{K}_{0}(u)\bar{R}_{0,1}(-u-\theta)\mathbf{K}_{1}(\theta)R_{0,1}(u-\theta)=R_{0,1}(u-\theta)\mathbf{K}_{1}(\theta)\bar{R}_{0,1}(-u-\theta)\mathbf{K}_{0}(u).
\end{equation}
We can use an equivalent form 

\begin{equation}
\mathbf{K}_{0}(u)\psi_{2,1}(\theta)\left[\bar{R}_{0,2}^{t_{2}}(-u-\theta)R_{0,1}(u-\theta)\right]=\psi_{2,1}(\theta)\left[R_{0,2}^{t_{2}}(u-\theta)\bar{R}_{0,1}(-u-\theta)\right]\mathbf{K}_{0}(u),
\end{equation}
where
\begin{equation}
\psi_{1,2}(\theta)=\sum_{i,j=1}^{N}\langle i,j|\otimes\mathbf{K}_{i,j}(\theta).
\end{equation}
The state $\psi_{2,1}(\theta)$ satisfy the $KT$-relation for length
two spin chain with monodromy matrices
\begin{equation}
T_{0}(u)=\bar{R}_{0,2}^{t_{2}}(-u-\theta)R_{0,1}(u-\theta),\quad\bar{T}_{0}(u)=R_{0,2}^{t_{2}}(-u-\theta)\bar{R}_{0,1}(u-\theta).
\end{equation}
Using the co-product property we can built boundary state for length
$2J$ as
\begin{equation}
\begin{split}\langle\Psi| & =\psi_{2J,2J-1}(\theta_{J})\dots\psi_{4,3}(\theta_{2})\psi_{2,1}(\theta_{1}),\\
T_{0}(u) & =\bar{R}_{0,2J}^{t_{2J}}(-u-\theta_{J})R_{0,2J-1}(u-\theta_{J})\dots\bar{R}_{0,2}^{t_{2}}(-u-\theta_{1})R_{0,1}(u-\theta_{1}).
\end{split}
\end{equation}
For the even sites we have
\begin{equation}
\bar{R}_{0,2j}^{t_{2j}}(-u-\theta_{j})=\begin{cases}
R_{0,2j}(u+\theta_{j}), & \text{for the crossed case,}\\
\widehat{R}_{0,2j}(u+\theta_{j}), & \text{for the uncrossed case,}
\end{cases}
\end{equation}
Using fusion, we can generalize the boundary states to any finite-dimensional
representation of $\mathfrak{gl}_{N}$. The details are provided in
Appendix \ref{sec:Boundary-states-gen-rep}. For a finite-dimensional
representation $\Lambda$, there is a two-site state $\psi_{2,1}^{\Lambda}(\theta)$.
Using the co-product property, we can construct a boundary state
\begin{equation}
\langle\Psi|=\psi_{2J,2J-1}^{\Lambda^{(J)}}(\theta_{J})\dots\psi_{4,3}^{\Lambda^{(2)}}(\theta_{2})\psi_{2,1}^{\Lambda^{(1)}}(\theta_{1}).
\end{equation}
This satisfies the $KT$-relation on a spin chain where the monodromy
matrix is defined as
\begin{equation}
T_{0}(u)=\left[\bar{L}_{0,2J}^{\Lambda^{(J)}}(-u-\theta_{J})\right]^{t_{2J}}L_{0,2J-1}^{\Lambda^{(J)}}(u-\theta_{J})\dots\left[\bar{L}_{0,2}^{\Lambda^{(1)}}(-u-\theta_{1})\right]^{t_{2}}L_{0,1}^{\Lambda^{(1)}}(u-\theta_{1}).
\end{equation}
The pseudo-vacuum eigenvalue is
\begin{equation}
\lambda_{k}(u)=\begin{cases}
\prod_{j=1}^{J}\lambda_{k}^{(j)}(u-\theta_{j})\hat{\lambda}_{k}^{(j)}(-u-\theta_{j}), & \text{for the crossed case,}\\
\prod_{j=1}^{J}\lambda_{k}^{(j)}(u-\theta_{j})\lambda_{N+1-k}^{(j)}(-u-\theta_{j}), & \text{for the uncrossed case,}
\end{cases}
\end{equation}
where
\begin{equation}
\lambda_{k}^{(j)}(u)=\frac{u+\Lambda_{k}^{(j)}}{u}.
\end{equation}
The pseudo-vacuum eigenvalues satisfy the symmetry conditions
\begin{equation}
\lambda_{k}(u)=\begin{cases}
\hat{\lambda}_{k}(-u), & \text{for the crossed case,}\\
\lambda_{N+1-k}(-u), & \text{for the uncrossed case.}
\end{cases}\label{eq:lamProp}
\end{equation}
 The symmetry properties of the $\alpha$-s are
\begin{equation}
\alpha_{k}(u)=\frac{\lambda_{k}(u)}{\lambda_{k+1}(u)}=\begin{cases}
\frac{\hat{\lambda}_{k}(-u)}{\hat{\lambda}_{k+1}(-u)}=\hat{\alpha}_{k}(-u)=\frac{1}{\alpha_{k}(-u-k)}, & \text{for the crossed case,}\\
\frac{\lambda_{N+1-k}(-u)}{\lambda_{N-k}(-u)}=\frac{1}{\alpha_{N-k}(-u)}, & \text{for the uncrossed case.}
\end{cases}\label{eq:symProp}
\end{equation}

\subsection{Integrability and pair structure}

Assuming that the $K$-matrices are invertible, we can product the
$KT$-relation with the inverse: 
\begin{equation}
\langle\Psi|T_{0}(z)=\mathbf{K}_{0}(z)^{-1}\langle\Psi|\bar{T}_{0}(-z)\mathbf{K}_{0}(z).
\end{equation}
Taking the trace in the auxiliary and the boundary space we obtain
that
\begin{equation}
\langle\mathrm{MPS}|\mathcal{T}(z)=\langle\mathrm{MPS}|\bar{\mathcal{T}}(-z).
\end{equation}
For on-shell states we have
\begin{equation}
\left(\tau(u|\bar{t})-\bar{\tau}(-u|\bar{t})\right)\langle\mathrm{MPS}|\mathbb{B}(\bar{t})=0.
\end{equation}
The non-vanishing overlap $\langle\mathrm{MPS}|\mathbb{B}(\bar{t})\neq0$
requires that
\begin{equation}
\tau(u|\bar{t})=\bar{\tau}(-u|\bar{t}).\label{eq:intcondtau-1}
\end{equation}

For the uncrossed case, the eigenvalue $\tau(-u|\bar{t})$ can be
written as
\begin{equation}
\begin{split}\tau(-u|\bar{t}) & =\sum_{i=1}^{N}\lambda_{i}(-u)f(\bar{t}^{i},-u)f(-u,\bar{t}^{i-1})=\\
 & =\sum_{i=1}^{N}\lambda_{i}(u)f(-\bar{t}^{N-i},u)f(u,-\bar{t}^{N+1-i})=\tau(u|\pi^{a}(\bar{t})),
\end{split}
\end{equation}
where we used the symmetry property (\ref{eq:lamProp}) and introduced
the following map of Bethe roots:
\begin{equation}
\pi^{a}(\bar{t}^{1},\bar{t}^{2},\dots,\bar{t}^{N-1})=\left(-\bar{t}^{N-1},-\bar{t}^{N-2},\dots,-\bar{t}^{1}\right).\label{eq:chpair-1}
\end{equation}

For the crossed case, the eigenvalue $\hat{\tau}(-u|\bar{t})$ can
be written as
\begin{equation}
\begin{split}\hat{\tau}(-u|\bar{t}) & =\sum_{i=1}^{N}\hat{\lambda}_{i}(-u)f(\bar{t}^{i-1}+(i-1),-u)f(-u,\bar{t}^{i}+i)=\\
 & =\sum_{i=1}^{N}\lambda_{i}(u)f(-\bar{t}^{i}-i,u)f(u,-\bar{t}^{i-1}-(i-1))=\tau(u|\pi^{c}(\bar{t})),
\end{split}
\end{equation}
where we used that the symmetry property (\ref{eq:lamProp}) and introduced
the following map of Bethe roots:
\begin{equation}
\pi^{c}(\bar{t}^{1},\bar{t}^{2},\dots,\bar{t}^{N-1})=\left(-\bar{t}^{1}-1,-\bar{t}^{2}-2,\dots,-\bar{t}^{N-1}-(N-1)\right).\label{eq:chpair}
\end{equation}
Substituting back to (\ref{eq:intcondtau-1}) the condition for non-vanishing
overlaps read as
\begin{equation}
\tau(u|\bar{t})=\tau(u|\pi^{a/c}(\bar{t})).
\end{equation}
Notice that the set $\pi^{a/c}(\bar{t})$ also satisfies the Bethe
equations i.e. the vector $\mathbb{B}(\pi^{a/c}(\bar{t}))$ is also
an on-shell Bethe vector. Since different Bethe vectors have different
eigenvalues, we just obtain that the Bethe roots have to satisfy the
selection rule
\begin{equation}
\bar{t}=\pi^{a/c}(\bar{t}),\label{eq:ps}
\end{equation}
for the non-vanishing overlaps. We call the Bethe roots with condition
(\ref{eq:ps}) as Bethe roots with \emph{pair structure}. 

For \emph{achiral pair structure}, i.e. when $\bar{t}=\pi^{a}(\bar{t})$,
every sets $\bar{t}^{\nu}$ satisfy the condition $\bar{t}^{\nu}=-\bar{t}^{N-\nu}$.
For \emph{chiral pair structure}, i.e. when $\bar{t}=\pi^{c}(\bar{t})$,
every sets $\bar{t}^{\nu}$ satisfy the condition $\bar{t}^{\nu}=-\bar{t}^{\nu}-\nu$.

\subsection{Recursion for off-shell overlaps}

Following \cite{Gombor:2021hmj} we can obtain a method for the evaluation
of off-shell overlaps. We use the recurrence relation and action formulas
for the Bethe states \cite{Hutsalyuk:2017tcx,Hutsalyuk:2020dlw}.
See the details in appendix \ref{sec:Reqursions-for-BS}. In the following
we just sketch the recursion for the overlaps and we concentrate on
the number of Bethe roots and we drop the numerical coefficients which
are irrelevant for this purpose. The recurrence equation of the Bethe
vectors reads as 
\begin{equation}
\mathbb{B}^{r_{1},\dots,r_{N-1}}=\sum_{k=2}^{N}T_{1,k}\sum\mathbb{B}^{r_{1}-1,\dots,,r_{k-1}-1,r_{k},\dots,r_{N-1}}(\dots),
\end{equation}
and the action formulas are
\begin{equation}
T_{i,j}\mathbb{B}^{r_{1},r_{2},\dots,r_{N-1}}=\begin{cases}
\sum(\dots)\mathbb{B}^{r_{1},\dots r_{i}+1,r_{i+1}+1,\dots,r_{j-1}+1,r_{j},\dots r_{N-1}}, & i\leq j,\\
\sum(\dots)\mathbb{B}^{r_{1},\dots r_{j}-1,r_{j+1}-1,\dots,r_{i-1}-1,r_{i},\dots r_{N-1}}, & i>j,
\end{cases}\label{eq:actionT}
\end{equation}
and
\begin{equation}
\begin{split}\widehat{T}_{i,j}\mathbb{B}^{r_{1},r_{2},\dots,r_{N-1}} & =\begin{cases}
\sum(\dots)\mathbb{B}^{r_{1},\dots r_{j}+1,r_{j+1}+1,\dots,r_{i-1}+1,r_{i},\dots r_{N-1}}, & j\leq i,\\
\sum(\dots)\mathbb{B}^{r_{1},\dots r_{i}-1,r_{i+1}-1,\dots,r_{j-1}-1,r_{j},\dots r_{N-1}}, & j>i.
\end{cases}\end{split}
\label{eq:actionTh}
\end{equation}
We can see that the diagonal elements $T_{i,i}$ and $\widehat{T}_{i,i}$
do not change the number of Bethe roots. For $i>j$, $T_{i,j}$ and
$\widehat{T}_{j,i}$ decrease the number of certain roots by one,
while $T_{j,i}$ and $\widehat{T}_{i,j}$ increase the number of those
same roots by one.

\subsubsection{The crossed overlaps}

We use the $(1,k)$ component of the crossed $KT$-relation
\begin{equation}
\left\langle \Psi\right|T_{1,k}(u)=\sum_{j=1}^{N}\mathbf{K}_{1,1}^{-1}(u)\left\langle \Psi\right|\widehat{T}_{1,j}(-u)\mathbf{K}_{j,k}(u)-\sum_{i=2}^{N}\mathbf{K}_{1,1}^{-1}(u)\mathbf{K}_{1,i}(u)\left\langle \Psi\right|T_{i,k}(u),\label{eq:crossrec}
\end{equation}
which can be used to change the creation operators $T_{1,k}$ to the
operators $\widehat{T}_{1,j}$ and $T_{i,k}$ where $i,k=2,\dots,N$
and $j=1,\dots,N$. For $j>1$, the operators $\widehat{T}_{1,j}$
decrease the number of the first Bethe roots by one and the operators
$\widehat{T}_{1,1}(-u)$ and $T_{i,k}(u)$ do not change it for $i,k\geq2$.
Using the recurrence relation (\ref{eq:crossrec}) and the action
formulas (\ref{eq:actionT}), (\ref{eq:actionTh}) we obtain a recursion
\begin{equation}
\left\langle \Psi\right|\mathbb{B}^{r_{1},\dots}=\sum(\dots)\left\langle \Psi\right|\mathbb{B}^{r_{1}-1,\dots}+\sum(\dots)\left\langle \Psi\right|\mathbb{B}^{r_{1}-2,\dots}.
\end{equation}
When we changed the creation operators to the diagonal and the annihilation
ones we assumed that the inverse $\mathbf{K}_{1,1}^{-1}(u)$ exists
which is possible only for $Y^{+}(N)$ $K$-matrices but it is never
true for $Y^{-}(N)$. This can be seen from the asymptotic expansion
(\ref{eq:seriestwY}). In the symmetric case, it is possible that
$\mathcal{U}_{1,1}\neq0$, but in the anti-symmetric case $\mathcal{U}_{1,1}=0$.
If $\mathcal{U}_{1,1}\neq0$ then the operator $\mathbf{K}_{1,1}^{-1}(u)$
can be computed order by order in an expansion in powers of $u^{-1}$.
In the following we concentrate on $Y^{+}(N)$ and deal the other
case $Y^{-}(N)$ later. 

Repeating these steps we can express the general overlap with a sum
of overlaps without the first Bethe roots
\begin{equation}
\left\langle \Psi\right|\mathbb{B}^{r_{1},\dots}=\sum(\dots)\left\langle \Psi\right|\mathbb{B}^{0,\dots}.
\end{equation}
We can see that we reduced the original $\mathfrak{gl}(N)$ Bethe
state to $\mathfrak{gl}(N-1)$ ones. The action of operators $\left\{ T_{a,b}\right\} _{a,b=2}^{N}$
and $\left\{ \widehat{T}_{a,b}\right\} _{a,b=2}^{N}$ do not lead
out of the subspace generated by the Bethe vectors $\mathbb{B}^{0,\dots}$
. Naively, the crossed $KT$-relation is not closed for the indexes
$a,b=2,\dots,N$ since
\begin{equation}
\sum_{c=2}^{N}\mathbf{K}_{a,c}\left\langle \Psi\right|T_{c,b}+\mathbf{K}_{a,1}\left\langle \Psi\right|T_{1,b}=\sum_{c=2}^{N}\left\langle \Psi\right|\widehat{T}_{a,c}\mathbf{K}_{c,b}+\left\langle \Psi\right|\widehat{T}_{a,1}\mathbf{K}_{1,b}.\label{eq:nestin1}
\end{equation}
We can see that the operators $T_{1,b}$ and $\widehat{T}_{a,1}$
create the first type of Bethe roots therefore we have to eliminate
them from the equation (\ref{eq:nestin1}). Let us get the components
$(1,b)$ , $(a,1)$ and $(1,1)$ of the crossed $KT$-relation 
\begin{align}
 & \sum_{c=2}^{N}\mathbf{K}_{1,c}\left\langle \Psi\right|T_{c,b}+\mathbf{K}_{1,1}\left\langle \Psi\right|T_{1,b}\cong\left\langle \Psi\right|\widehat{T}_{1,1}\mathbf{K}_{1,b},\label{eq:nestin2}\\
 & \mathbf{K}_{a,1}\left\langle \Psi\right|T_{1,1}\cong\sum_{c=2}^{N}\left\langle \Psi\right|\widehat{T}_{a,c}\mathbf{K}_{c,1}+\left\langle \Psi\right|\widehat{T}_{a,1}\mathbf{K}_{1,1},\label{eq:nestin3}\\
 & \mathbf{K}_{1,1}\left\langle \Psi\right|T_{1,1}\cong\left\langle \Psi\right|\widehat{T}_{1,1}\mathbf{K}_{1,1},\label{eq:nestin4}
\end{align}
where $\cong$ denotes the equality in the subspace generated by the
Bethe vectors $\mathbb{B}^{0,\dots}$. Combining the equations (\ref{eq:nestin1}),
(\ref{eq:nestin2}), (\ref{eq:nestin3}) and (\ref{eq:nestin4}) we
obtain a new crossed $KT$-relation on the subspace generated by the
Bethe states $\mathbb{B}^{0,\dots}$
\begin{equation}
\sum_{c=2}^{N}\mathbf{K}_{a,c}^{(2)}\left\langle \Psi\right|T_{c,b}\cong\sum_{c=2}^{N}\left\langle \Psi\right|\widehat{T}_{a,c}\mathbf{K}_{c,b}^{(2)},
\end{equation}
where
\begin{equation}
\mathbf{K}_{a,b}^{(2)}(z)=\mathbf{K}_{a,b}(z)-\mathbf{K}_{a,1}(z)\mathbf{K}_{1,1}^{-1}(z)\mathbf{K}_{1,b}(z),
\end{equation}
for $a,b=2,\dots,N$. Repeating the previous method we can eliminate
the sets of Bethe roots $\bar{t}^{\nu}$, therefore we obtain a recursion
for the overlaps and in the end of the day we obtain an explicit formula
for the overlap. In the $(k+1)$-th step of the nesting we have a
$\mathfrak{gl}(N-k)$ $K$-matrix:
\begin{equation}
\mathbf{K}_{a,b}^{(k+1)}(z)=\mathbf{K}_{a,b}^{(k)}(z)-\mathbf{K}_{a,k}^{(k)}(z)\left[\mathbf{K}_{k,k}^{(k)}(z)\right]^{-1}\mathbf{K}_{k,b}^{(k)}(z),\label{eq:nestedKT}
\end{equation}
which satisfy the $KT$-relations
\begin{equation}
\sum_{c=k}^{N}\mathbf{K}_{a,c}^{(k)}(u)\left\langle \Psi\right|T_{c,b}(u)\cong\sum_{c=k}^{N}\left\langle \Psi\right|\widehat{T}_{a,c}(-u)\mathbf{K}_{c,b}^{(k)}(u),
\end{equation}
for $a,b=k,\dots,N$ where $\cong$ denotes the equality in the subspace
generated by the Bethe vectors $\mathbb{B}^{0,\dots,0,r_{k},\dots,r_{N-1}}$. 

\subsubsection{The uncrossed overlaps\label{subsec:The-uncrossed-overlaps}}

We can use similar recursion for the uncrossed overlaps. Now we use
the $(N,k)$ component of the uncrossed $KT$-relation
\begin{equation}
\left\langle \Psi\right|T_{1,k}(u)=\sum_{j=1}^{N}\mathbf{K}_{N,1}^{-1}(u)\left\langle \Psi\right|T_{N,j}(-u)\mathbf{K}_{j,k}(u)-\sum_{i=2}^{N}\mathbf{K}_{N,1}^{-1}(u)\mathbf{K}_{N,i}(u)\left\langle \Psi\right|T_{i,k}(u).
\end{equation}
Analogous way, we can use this equation to decrease the numbers of
the first Bethe roots while the cardinality of $\bar{t}^{N-1}$ is
not increased. Using the action formulas we obtain a recursion
\begin{equation}
\left\langle \Psi\right|\mathbb{B}^{r_{1},\dots,r_{N-1}}=\sum(\dots)\mathbb{B}^{\tilde{r}_{1},\dots,\tilde{r}_{N-1}}.
\end{equation}
where $r_{1}>\tilde{r}_{1}$ and $r_{N-1}\geq\tilde{r}_{N-1}$. We
can also use the $(k,1)$ component of the uncrossed $KT$-relation
\begin{equation}
\left\langle \Psi\right|T_{k,N}(-u)=\sum_{i=1}^{N}\mathbf{K}_{k,i}(u)\left\langle \Psi\right|T_{i,1}(u)\mathbf{K}_{N,1}^{-1}(u)-\sum_{j=1}^{N-1}\left\langle \Psi\right|T_{k,j}(-u)\mathbf{K}_{j,1}(u)\mathbf{K}_{N,1}^{-1}(u).
\end{equation}
Analogous way, we can use this equation to decrease the cardinality
of $\bar{t}^{N-1}$ while the cardinality of $\bar{t}^{1}$ is not
increased. Using the action formulas we obtain a recursion
\begin{equation}
\left\langle \Psi\right|\mathbb{B}^{r_{1},\dots,r_{N-1}}=\sum(\dots)\mathbb{B}^{\tilde{r}_{1},\dots,\tilde{r}_{N-1}}.
\end{equation}
where $r_{1}\geq\tilde{r}_{1}$ and $r_{N-1}>\tilde{r}_{N-1}$. 

Repeating these steps we can express the general overlap with a sum
of overlaps without the first and last types of Bethe roots
\begin{equation}
\left\langle \Psi\right|\mathbb{B}^{r_{1},\dots,r_{N-1}}=\sum(\dots)\left\langle \Psi\right|\mathbb{B}^{0,\dots,0}.
\end{equation}
We can see that we reduced the original $\mathfrak{gl}(N)$ Bethe
state to $\mathfrak{gl}(N-2)$ ones. Naively, the uncrossed $KT$-relation
is not closed for this $\mathfrak{gl}(N-2)$ subsector since
\begin{equation}
\sum_{c=2}^{N-1}\mathbf{K}_{a,c}\left\langle \Psi\right|T_{c,b}+\mathbf{K}_{a,1}\left\langle \Psi\right|T_{1,b}+\mathbf{K}_{a,N}\left\langle \Psi\right|T_{N,b}=\sum_{c=2}^{N-1}\left\langle \Psi\right|\bar{T}_{a,c}\mathbf{K}_{c,b}+\left\langle \Psi\right|\bar{T}_{a,1}\mathbf{K}_{1,b}+\left\langle \Psi\right|\bar{T}_{a,N}\mathbf{K}_{N,b},\label{eq:nestin1-1}
\end{equation}
where $T_{i,j}\equiv T_{i,j}(u)$ and $\bar{T}_{i,j}\equiv T_{i,j}(-u)$.
We can see that the operators $T_{1,b}$ and $T_{a,N}$ create the
first and last types of Bethe roots therefore we have to eliminate
them from the equation (\ref{eq:nestin1-1}). We use the components
$(N,b)$, $(a,1)$ and $(N,1)$ of the uncrossed $KT$-relation 
\begin{align}
 & \sum_{c=2}^{N-1}\mathbf{K}_{N,c}\left\langle \Psi\right|T_{c,b}+\mathbf{K}_{N,1}\left\langle \Psi\right|T_{1,b}\cong\left\langle \Psi\right|\bar{T}_{N,N}\mathbf{K}_{N,b},\label{eq:nestin2-1}\\
 & \mathbf{K}_{a,1}\left\langle \Psi\right|T_{1,1}\cong\sum_{c=2}^{N-1}\left\langle \Psi\right|\bar{T}_{a,c}\mathbf{K}_{c,1}+\left\langle \Psi\right|\bar{T}_{a,N}\mathbf{K}_{N,1},\label{eq:nestin3-1}\\
 & \mathbf{K}_{N,1}\left\langle \Psi\right|T_{1,1}\cong\left\langle \Psi\right|\bar{T}_{N,N}\mathbf{K}_{N,1},\label{eq:nestin4-1}
\end{align}
where $\cong$ denotes the equality in the subspace generated by the
Bethe vectors $\mathbb{B}^{0,\dots,0}$. Combining the equations (\ref{eq:nestin1-1}),
(\ref{eq:nestin2-1}), (\ref{eq:nestin3-1}) and (\ref{eq:nestin4-1})
we obtain a new uncrossed $KT$-relation on the subspace generated
by the Bethe states $\mathbb{B}^{0,\dots,0}$
\begin{equation}
\sum_{c=2}^{N}\mathbf{K}_{a,c}^{(2)}(u)\left\langle \Psi\right|T_{c,b}(u)\cong\sum_{c=2}^{N}\left\langle \Psi\right|T_{a,c}(-u)\mathbf{K}_{c,b}^{(2)}(u),
\end{equation}
where
\begin{equation}
\mathbf{K}_{a,b}^{(2)}(z)=\mathbf{K}_{a,b}(z)-\mathbf{K}_{a,1}(z)\mathbf{K}_{N,1}^{-1}(z)\mathbf{K}_{N,b}(z),
\end{equation}
for $a,b=2,\dots,N-1$. Repeating the previous method we can eliminate
the sets of Bethe roots $\bar{t}^{\nu}$, therefore we obtain a recursion
for the overlaps and in the end of the day we obtain an explicit formula
for the overlap. In the $(k+1)$-th step of the nesting we have a
$\mathfrak{gl}(N-2k)$ $K$-matrix:
\begin{equation}
\mathbf{K}_{a,b}^{(k+1)}(z)=\mathbf{K}_{a,b}^{(k)}(z)-\mathbf{K}_{a,k}^{(k)}(z)\left[\mathbf{K}_{N+1-k,k}^{(k)}(z)\right]^{-1}\mathbf{K}_{N+1-k,b}^{(k)}(z),\label{eq:nestedKT-1}
\end{equation}
which satisfy the KT-relations
\begin{equation}
\sum_{c=k}^{N+1-k}\mathbf{K}_{a,c}^{(k)}(u)\left\langle \Psi\right|T_{c,b}(u)\cong\sum_{c=k}^{N+1-k}\left\langle \Psi\right|T_{a,c}(-u)\mathbf{K}_{c,b}^{(k)}(u),
\end{equation}
for $a,b=k,\dots,N+1-k$ where $\cong$ denotes the equality in the
subspace generated by the Bethe vectors $\mathbb{B}^{\dots,0,r_{k},\dots,r_{N-k},0\dots}$. 

We can examine the leading order of the nested $K$-matrices in the
asymptotic limit. It is easy to see that
\begin{equation}
\mathbf{K}_{i,j}^{(k)}(u)=\mathcal{U}_{i,j}^{(k)}\mathbf{1}+\mathcal{O}(u^{-1}).
\end{equation}
From the recursive equations (\ref{eq:nestedKT-1}), if (\ref{eq:UUa})
holds, then the result follows
\begin{equation}
\begin{split}\sum_{n=k}^{N+1-k}\mathcal{U}_{i,n}^{(k)}\mathcal{U}_{n,l}^{(k)} & =\sum_{n=1}^{N}\mathcal{U}_{i,n}\mathcal{U}_{n,l},\\
\sum_{n=k}^{N+1-k}\mathcal{U}_{n,n}^{(k)} & =\sum_{n=1}^{N}\mathcal{U}_{n,n}.
\end{split}
\label{eq:Urec}
\end{equation}
If the original $K$-matrix is a $\mathcal{B}(N,M)$ representation,
then $\mathcal{U}^{(k)}$ is an $N-2k+2$ dimensional matrix with
$M-k+1$ eigenvalues equal to $-1$ and $N-M-k+1$ eigenvalues equal
to $+1$. Based on this, at step $M+1$, $\mathcal{U}^{(M+1)}$ becomes
the identity matrix. 

During the calculation we assumed that the inverse $\left[\mathbf{K}_{N+1-k,k}^{(k)}(z)\right]^{-1}$
exists. As long as the components $\mathcal{U}_{N+1-k,k}^{(k)}$ are
non-zero, the matrices $\left[\mathbf{K}_{N+1-k,k}^{(k)}(z)\right]^{-1}$
exist. As long as the matrix $\mathcal{U}^{(k)}$ is not the identity,
$\mathcal{U}_{N+1-k,k}^{(k)}\neq0$ is guaranteed to hold after appropriate
rotations. As we saw earlier, $\mathcal{U}^{(k)}$ has $M-k+1$ eigenvalues
equal to $-1$, and the last inverse we need is at $k=\left\lfloor \frac{N}{2}\right\rfloor $,
meaning the necessary inverses exist if
\[
M-\left\lfloor \frac{N}{2}\right\rfloor +1\geq1,
\]
i.e., there is at least one eigenvalue equal to $-1$ at step $k=\left\lfloor \frac{N}{2}\right\rfloor $.

In summary, the necessary nested $K$-matrices exist only in the case
of $\mathcal{B}(N,\left\lfloor \frac{N}{2}\right\rfloor )$. In the
next section, we will precisely derive the overlaps for this case.

\section{Derivations of the overlaps\label{sec:Derivations-of-the-overlaps}}

For $Y^{+}(N)$ or $\mathcal{B}(N,\left\lfloor \frac{N}{2}\right\rfloor )$
$K$-matrices, with the recursions of the previous section we can
eliminate all the Bethe roots and in the end the off-shell overlaps
are expressed by a function of the Bethe roots $t_{k}^{\nu}$, the
$\alpha$-functions $\alpha_{\nu},$ and components of the $K$-matrices
$\mathbf{K}_{i,j}$ and the vacuum overlap $\mathbf{B}=\langle\Psi|0\rangle$,
i.e. we can introduce the off-shell overlap functions as
\begin{equation}
\mathbf{S}_{\bar{\alpha},\mathbf{B}}(\bar{t}):=\langle\Psi|\mathbb{B}(\bar{t}).
\end{equation}
In the calculations, we use only the recursion and action formulas
for off-shell Bethe vectors and the $KT$-relation. These formulas
do not depend on the specific representations used in the quantum
space. In a specific case, $\alpha_{k}(t_{j}^{k})$ depends on the
discrete variables $J$, $\Lambda^{(j)}$ and the continuous variables
$\theta_{j}$ (where $J$ is the number of inhomogeneities and representations).
For sufficiently large $J$, the variables $t_{j}^{k}$ and $\alpha_{k}(t_{j}^{k})$
can be considered independent. The overlap also depends on the vacuum
overlap $\mathbf{B}$ which is a matrix in the boundary space. This
matrix also depends on the discrete variables $J$, $\Lambda^{(j)}$,
and the continuous variables $\theta_{j}$. Let the eigenvalues of
$\mathbf{B}$ be $\beta_{\ell}$ for $\ell=1,\dots,d_{B}$. For sufficiently
large $J$, these can also be considered independent variables.

The derivations are essentially the same as those presented in the
earlier paper \cite{Gombor:2021hmj}. In that earlier work, the case
$d_{B}=1$ was studied, meaning that the additional difficulty now
arises from the fact that quantities which were previously scalars
are now matrices. Some parts of the derivations can be applied to
matrices without difficulty. The parts that require special attention
in the case $d_{B}>1$ are discussed in the next subsection \ref{subsec:Theorems-for-K}.
There, we present various theorems that describe the algebra of the
nested $K$-matrices and define commuting subalgebras. In subsections
\ref{subsec:Sum-formulas} and \ref{subsec:On-shell-limit}, we present
the theorems necessary for deriving the on-shell overlaps. These are
essentially the same as those discussed in the $d_{B}=1$ case, with
the difference that the scalar quantities depending on the $K$-matrix
are now replaced by the commuting operators introduced in \ref{subsec:Theorems-for-K}.

\subsection{Theorems for the $K$-matrices\label{subsec:Theorems-for-K}}

\subsubsection{Crossed $K$-matrices}

The nesting of the previous section defines the series of operators
\[
\begin{array}{ccccccccc}
Y^{+}(N) & \to & Y^{+}(N-1) & \to & \dots & \to & Y^{+}(2) & \to & Y^{+}(1)\\
\mathbf{K}^{(1)}\equiv\mathbf{K} & \to & \mathbf{K}^{(2)} & \to & \dots & \to & \mathbf{K}^{(N-1)} & \to & \mathbf{K}^{(N)}\\
\downarrow &  & \downarrow &  & \dots &  & \downarrow &  & \downarrow\\
\mathbf{G}^{(1)} &  & \mathbf{G}^{(2)} &  & \dots &  & \mathbf{G}^{(N-1)} &  & \mathbf{G}^{(N)}
\end{array}
\]
In the following we give the precise definitions and theorems for
these nested operators.
\begin{defn}
\label{def:The-nested--K}The nested $K$-matrices are
\begin{equation}
\mathbf{K}_{a,b}^{(k+1)}(z):=\mathbf{K}_{a,b}^{(k)}(z)-\mathbf{K}_{a,k}^{(k)}(z)\left[\mathbf{K}_{k,k}^{(k)}(z)\right]^{-1}\mathbf{K}_{k,b}^{(k)}(z),\label{eq:crNestedK}
\end{equation}
for $a,b=k+1,\dots,N$ and 
\begin{equation}
\mathbf{G}^{(k)}(u):=\mathbf{K}_{k,k}^{(k)}(u),
\end{equation}
for $k=1,\dots,N$.
\end{defn}
\begin{thm}
\label{thm:crossed-nested-K}The crossed nested $K$-matrices $\mathbf{K}^{(k)}(u-(k-1)/2)$
are representations of the twisted Yangian $Y^{+}(N+1-k)$ and the
G-operators satisfy the relations
\begin{align}
\begin{split}\left[\mathbf{G}^{(k)}(u_{2}),\mathbf{K}_{a,b}^{(k+1)}(u_{1})\right] & =0,\\
\left[\mathbf{G}^{(k)}(u_{2}),\mathbf{G}^{(k)}(u_{1})\right] & =0,\\
\left[\mathbf{G}^{(k)}(u_{2}),\mathbf{B}\right] & =0,
\end{split}
\label{eq:felcs}
\end{align}
for $k=1,\dots,N$ and $k<a,b\leq N$.
\end{thm}
\begin{cor}
The operators $\mathbf{G}^{(k)}(u)$ generate a Cartan subalgebra
of $Y^{+}(N)$

\begin{equation}
\left[\mathbf{G}^{(k)}(u_{2}),\mathbf{G}^{(l)}(u_{1})\right]=0,
\end{equation}
for $k,l=1,\dots,N$.
\end{cor}

\subsubsection{Uncrossed $K$-matrices}

In the uncrossed case we introduce 
\begin{equation}
n=\left\lfloor \frac{N}{2}\right\rfloor .
\end{equation}
The nesting of the uncrossed $K$-matrices defines the series 
\[
\begin{array}{ccccccccc}
\mathcal{B}(2n+1,n) & \to & \mathcal{B}(2n-1,n-1) & \to & \dots & \to & \mathcal{B}(3,1) & \to & \mathcal{B}(1,0)\\
\mathbf{K}^{(1)}\equiv\mathbf{K} & \to & \mathbf{K}^{(2)} & \to & \dots & \to & \mathbf{K}^{(n)} & \to & \mathbf{K}^{(n+1)}\\
\downarrow &  & \downarrow &  & \dots &  & \downarrow &  & \downarrow\\
\mathbf{G}^{(1)} &  & \mathbf{G}^{(2)} &  & \dots &  & \mathbf{G}^{(n)} &  & \mathbf{G}^{(n+1)}
\end{array}
\]
for $N=2n+1$ and
\[
\begin{array}{ccccccccc}
\mathcal{B}(2n,n) & \to & \mathcal{B}(2n-2,n-1) & \to & \dots & \to & \mathcal{B}(2,1) & \to & \mathcal{B}(1,0)\\
\mathbf{K}^{(1)}\equiv\mathbf{K} & \to & \mathbf{K}^{(2)} & \to & \dots & \to & \mathbf{K}^{(n)} & \to & \mathbf{K}^{(n+1)}\\
\downarrow &  & \downarrow &  & \dots &  & \downarrow &  & \downarrow\\
\mathbf{G}^{(1)} &  & \mathbf{G}^{(2)} &  & \dots &  & \mathbf{G}^{(n)} &  & \mathbf{G}^{(n+1)}
\end{array}
\]
for $N=2n$.
\begin{defn}
The nested $K$-matrices are
\begin{equation}
\mathbf{K}_{a,b}^{(k+1)}(z):=\mathbf{K}_{a,b}^{(k)}(z)-\mathbf{K}_{a,k}^{(k)}(z)\left[\mathbf{K}_{N+1-k,k}^{(k)}(z)\right]^{-1}\mathbf{K}_{N+1-k,b}^{(k)}(z),
\end{equation}
for $a,b=k+1,\dots,N-k$ and 
\begin{equation}
\mathbf{G}^{(k)}(u)=\mathbf{K}_{N+1-k,k}^{(k)}(u),
\end{equation}
for $k=1,\dots,n+1$.
\end{defn}
\begin{thm}
\label{thm:uncrossed-nestedK}The uncrossed nested $K$-matrices $\mathbf{K}^{(k)}(u)$
are representations the reflection algebras $\mathcal{B}(N+2-2k,n+1-k)$
and the G-operators satisfy the relations
\begin{align}
\left[\mathbf{G}^{(k)}(u_{2}),\mathbf{K}_{a,b}^{(k+1)}(u_{1})\right] & =0,\nonumber \\
\left[\mathbf{G}^{(k)}(u_{2}),\mathbf{G}^{(k)}(u_{1})\right] & =0,\label{eq:felcsUcr}\\
\left[\mathbf{G}^{(k)}(u_{2}),\mathbf{B}\right] & =0,\nonumber 
\end{align}
for $k=1,\dots,n+1$ and $k<a,b<N+1-k$.
\end{thm}
\begin{cor}
The operators $\mathbf{G}^{(k)}(u)$ generate a Cartan subalgebra
of $\mathcal{B}(N,n)$

\begin{equation}
\left[\mathbf{G}^{(k)}(u_{2}),\mathbf{G}^{(l)}(u_{1})\right]=0,
\end{equation}
for $k,l=1,\dots,n+1$.
\end{cor}

\subsection{Off-shell overlap formulas\label{subsec:Sum-formulas}}

In this subsection we show theorems for the off-shell overlaps $\mathbf{S}_{\alpha,\mathbf{B}}(\bar{t})$.
The derivations can be found in the Appendix \ref{sec:Proofs-for-overlaps}.

The off-shell overlap functions will be written as sums over partitions
of the sets of Bethe roots. A partition $\bar{t}^{s}\vdash\{\bar{t}_{\textsc{i}}^{s},\bar{t}_{\textsc{ii}}^{s}\}$
corresponds to a decomposition into (possibly empty) disjoint subsets
$\bar{t}_{\textsc{i}}^{s},\bar{t}_{\textsc{ii}}^{s}$ such that $\bar{t}^{s}=\bar{t}_{\textsc{i}}^{s}\cup\bar{t}_{\textsc{ii}}^{s}$
and $\bar{t}_{\textsc{i}}^{s}\cap\bar{t}_{\textsc{ii}}^{s}=\emptyset$. 
\begin{thm}
\label{thm:The-sum-formula}The sum formula can be written as
\begin{equation}
\mathbf{S}_{\bar{\alpha},\mathbf{B}}(\bar{t})=\sum_{\mathrm{part}(\bar{t})}\frac{\prod_{\nu=1}^{N-1}f(\bar{t}_{\textsc{ii}}^{\nu},\bar{t}_{\textsc{i}}^{\nu})}{\prod_{\nu=1}^{N-2}f(\bar{t}_{\textsc{ii}}^{\nu+1},\bar{t}_{\textsc{i}}^{\nu})}\bar{\mathbf{Z}}(\bar{t}_{\textsc{ii}})\mathbf{B}\mathbf{Z}(\bar{t}_{\textsc{i}})\prod_{\nu=1}^{N-1}\alpha_{\nu}(\bar{t}_{\textsc{i}}^{\nu}),
\end{equation}
where the sum goes over partitions $\bar{t}^{s}\vdash\{\bar{t}_{\textsc{i}}^{s},\bar{t}_{\textsc{ii}}^{s}\}$.
The highest coefficients $\mathbf{Z}$ and \textup{$\bar{\mathbf{Z}}$
satisfy
\begin{equation}
\mathbf{Z}^{\mathbf{K}}(\bar{t})=(-1)^{\#\bar{t}}\left(\prod_{s=1}^{N-2}f(\bar{t}^{s+1},\bar{t}^{s})\right)^{-1}\left[\bar{\mathbf{Z}}^{\mathbf{K}^{T}}(\pi^{c}(\bar{t}))\right]^{t_{B}},\label{eq:ZZBcr}
\end{equation}
}\textup{\emph{in the crossed case and }}\textup{
\begin{equation}
\mathbf{Z}^{\mathbf{K}}(\bar{t})=\left[\bar{\mathbf{Z}}^{\mathbf{K}^{\Pi}}(\pi^{a}(\bar{t}))\right]^{t_{B}},
\end{equation}
}\textup{\emph{in the uncrossed case, where $^{T}$ denotes the transposition
on the auxiliary and boundary spaces and}}\\
\textup{\emph{ $\mathbf{K}_{i,j}^{\Pi}=\left(\mathbf{K}_{N+1-j,N+1-i}\right)^{t_{B}}$.}}\textup{}
\end{thm}
The HC-s have pole in the pair structure limit. In the crossed and
the uncrossed cases the pair structure limits are $t_{l}^{\nu}\to-t_{k}^{\nu}-\nu$
and $t_{l}^{N-\nu}\to-t_{k}^{\nu}$, respectively. We can use a common
notation $t_{l}^{\tilde{\nu}}\to-t_{k}^{\nu}-\nu\tilde{c}$ for the
pair structure limits, where 
\begin{equation}
(\tilde{\nu},\tilde{c})=\begin{cases}
(\nu,1), & \text{for the crossed case},\\
(N-\nu,0), & \text{for the uncrossed case}.
\end{cases}
\end{equation}
 
\begin{thm}
\label{thm:HC-poles}The HC-s have poles at $t_{l}^{\tilde{\nu}}\to-t_{k}^{\nu}-\nu\tilde{c}$:
\begin{equation}
\mathbf{Z}(\bar{t})\to\frac{-1}{t_{l}^{\tilde{\nu}}+t_{k}^{\nu}+\nu\tilde{c}}\frac{f(t_{k}^{\nu},\bar{\tau}^{\nu})f(t_{l}^{\tilde{\nu}},\bar{\tau}^{\tilde{\nu}})}{f(t_{k}^{\nu},\bar{t}^{\nu-1})f(t_{l}^{\tilde{\nu}},\bar{t}^{\tilde{\nu}-1})}\mathbf{F}^{(\nu)}(t_{k}^{\nu})\mathbf{Z}(\bar{\tau})+reg.,\label{eq:poleHC-1-1}
\end{equation}
and
\begin{equation}
\bar{\mathbf{Z}}(\bar{t})\to\frac{1}{t_{l}^{\tilde{\nu}}+t_{k}^{\nu}+\nu\tilde{c}}\frac{f(\bar{\tau}^{\nu},t_{k}^{\nu})f(\bar{\tau}^{\tilde{\nu}},t_{l}^{\tilde{\nu}})}{f(\bar{t}^{\nu+1},t_{k}^{\nu})f(\bar{t}^{\tilde{\nu}+1},t_{l}^{\tilde{\nu}})}\bar{\mathbf{Z}}(\bar{\tau})\mathbf{F}^{(\nu)}(t_{k}^{\nu})+reg.,\label{eq:poleHC-1}
\end{equation}
where $\bar{\tau}=\bar{t}\backslash\{t_{k}^{\nu},t_{l}^{\tilde{\nu}}\}$
and 
\begin{equation}
\mathbf{F}^{(\nu)}(z)=(-1)^{\tilde{c}}\left[\mathbf{G}^{(\nu)}(z)\right]^{-1}\mathbf{G}^{(\nu+1)}(z).
\end{equation}
\end{thm}
We have the commuting set of operators

\begin{equation}
\left[\mathbf{F}^{(k)}(u),\mathbf{F}^{(l)}(v)\right]=0,\quad\left[\mathbf{F}^{(k)}(u),\mathbf{B}\right]=0,
\end{equation}
therefore we can diagonalize the matrices simultaneously:
\begin{equation}
\left(\mathbf{F}^{(k)}(u)\right)_{n,m}=\delta_{n,m}\mathcal{F}_{n}^{(k)}(u),\quad\left(\mathbf{B}\right)_{n,m}=\delta_{n,m}\beta_{n}.
\end{equation}
The off-shell overlap have the sum formula
\begin{equation}
\begin{split}\langle\mathrm{MPS}|\mathbb{B}(\bar{t}) & =\sum_{\mathrm{part}}\frac{\prod_{\nu=1}^{N-1}f(\bar{t}_{\textsc{ii}}^{\nu},\bar{t}_{\textsc{i}}^{\nu})}{\prod_{\nu=1}^{N-2}f(\bar{t}_{\textsc{ii}}^{\nu+1},\bar{t}_{\textsc{i}}^{\nu})}\mathrm{Tr}\left[\bar{\mathbf{Z}}(\bar{t}_{\textsc{ii}})\mathbf{B}\mathbf{Z}(\bar{t}_{\textsc{i}})\right]\prod_{\nu}\alpha_{\nu}(\bar{t}_{\textsc{i}}^{\nu})\\
 & =\sum_{\ell=1}^{d_{B}}\beta_{\ell}\sum_{\mathrm{part}}\frac{\prod_{\nu=1}^{N-1}f(\bar{t}_{\textsc{ii}}^{\nu},\bar{t}_{\textsc{i}}^{\nu})}{\prod_{\nu=1}^{N-2}f(\bar{t}_{\textsc{ii}}^{\nu+1},\bar{t}_{\textsc{i}}^{\nu})}\left(\mathbf{Z}(\bar{t}_{\textsc{i}})\bar{\mathbf{Z}}(\bar{t}_{\textsc{ii}})\right)_{\ell,\ell}\prod_{\nu}\alpha_{\nu}(\bar{t}_{\textsc{i}}^{\nu}).
\end{split}
\end{equation}
We can define the expressions
\begin{equation}
S_{\bar{\alpha}}^{(\ell)}(\bar{t}):=\sum_{\mathrm{part}}\frac{\prod_{\nu=1}^{N-1}f(\bar{t}_{\textsc{ii}}^{\nu},\bar{t}_{\textsc{i}}^{\nu})}{\prod_{\nu=1}^{N-2}f(\bar{t}_{\textsc{ii}}^{\nu+1},\bar{t}_{\textsc{i}}^{\nu})}\left(\mathbf{Z}(\bar{t}_{\textsc{i}})\bar{\mathbf{Z}}(\bar{t}_{\textsc{ii}})\right)_{\ell,\ell}\prod_{\nu}\alpha_{\nu}(\bar{t}_{\textsc{i}}^{\nu}),
\end{equation}
for which
\begin{equation}
\langle\mathrm{MPS}|\mathbb{B}(\bar{t})=\sum_{\ell=1}^{d_{B}}\beta_{\ell}S_{\bar{\alpha}}^{(\ell)}(\bar{t}).
\end{equation}

\begin{thm}
\label{thm:Spair}The expressions $S_{\bar{\alpha}}^{(n)}(\bar{t})$
has the following pair structure limit
\begin{equation}
\lim_{t_{l}^{\tilde{\nu}}\to-t_{k}^{\nu}-\nu\tilde{c}}S_{\bar{\alpha}}^{(\ell)}(\bar{t})=\mathcal{F}_{\ell}^{(\nu)}(t_{k}^{\nu})X_{k}^{\nu}\frac{f(\bar{\tau}^{\nu},t_{k}^{\nu})f(\bar{\tau}^{\tilde{\nu}},t_{l}^{\tilde{\nu}})}{f(\bar{t}^{\nu+1},t_{k}^{\nu})f(\bar{t}^{\tilde{\nu}+1},t_{l}^{\tilde{\nu}})}S_{\bar{\alpha}^{mod}}^{(\ell)}(\bar{\tau})+\tilde{S},\label{eq:Spair}
\end{equation}
where $\bar{\tau}=\bar{t}\backslash\{t_{k}^{\nu},t_{l}^{\tilde{\nu}}\}$
and
\begin{equation}
X_{k}^{\nu}=-\frac{d}{dt_{k}^{\nu}}\log\alpha(t_{k}^{\nu}),\label{eq:Xdef}
\end{equation}
the $\tilde{S}$ does not depend on \textup{$X_{k}^{\nu}$.}\textup{\emph{
The modified $\alpha$-s are}}\textup{
\begin{equation}
\alpha_{\mu}^{mod}(u)=\alpha_{\mu}(u)\left(\frac{f(t_{k}^{\nu},u)}{f(u,t_{k}^{\nu})}\right)^{\delta_{\nu,\mu}}\left(\frac{f(-t_{k}^{\nu}-\nu\tilde{c},u)}{f(u,-t_{k}^{\nu}-\nu\tilde{c})}\right)^{\delta_{\tilde{\nu},\mu}}\frac{f(u,t_{k}^{\nu})^{\delta_{\nu,\mu-1}}}{f(t_{k}^{\nu},u)^{\delta_{\nu,\mu+1}}}\frac{f(u,-t_{k}^{\nu}-\nu\tilde{c})^{\delta_{\tilde{\nu},\mu-1}}}{f(-t_{k}^{\nu}-\nu\tilde{c},u)^{\delta_{\tilde{\nu},\mu+1}}}.\label{eq:amod}
\end{equation}
}
\end{thm}
The definition of $\alpha_{\mu}^{mod}$ is compatible with the on-shell
limit. In the original overlap function, the on-shell Bethe roots
satisfy the Bethe equations 
\begin{equation}
\alpha_{\mu}(t_{k}^{\mu})=\frac{f(t_{k}^{\mu},\bar{t}_{k}^{\mu})}{f(\bar{t}_{k}^{\mu},t_{k}^{\mu})}\frac{f(\bar{t}^{\mu+1},t_{k}^{\mu})}{f(t_{k}^{\mu},\bar{t}^{\mu-1})}.\label{eq:BE}
\end{equation}
On the right-hand side of (\ref{eq:Spair}), we have overlaps where
the Bethe roots $t_{k}^{\nu}$ and $t_{l}^{\tilde{\nu}}$ have been
omitted. The remaining Bethe roots must satisfy the same equations
on-shell as before. It is clear that with the original $\alpha$-functions,
the Bethe equations are not satisfied for the set $\bar{\tau}$:
\begin{equation}
\alpha_{\mu}(\tau_{k}^{\mu})\neq\frac{f(\tau_{k}^{\mu},\bar{\tau}_{k}^{\mu})}{f(\bar{\tau}_{k}^{\mu},\tau_{k}^{\mu})}\frac{f(\bar{\tau}^{\mu+1},\tau_{k}^{\mu})}{f(\tau_{k}^{\mu},\bar{\tau}^{\mu-1})}.
\end{equation}
However, with the modified $\alpha$-functions, they are satisfied.
To see this, separate the terms on the right-hand side of the Bethe
equations (\ref{eq:BE}) according to the decomposition $\bar{t}=\bar{\tau}\cup\{t_{k}^{\nu},t_{l}^{\tilde{\nu}}\}$:
\begin{equation}
\frac{f(\tau_{k}^{\mu},\bar{t}_{k}^{\mu})}{f(\bar{t}_{k}^{\mu},\tau_{k}^{\mu})}\frac{f(\bar{t}^{\mu+1},\tau_{k}^{\mu})}{f(\tau_{k}^{\mu},\bar{t}^{\mu-1})}=\frac{f(\tau_{k}^{\mu},\bar{\tau}_{k}^{\mu})}{f(\bar{\tau}_{k}^{\mu},\tau_{k}^{\mu})}\frac{f(\bar{\tau}^{\mu+1},\tau_{k}^{\mu})}{f(\tau_{k}^{\mu},\bar{\tau}^{\mu-1})}\left(\frac{f(\tau_{k}^{\mu},t_{k}^{\nu})}{f(t_{k}^{\nu},\tau_{k}^{\mu})}\right)^{\delta_{\nu,\mu}}\left(\frac{f(\tau_{k}^{\mu},t_{l}^{\tilde{\nu}})}{f(t_{l}^{\tilde{\nu}},\tau_{k}^{\mu})}\right)^{\delta_{\tilde{\nu},\mu}}\frac{f(t_{k}^{\nu},\tau_{k}^{\mu})^{\delta_{\nu,\mu+1}}}{f(\tau_{k}^{\mu},t_{k}^{\nu})^{\delta_{\nu,\mu-1}}}\frac{f(t_{l}^{\tilde{\nu}},\tau_{k}^{\mu})^{\delta_{\tilde{\nu},\mu+1}}}{f(\tau_{k}^{\mu},t_{l}^{\tilde{\nu}})^{\delta_{\tilde{\nu},\mu-1}}}.
\end{equation}
Substituting back into the original Bethe equation (\ref{eq:BE}),
we obtain the result
\begin{equation}
\alpha_{\mu}^{mod}(\tau_{k}^{\mu})=\frac{f(\tau_{k}^{\mu},\bar{\tau}_{k}^{\mu})}{f(\bar{\tau}_{k}^{\mu},\tau_{k}^{\mu})}\frac{f(\bar{\tau}^{\mu+1},\tau_{k}^{\mu})}{f(\tau_{k}^{\mu},\bar{\tau}^{\mu-1})},
\end{equation}
where $\alpha_{\mu}^{mod}$ is defined according to (\ref{eq:amod}),
after the substitution $t_{l}^{\tilde{\nu}}=-t_{k}^{\nu}-\nu\tilde{c}$.

\subsection{On-shell limit\label{subsec:On-shell-limit}}

We saw that the on-shell overlaps are non-vanishing only if the Bethe
roots have pair structure. In the crossed case the Bethe roots have
chiral pair structure. For simplicity, we assume that the numbers
$r_{\nu}$ are even \footnote{The case of odd values is discussed in Section S-VIII of \cite{Gombor:2024iix}.}.
In the chiral pair structure, we can decompose the sets $\bar{t}^{\nu}\vdash\{\bar{t}^{+,\nu},\bar{t}^{-,\nu}\}$,
where $\#\bar{t}^{+,\nu}=\#\bar{t}^{-,\nu}=r_{\nu}^{+}=r_{\nu}/2$.
In the pair structure limit, $t_{k}^{-,\nu}=-t_{k}^{+,\nu}-\nu$ for
$k=1,\dots,r_{\nu}^{+}$. 

In the uncrossed case the Bethe roots have achiral pair structure.
For odd $N$, i.e., when $N=2n+1$, we define the notations $\bar{t}^{+,\nu}=\bar{t}^{\nu}$
and $\bar{t}^{-,\nu}=\bar{t}^{N-\nu}$, where $\#\bar{t}^{+,\nu}=\#\bar{t}^{-,\nu}=r_{\nu}^{+}=r_{\nu}$
for $\nu=1,\dots,n$. In the pair structure limit, $t_{k}^{-,\nu}=-t_{k}^{+,\nu}$
for $k=1,\dots,r_{\nu}^{+}$. 

If $N=2n$, then for simplicity we assume that the number $r_{n}$
is even. In this case, we can decompose the set $\bar{t}^{n}\vdash\{\bar{t}^{+,n},\bar{t}^{-,n}\}$
where $\#\bar{t}^{+,n}=\#\bar{t}^{-,n}=r_{n}^{+}=r_{n}/2$. For the
other nodes, we introduce the earlier notation $\bar{t}^{+,\nu}=\bar{t}^{\nu}$
and $\bar{t}^{-,\nu}=\bar{t}^{N-\nu}$ for $\nu=1,\dots,n-1$. In
the pair structure limit, $t_{k}^{-,\nu}=-t_{k}^{+,\nu}$ for $k=1,\dots,r_{\nu}^{+}$
and $\nu=1,\dots,n$. 

In both the chiral and achiral cases, we introduce the set of sets
$\bar{t}^{\pm}=\left\{ \bar{t}^{\pm,\nu}\right\} _{\nu=1}^{n_{+}}$,
where
\[
n_{+}=\begin{cases}
N-1, & \text{for crossed cases,}\\
n, & \text{for uncrossed cases.}
\end{cases}
\]

The Theorem \ref{thm:Spair} is sufficient to prove the form of the
overlaps involving Gaudin determinants.
\begin{defn}
The Gaudin matrices are defined by
\begin{align}
G_{j,k}^{\pm,(\mu,\nu)} & =-\left(\frac{\partial}{\partial t_{k}^{+,\nu}}\pm\frac{\partial}{\partial t_{k}^{-,\nu}}\right)\log\Phi_{j}^{+,(\mu)}\Biggr|_{\bar{t}^{-}=\pi^{a/c}\left(\bar{t}^{+}\right)},
\end{align}
where $\mu,\nu=1,\dots,n_{+}$, $j,k=1,\dots,r_{\nu}^{+}$ and
\begin{equation}
\Phi_{k}^{+,(\mu)}=\alpha_{\mu}(t_{k}^{+,\mu})\frac{f(\bar{t}_{k}^{\mu},t_{k}^{+,\mu})}{f(t_{k}^{+,\mu},\bar{t}_{k}^{\mu})}\frac{f(t_{k}^{+,\mu},\bar{t}^{\mu-1})}{f(\bar{t}^{\mu+1},t_{k}^{+,\mu})}.
\end{equation}
\end{defn}
\begin{cor}
The Gaudin-determinant factorizes as
\begin{equation}
\det G=\det G^{+}\det G^{-},
\end{equation}
for Bethe roots with pair structure $\bar{t}=\pi^{a/c}\left(\bar{t}\right)$.
\end{cor}
After performing the pair structure limit, the elementary overlap
functions $S^{(\ell)}$ depend only on the variables $t_{k}^{+,\mu}$,
$\alpha_{k}^{+,\mu}\equiv\alpha_{\mu}(t_{k}^{+,\mu})$ and $X_{k}^{+,\mu}$.
In the generalized model, these are independent variables. After the
pair structure limit, we can also perform the on-shell limit via
\begin{equation}
\alpha_{k}^{+,\mu}\to\frac{f(t_{k}^{+,\mu},\bar{t}_{k}^{\mu})}{f(\bar{t}_{k}^{\mu},t_{k}^{+,\mu})}\frac{f(\bar{t}^{\mu+1},t_{k}^{+,\mu})}{f(t_{k}^{+,\mu},\bar{t}^{\mu-1})}
\end{equation}
substitution. As a result, we obtain functions that depend only on
the variables $t_{k}^{+,\mu}$ and $X_{k}^{+,\mu}$ :
\begin{equation}
S_{\bar{\alpha}}^{(n)}(\bar{t})\to S^{(n)}(\bar{t}^{+}|\bar{X}^{+}).
\end{equation}

\begin{thm}
\label{thm:onshell}The crossed elementary overlaps are
\begin{equation}
S^{(\ell)}(\bar{t}^{+}|\bar{X}^{+})=\prod_{\nu=1}^{N-1}\mathcal{F}_{\ell}^{(\nu)}(\bar{t}^{+,\nu})\times\prod_{\nu=1}^{N-1}\prod_{k\neq l}f(t_{l}^{+,\nu},t_{k}^{+,\nu})\prod_{k<l}f(t_{l}^{+,\nu},-t_{k}^{+,\nu}-\nu c)f(-t_{l}^{+,\nu}-\nu c,t_{k}^{+,\nu})\det G^{+}.
\end{equation}
The uncrossed elementary overlaps are
\begin{equation}
S^{(\ell)}(\bar{t}^{+}|\bar{X}^{+})=\prod_{\nu=1}^{n}\mathcal{F}_{\ell}^{(\nu)}(\bar{t}^{+,\nu})\times\frac{\prod_{\nu=1}^{n}\prod_{k\neq l}f(t_{l}^{+,\nu},t_{k}^{+,\nu})}{\prod_{\nu=1}^{n-1}f(\bar{t}^{+,\nu+1},\bar{t}^{+,\nu})\prod_{k\leq l}f(-t_{l}^{+,\frac{N-1}{2}},t_{k}^{+,\frac{N-1}{2}})}\det G^{+},\label{eq:onshell2}
\end{equation}
for $N=2n+1$ and
\begin{equation}
S^{(\ell)}(\bar{t}^{+}|\bar{X}^{+})=\prod_{\nu=1}^{n}\mathcal{F}_{\ell}^{(\nu)}(\bar{t}^{+,\nu})\times\frac{\prod_{\nu=1}^{n}\prod_{k\neq l}f(t_{l}^{+,\nu},t_{k}^{+,\nu})\prod_{k<l}f(-t_{l}^{+,\frac{N}{2}},t_{k}^{+,\frac{N}{2}})f(t_{l}^{+,\frac{N}{2}},-t_{k}^{+,\frac{N}{2}})}{\prod_{\nu=1}^{n-1}f(\bar{t}^{+,\nu+1},\bar{t}^{+,\nu})f(-\bar{t}^{+,\frac{N}{2}},\bar{t}^{+,\frac{N}{2}-1})}\det G^{+},\label{eq:onshell3}
\end{equation}
for $N=2n$.
\end{thm}
\begin{cor}
The normalized on-shell overlaps are
\begin{equation}
\frac{\langle\mathrm{MPS}|\mathbb{B}(\bar{t})}{\sqrt{\mathbb{C}(\bar{t})\mathbb{B}(\bar{t})}}=\sum_{\ell=1}^{d_{B}}\beta_{\ell}\prod_{\nu=1}^{n_{+}}\tilde{\mathcal{F}}_{\ell}^{(\nu)}(\bar{t}^{+,\nu})\sqrt{\frac{\det G^{+}}{\det G^{-}}},\label{eq:onOV}
\end{equation}
where
\begin{equation}
\tilde{\mathcal{F}}_{\ell}^{(\nu)}(u)=\begin{cases}
\frac{\mathcal{F}_{\ell}^{(\nu)}(u)}{\sqrt{f(-u-\nu,u)f(u,-u-\nu)}}, & \text{for the chiral case,}\\
\mathcal{F}_{\ell}^{(\nu)}(u)\left(\frac{1}{\sqrt{f(-u,u)}}\right)^{\delta_{\nu,\frac{N-1}{2}}}\left(\frac{1}{\sqrt{f(-u,u)f(u,-u)}}\right)^{\delta_{\nu,\frac{N}{2}}}, & \text{for the achiral case.}
\end{cases}
\end{equation}
\end{cor}

\subsection{Redefinition of the Bethe roots}

In applications, the most common convention is the one where the Bethe
equations contain product of phases. To align with this, we introduce
a new convention for the Bethe roots:
\begin{equation}
t_{k}^{\nu}=\begin{cases}
iu_{k}^{\nu}-\nu/2, & \text{for the crossed case},\\
iu_{k}^{\nu}+\frac{N-2\nu}{4}, & \text{for the uncrossed case}.
\end{cases}
\end{equation}
The Bethe equations with these new variables are
\begin{equation}
\tilde{\alpha}_{\mu}(u_{k}^{\mu})=-\prod_{j=1}^{r_{\mu}}\frac{u_{k}^{\mu}-u_{j}^{\mu}-i}{u_{k}^{\mu}-u_{j}^{\mu}+i}\prod_{j=1}^{r_{\mu+1}}\frac{u_{k}^{\mu}-u_{j}^{\mu+1}+i/2}{u_{k}^{\mu}-u_{j}^{\mu+1}-i/2}\prod_{j=1}^{r_{\mu-1}}\frac{u_{k}^{\mu}-u_{j}^{\mu-1}+i/2}{u_{k}^{\mu}-u_{j}^{\mu-1}-i/2},
\end{equation}
where
\begin{equation}
\tilde{\alpha}_{\mu}(u)=\begin{cases}
\alpha_{\mu}(iu-\mu/2), & \text{for the crossed case},\\
\alpha_{\mu}(iu+\frac{N-2\mu}{4}), & \text{for the uncrossed case}.
\end{cases}
\end{equation}
Let us also introduce a new notation for the Bethe vectors
\begin{equation}
|\bar{u}\rangle=\begin{cases}
\mathbb{B}(\left\{ i\bar{u}^{\nu}-\nu/2\right\} _{\nu=1}^{N-1}), & \text{for the crossed case},\\
\mathbb{B}(\left\{ i\bar{u}^{\nu}+\frac{N-2\nu}{4}\right\} _{\nu=1}^{N-1}), & \text{for the uncrossed case},
\end{cases}
\end{equation}
for which the on-shell overlap formulas have the form
\begin{equation}
\frac{\langle\mathrm{MPS}|\bar{u}\rangle}{\sqrt{\langle\bar{u}|\bar{u}\rangle}}=\sum_{\ell=1}^{d_{B}}\beta_{\ell}\prod_{\nu=1}^{n_{+}}\tilde{\mathcal{F}}_{\ell}^{(\nu)}(\bar{u}^{+,\nu})\sqrt{\frac{\det G^{+}}{\det G^{-}}},
\end{equation}
where
\begin{equation}
\tilde{\mathcal{F}}_{\ell}^{(\nu)}(u)=\begin{cases}
\mathcal{F}_{\ell}^{(\nu)}(iu-\nu/2)\sqrt{\frac{u^{2}}{u^{2}+1/4}}, & \text{for the chiral case,}\\
\mathcal{F}_{\ell}^{(\nu)}(iu+\frac{N-2\nu}{4})\left(\sqrt{\frac{u-i/4}{u+i/4}}\right)^{\delta_{\nu,\frac{N-1}{2}}}\left(\sqrt{\frac{u^{2}}{u^{2}+1/4}}\right)^{\delta_{\nu,\frac{N}{2}}}, & \text{for the achiral case.}
\end{cases}\label{eq:Ftdef}
\end{equation}

\section{Other reflection algebras of $\mathfrak{gl}_{N}$ spin chains\label{sec:Other-reflection-algebras}}

The formulas from the previous section can be applied to the $Y^{+}(N)$
and $\mathcal{B}(N,\left\lfloor \frac{N}{2}\right\rfloor )$ $K$-matrices.
From now on, we assume that the universal overlap formula (\ref{eq:onOV-1})
also exists for the other $\mathfrak{gl}_{N}$ reflection algebras.
In this section, we extend the definition of the $\mathbf{F}^{(\nu)}$
operators to the remaining cases as well.

\subsection{Generalization for $Y^{-}(2n)$\label{subsec:Generalization-Ym}}

First, we extend the $F$-operators of the $Y^{+}(N)$ algebra to
the $Y^{-}(N)$ case. We begin by solving the recursive equation (\ref{eq:crNestedK})
for the nested $K$-matrix, from which we obtain the $G$-operators
in closed form. Additionally, we formulate a new recursion for the
$G$-operators, which we apply to the $Y^{-}(N)$ $K$-matrices. 

\subsubsection{Quasi-determinants\label{subsec:Quasi-determinants}}

The $G$-operators have several equivalent definitions. To express
these, it is useful to introduce new notations. Let $X=\left(X_{i,j}\right)_{i,j=1}^{N}$
be an $N\times N$ matrix over a ring. Let $1\leq M<N$ and introduce
the following block notations,
\begin{equation}
X=\left(\begin{array}{cc}
A & B\\
C & D
\end{array}\right),
\end{equation}
where $A$ is an $M\times M$, $B$ an $M\times(N-M)$, $C$ a $(N-M)\times M$
and $D$ a $(N-M)\times(N-M)$ matrix over the ring. Assume that $A$
is invertible, and then let 
\begin{equation}
X^{(M+1)}\equiv\left(\begin{array}{cc}
A & B\\
C & \boxed{D}
\end{array}\right):=D-CA^{-1}B
\end{equation}
be a $(N-M)\times(N-M)$ matrix. This matrix can also be defined in
another equivalent way. Let $\mathcal{I}$ denote the inversion operation,
i.e., $\mathcal{I}(X):=X^{-1}$, and let $\Pi_{M}$ be a projection
such that
\begin{equation}
\Pi_{M}(X)=\Pi_{M}\left(\left(\begin{array}{cc}
A & B\\
C & D
\end{array}\right)\right):=D.
\end{equation}
Combining these, we introduce the following operation.
\begin{equation}
\omega_{M}:=\mathcal{I}\circ\Pi_{M}\circ\mathcal{I}.
\end{equation}
It is easy to verify that the inverse matrix can be expressed in the
following way:
\begin{equation}
X^{-1}=\left(\begin{array}{cc}
(A-BD^{-1}C)^{-1} & -A^{-1}B(D-CA^{-1}B)^{-1}\\
-D^{-1}C(A-BD^{-1}C)^{-1} & (D-CA^{-1}B)^{-1}
\end{array}\right),
\end{equation}
therefore
\begin{equation}
X^{(M+1)}=\left(\begin{array}{cc}
A & B\\
C & \boxed{D}
\end{array}\right)=\omega_{M}(X).
\end{equation}
Clearly, $\omega_{M_{1}+M_{2}}=\omega_{M_{1}}\circ\omega_{M_{2}}$,
that is,
\begin{equation}
X^{(M_{1}+M_{2})}=\omega_{M_{1}}(X^{(M_{2})}).\label{eq:omId}
\end{equation}

These definitions can be applied to the nested $K$-matrices. The
recursive step can be written using the new notations as follows
\[
\mathbf{K}^{(s+1)}=\left(\begin{array}{cc}
\mathbf{K}_{s,s}^{(s)} & \begin{array}{cc}
\mathbf{K}_{s,s+1}^{(s)} & \dots\end{array}\\
\begin{array}{c}
\mathbf{K}_{s+1,s}^{(s)}\\
\vdots
\end{array} & \boxed{\begin{array}{cc}
\mathbf{K}_{s+1,s+1}^{(s)} & \ldots\\
\vdots & \ddots
\end{array}}
\end{array}\right)=\omega_{1}(\mathbf{K}^{(s)}).
\]
Using the identity (\ref{eq:omId}), the recursive equations can be
solved
\[
\mathbf{K}^{(s)}=\omega_{s-1}(\mathbf{K})=\left(\begin{array}{cccc}
\mathbf{K}_{1,1} & \dots & \mathbf{K}_{1,s-1} & \begin{array}{ccc}
\mathbf{K}_{1,s} & \mathbf{K}_{1,s+1} & \dots\end{array}\\
\vdots & \ddots & \vdots & \vdots\\
\mathbf{K}_{s-1,1} & \dots & \mathbf{K}_{s-1,s-1} & \begin{array}{ccc}
\mathbf{K}_{s-1,s} & \mathbf{K}_{s-1,s+1} & \dots\end{array}\\
\begin{array}{c}
\mathbf{K}_{s,1}\\
\mathbf{K}_{s+1,1}\\
\vdots
\end{array} & \dots & \begin{array}{c}
\mathbf{K}_{s,s-1}\\
\mathbf{K}_{s+1,s-1}\\
\vdots
\end{array} & \boxed{\begin{array}{ccc}
\mathbf{K}_{s,s} & \mathbf{K}_{s,s+1} & \ldots\\
\mathbf{K}_{s+1,s} & \mathbf{K}_{s+1,s+1} & \ldots\\
\vdots & \vdots & \ddots
\end{array}}
\end{array}\right),
\]
or simply by using components
\begin{equation}
\mathbf{K}_{i,j}^{(s)}=\left(\begin{array}{cccc}
\mathbf{K}_{1,1} & \dots & \mathbf{K}_{1,s-1} & \mathbf{K}_{1,j}\\
\vdots & \ddots & \vdots & \vdots\\
\mathbf{K}_{s-1,1} & \dots & \mathbf{K}_{s-1,s-1} & \mathbf{K}_{s-1,j}\\
\mathbf{K}_{i,1} & \dots & \mathbf{K}_{i,s-1} & \boxed{\mathbf{K}_{i,j}}
\end{array}\right),\label{eq:Kmqdet}
\end{equation}
where $i,j=s,\dots,N$. Based on this, the $G$-operators can also
be expressed in closed form.
\begin{equation}
\mathbf{G}^{(s)}=\left(\begin{array}{cccc}
\mathbf{K}_{1,1} & \dots & \mathbf{K}_{1,s-1} & \mathbf{K}_{1,s}\\
\vdots & \ddots & \vdots & \vdots\\
\mathbf{K}_{s-1,1} & \dots & \mathbf{K}_{s-1,s-1} & \mathbf{K}_{s-1,s}\\
\mathbf{K}_{s,1} & \dots & \mathbf{K}_{s,s-1} & \boxed{\mathbf{K}_{s,s}}
\end{array}\right).\label{eq:GG}
\end{equation}
The definitions above assume that the matrices
\begin{equation}
\left(\begin{array}{ccc}
\mathbf{K}_{1,1} & \dots & \mathbf{K}_{1,s}\\
\vdots & \ddots & \vdots\\
\mathbf{K}_{s,1} & \dots & \mathbf{K}_{s,s}
\end{array}\right)
\end{equation}
are invertible. This holds for the $Y^{+}(N)$ algebra, but for the
$Y^{-}(2n)$ algebra, the inverse exists only for even $s$. 

For odd $K$-matrices, there exists another recursion. We can introduce
another series of subalgebras
\[
\begin{array}{ccccccccc}
Y^{+}(2n) & \to & Y^{+}(2n-2) & \to & \dots & \to & Y^{+}(4) & \to & Y^{+}(2)\\
\mathbf{K}^{(1)}\equiv\mathbf{K} & \to & \mathbf{K}^{(3)} & \to & \dots & \to & \mathbf{K}^{(2n-3)} & \to & \mathbf{K}^{(2n-1)}\\
\downarrow &  & \downarrow &  & \dots &  & \downarrow &  & \downarrow\\
\mathbf{k}^{(1)} &  & \mathbf{k}^{(2)} &  & \dots &  & \mathbf{k}^{(n-1)} &  & \mathbf{k}^{(n)}
\end{array}
\]
In this case, the recurrence equation is
\begin{equation}
\mathbf{K}^{(2k+1)}=\omega_{2}\left(\mathbf{K}^{(2k-1)}\right),
\end{equation}
which explicitly gives

\begin{equation}
\mathbf{K}_{a,b}^{(2k+1)}(z):=\mathbf{K}_{a,b}^{(2k-1)}(z)-\sum_{\alpha,\beta=2k-1}^{2k}\mathbf{K}_{a,\alpha}^{(2k-1)}(z)\widehat{\mathbf{K}}_{\alpha,\beta}^{(2k-1)}(z)\mathbf{K}_{\beta,b}^{(2k-1)}(z),
\end{equation}
for $a,b=2k+1,\dots,2n$, where $\widehat{\mathbf{K}}_{\alpha,\beta}^{(2k-1)}(z)$
is the inverse of a $2\times2$ block
\begin{equation}
\sum_{\gamma=2k-1}^{2k}\mathbf{K}_{\alpha,\gamma}^{(2k-1)}(z)\widehat{\mathbf{K}}_{\gamma,\beta}^{(2k-1)}(z)=\delta_{\alpha,\beta}\mathbf{1}.
\end{equation}
We also define the $2\times2$ $K$-matrices as
\begin{equation}
\mathbf{k}_{\alpha,\beta}^{(k)}(u):=\mathbf{K}_{\alpha+2k-2,\beta+2k-2}^{(2k-1)}(u),\label{eq:defk}
\end{equation}
for $\alpha,\beta=1,2$.

The nested $K$-matrices $\mathbf{K}^{(2k-1)}$ are the same as before,
i.e., they satisfy the reflection equation (\ref{eq:refl}). The $2\times2$
matrices $\mathbf{k}_{\alpha,\beta}^{(k)}(u-(k-1))$ for $\alpha,\beta=1,2$
generate $Y^{+}(2)$ subalgebras. The element $\mathbf{k}_{1,1}^{(k)}=\mathbf{K}_{2k-1,2k-1}^{(2k-1)}$,
which by definition is the $\mathbf{G}^{(2k-1)}$ operator. To compute
$\mathbf{G}^{(2k)}(u)$ we can use the formula (\ref{eq:GG}) and
the identity $\omega_{2k-1}=\omega_{1}\circ\omega_{2k-2}$, that is,
\begin{align}
\mathbf{G}^{(2k)} & =\omega_{2k-1}\circ\left(\begin{array}{cccc}
\mathbf{K}_{1,1} & \dots & \mathbf{K}_{1,2k-1} & \mathbf{K}_{1,2k}\\
\vdots & \ddots & \vdots & \vdots\\
\mathbf{K}_{2k-1,1} & \dots & \mathbf{K}_{2k-1,2k-1} & \mathbf{K}_{2k-1,2k}\\
\mathbf{K}_{2k,1} & \dots & \mathbf{K}_{2k,2k-1} & \mathbf{K}_{2k,2k}
\end{array}\right)=\omega_{1}\circ\omega_{2k-2}\circ\left(\begin{array}{cccc}
\mathbf{K}_{1,1} & \dots & \mathbf{K}_{1,2k-1} & \mathbf{K}_{1,2k}\\
\vdots & \ddots & \vdots & \vdots\\
\mathbf{K}_{2k-1,1} & \dots & \mathbf{K}_{2k-1,2k-1} & \mathbf{K}_{2k-1,2k}\\
\mathbf{K}_{2k,1} & \dots & \mathbf{K}_{2k,2k-1} & \mathbf{K}_{2k,2k}
\end{array}\right)\nonumber \\
 & =\omega_{1}\circ\left(\begin{array}{cc}
\mathbf{K}_{2k-1,2k-1}^{(2k-1)} & \mathbf{K}_{2k-1,2k}^{(2k-1)}\\
\mathbf{K}_{2k,2k-1}^{(2k-1)} & \mathbf{K}_{2k,2k}^{(2k-1)}
\end{array}\right)=\left(\begin{array}{cc}
\mathbf{k}_{1,1}^{(k)} & \mathbf{k}_{1,2}^{(k)}\\
\mathbf{k}_{2,1}^{(k)} & \boxed{\mathbf{k}_{2,2}^{(k)}}
\end{array}\right).
\end{align}
In summary, the $G$-operators have an alternative expression
\begin{equation}
\begin{split}\mathbf{G}^{(2k-1)}(z)= & \mathbf{k}_{1,1}^{(k)}(z),\\
\mathbf{G}^{(2k)}(z)= & \mathbf{k}_{2,2}^{(k)}(z)-\mathbf{k}_{2,1}^{(k)}(z)\left[\mathbf{k}_{1,1}^{(k)}(z)\right]^{-1}\mathbf{k}_{1,2}^{(k)}(z).
\end{split}
\label{eq:GYm}
\end{equation}

\subsubsection{The $Y^{-}(2)$ overlaps\label{subsec:Ym2-overlaps}}

In the case of $Y^{-}(2)$ $K$-matrices, the asymptotic expansion
begins with
\begin{equation}
\mathbf{K}_{i,j}(u)=\epsilon_{i,j}\mathbf{1}+\mathcal{O}(u^{-1}),
\end{equation}
where $i,j=1,2$ and $\epsilon_{i,j}=-\epsilon_{j,i}$, $\epsilon_{1,2}=1$.
It is clear that the $\mathbf{K}_{1,1}(u)$ component is not invertible,
so the previous overlap formula cannot be applied. However, for finite-dimensional
irreducible $K$-matrices, there exists a continuous deformation $\mathbf{K}\to\tilde{\mathbf{K}}$,
that still solves the reflection equation and for which the $\tilde{\mathbf{K}}_{1,1}$
component is invertible. For this deformed $K$-matrix, the overlap
formula derived in the previous section can be applied. Taking the
zero limit of the deformation parameter yields the overlap formula
for the original $K$-matrix.

This deformation is based on the theorem: every finite-dimensional
irreducible representation of $Y^{-}(2)$ arises from a representation
of $Y(2)$ via the embedding $Y^{-}(2)\hookrightarrow Y(2)$ \cite{Molev:1997wp}.
That is, for every finite-dimensional irreducible $Y^{-}(2)$ $K$-matrix,
there exists a Lax operator (a $Y(2)$ representation) such that
\begin{equation}
\mathbf{K}_{i,j}(u)=\sum_{k,l}\mathbf{L}_{k,i}(u)\epsilon_{k,l}\mathbf{L}_{l,j}(-u).
\end{equation}
The $\mathbf{K}_{1,1}(u)$ component is 
\begin{equation}
\mathbf{K}_{1,1}(u)=\mathbf{L}_{1,1}(u)\mathbf{L}_{2,1}(-u)-\mathbf{L}_{2,1}(u)\mathbf{L}_{1,1}(-u).
\end{equation}
Since $Y(2)$ contains an $\mathfrak{sl}_{2}$ subalgebra, the boundary
space is also an $\mathfrak{sl}_{2}$ representation. The operator
$\mathbf{L}_{1,1}$ preserves the $\mathfrak{sl}_{2}$ quantum number,
while $\mathbf{L}_{2,1}$ lowers it by one, meaning that for finite-dimensional
representations, $\mathbf{K}_{1,1}(u)$ is a nilpotent operator.

There is a special scalar solution of the reflection equation ($d_{B}=1$)
\begin{equation}
K(u|\mathfrak{b})=\left(\begin{array}{cc}
\mathfrak{b}(u+1/2) & 1\\
-1 & 0
\end{array}\right),
\end{equation}
where $\mathfrak{b}\in\mathbb{C}$ is a scalar parameter. From this
scalar representation and the Lax operator $\mathbf{L}$, we can construct
a deformed $K$-matrix:
\begin{equation}
\tilde{\mathbf{K}}_{0}(u|\mathfrak{b})=\mathbf{L}_{0}^{t_{0}}(u)K_{0}(u|\mathfrak{b})\mathbf{L}_{0}(-u),
\end{equation}
that is
\begin{align}
\tilde{\mathbf{K}}_{i,j}(u|\mathfrak{b}) & =\sum_{k,l}\mathbf{L}_{k,i}(u)\epsilon_{k,l}\mathbf{L}_{l,j}(-u)+\mathfrak{b}(u+1/2)\mathbf{L}_{1,i}(u)\mathbf{L}_{1,j}(-u)\nonumber \\
 & =\mathbf{K}_{i,j}(u)+\mathfrak{b}(u+1/2)\mathbf{L}_{1,i}(u)\mathbf{L}_{1,j}(-u).
\end{align}
This matrix satisfies the reflection equation, the component $\tilde{\mathbf{K}}_{1,1}(u|\mathfrak{b})$
is invertible, and $\tilde{\mathbf{K}}_{i,j}(u|0)=\mathbf{K}_{i,j}(u)$. 

Let us examine the component $\tilde{\mathbf{K}}_{1,1}$, that is,
the first $G$-operator:
\[
\tilde{\mathbf{G}}^{(1)}(u|\mathfrak{b})=\tilde{\mathbf{K}}_{1,1}(u|\mathfrak{b})=\mathbf{K}_{1,1}(u)+\mathfrak{b}(u+1/2)\mathbf{L}_{1,1}(u)\mathbf{L}_{1,1}(-u).
\]
Since the operator $\mathbf{K}_{1,1}$ lowers the $\mathfrak{sl}_{2}$
spin and $\mathbf{L}_{1,1}(u)\mathbf{L}_{1,1}(-u)$ does not change
it, in a specific basis, the former is an upper triangular matrix
and the latter is a diagonal matrix, therefore the eigenvalues of
$\tilde{\mathbf{G}}^{(1)}$ coincide with those of $\mathfrak{b}(u+1/2)\mathbf{L}_{1,1}(u)\mathbf{L}_{1,1}(-u)$.

Let us continue with the second $G$-operator.
\begin{equation}
\tilde{\mathbf{G}}^{(2)}(u|\mathfrak{b})=\tilde{\mathbf{K}}_{2,2}(u|\mathfrak{b})-\tilde{\mathbf{K}}_{2,1}(u|\mathfrak{b})\left[\tilde{\mathbf{K}}_{1,1}(u|\mathfrak{b})\right]^{-1}\tilde{\mathbf{K}}_{1,2}(u|\mathfrak{b}).
\end{equation}
It can be shown that this operator can be expressed using the Sklyanin
determinant (see (\ref{eq:SklyaninMin}) for the definition of $\tilde{\mathbf{k}}_{1,2}^{1,2}$)
\begin{equation}
\tilde{\mathbf{k}}_{1,2}^{1,2}(u|\mathfrak{b})=\tilde{\mathbf{G}}^{(1)}(u+1|\mathfrak{b})\tilde{\mathbf{G}}^{(2)}(u|\mathfrak{b}).
\end{equation}
The Sklyanin determinant, in turn, factorizes
\begin{equation}
\tilde{\mathbf{k}}_{1,2}^{1,2}(u|\mathfrak{b})=\mathbf{l}_{1,2}^{1,2}(-u)\mathbf{l}_{1,2}^{1,2}(u+1)=\mathbf{k}_{1,2}^{1,2}(u),
\end{equation}
where the quantum determinant is defined as
\begin{equation}
R_{1,2}(-1)\mathbf{L}_{1}(u-1)\mathbf{L}_{2}(u)=\sum_{a_{i},b_{i}=1}^{2}e_{a_{1},b_{1}}\otimes e_{a_{2},b_{2}}\otimes\mathbf{l}_{a_{1},a_{2}}^{b_{1},b_{2}}(u).
\end{equation}
(The proof matches Theorem 2.5.3 in \cite{MolevBook}), meaning that
the Sklyanin determinant does not depend on the deformation parameter
$\mathfrak{b}$. Based on this, the dependence of the eigenvalues
of the $F$-operator 
\begin{equation}
\mathbf{F}^{(1)}(u|\mathfrak{b})=\left[\tilde{\mathbf{G}}^{(1)}(u|\mathfrak{b})\right]^{-1}\tilde{\mathbf{G}}^{(2)}(u|\mathfrak{b})
\end{equation}
on $\mathfrak{b}$ is extremely simple
\begin{equation}
\mathcal{F}_{\ell}(u|\mathfrak{b})=\mathfrak{b}^{-2}\mathcal{F}_{\ell}(u),
\end{equation}
where the functions $\mathcal{F}_{\ell}(u)$ are the eigenvalues of
the $F$-operator:
\begin{equation}
\mathbf{F}(u):=\left[\tilde{\mathbf{G}}^{(1)}(u|1)\right]^{-1}\tilde{\mathbf{G}}^{(2)}(u|1).\label{eq:FYm}
\end{equation}

The boundary state must also be deformed. The boundary state corresponding
to the deformed $K$-matrix, $\langle\tilde{\Psi}(\mathfrak{b})|$
can be obtained based on the description in Appendix \ref{sec:Boundary-states-gen-rep}.
The vacuum overlap in the deformed case is
\begin{equation}
\tilde{\mathbf{B}}(\mathfrak{b})=\langle\tilde{\Psi}(\mathfrak{b})|0\rangle=\prod_{j=1}^{J}\prod_{l=1}^{\Lambda_{1}^{(j)}-\Lambda_{2}^{(j)}}\tilde{\mathbf{K}}_{1,1}(\theta_{j}-l|\mathfrak{b})=\prod_{j=1}^{J}\prod_{l=1}^{k_{j}}\tilde{\mathbf{G}}^{(1)}(\theta_{j}-l|\mathfrak{b}).
\end{equation}
Since the eigenvalues of $\tilde{\mathbf{G}}^{(1)}$ are linear in
$\mathfrak{b}$, the dependence of the eigenvalues of the operator
$\tilde{\mathbf{B}}$ on $\mathfrak{b}$ is
\begin{equation}
\tilde{\beta}_{\ell}(\mathfrak{b})=\mathfrak{b}^{\frac{\Lambda_{1}-\Lambda_{2}}{2}}\beta_{\ell},
\end{equation}
where the $\beta_{\ell}$ are the eigenvalues of the operator
\begin{equation}
\mathbf{B}:=\langle\tilde{\Psi}(1)|0\rangle.\label{eq:BYm}
\end{equation}
Applying the overlap formula:
\begin{equation}
\frac{\langle\mathrm{MPS}(\mathfrak{b})|\mathbb{B}(\bar{t})}{\sqrt{\mathbb{C}(\bar{t})\mathbb{B}(\bar{t})}}=\mathfrak{b}^{\frac{\Lambda_{1}-\Lambda_{2}}{2}-r_{1}}\sum_{\ell=1}^{d_{B}}\beta_{\ell}\tilde{\mathcal{F}}_{\ell}(\bar{t}^{+})\sqrt{\frac{\det G^{+}}{\det G^{-}}}.
\end{equation}
Now we can take the $\mathfrak{b}\to0$ limit. The Bethe states exists
only when $2r_{1}\leq\Lambda_{1}-\Lambda_{2}$ therefore the limit
is non-vanishing only when 
\begin{equation}
r_{1}=\frac{\Lambda_{1}-\Lambda_{2}}{2},\label{eq:selrule}
\end{equation}
and the $Y^{-}(2)$ overlap is
\begin{equation}
\frac{\langle\mathrm{MPS}|\mathbb{B}(\bar{t})}{\sqrt{\mathbb{C}(\bar{t})\mathbb{B}(\bar{t})}}=\delta_{2r_{1},\Lambda_{1}-\Lambda_{2}}\sum_{\ell=1}^{d_{B}}\beta_{\ell}\tilde{\mathcal{F}}_{\ell}(\bar{t}^{+})\sqrt{\frac{\det G^{+}}{\det G^{-}}},
\end{equation}
where the $\beta_{\ell}$ and $\mathcal{F}_{\ell}$ are the eigenvalues
of the operators (\ref{eq:BYm}) and (\ref{eq:FYm}).

The selection rule (\ref{eq:selrule}) means that the $Y^{-}(2)$
MPSs are $\mathfrak{gl}_{2}$ singlet states.

\subsubsection{Proposal for $Y^{-}(2n)$ overlaps}

In this section, we combine the results of \ref{subsec:Quasi-determinants}
and \ref{subsec:Ym2-overlaps}. The definitions discussed in Section
\ref{sec:Derivations-of-the-overlaps} cannot be used in the $Y^{-}(2n)$
case because those definitions involve the inverse of the $\mathbf{K}_{1,1}$
matrix, which is not invertible in this case. However, in \ref{subsec:Quasi-determinants}
we saw that there exists another equivalent definition for which the
nested $K$-matrices can be defined without requiring $\mathbf{K}_{1,1}$
to be invertible.

Based on Subsection \ref{subsec:Quasi-determinants}, we have the
series of subalgebras
\[
\begin{array}{ccccccccc}
Y^{-}(2n) & \to & Y^{-}(2n-2) & \to & \dots & \to & Y^{-}(4) & \to & Y^{-}(2)\\
\mathbf{K}^{(1)}\equiv\mathbf{K} & \to & \mathbf{K}^{(3)} & \to & \dots & \to & \mathbf{K}^{(2n-3)} & \to & \mathbf{K}^{(2n-1)}\\
\downarrow &  & \downarrow &  & \dots &  & \downarrow &  & \downarrow\\
\mathbf{k}^{(1)} &  & \mathbf{k}^{(2)} &  & \dots &  & \mathbf{k}^{(n-1)} &  & \mathbf{k}^{(n)}
\end{array}
\]
Based on the definition in (\ref{eq:defk}) the $K$-matrices $\mathbf{k}^{(s)}(u-(s-1))$
are representations of the $Y^{-}(2)$. For these, the $G$-operators
defined by (\ref{eq:GYm}) do not exist. However, we can apply the
method described in \ref{subsec:Ym2-overlaps} i.e., we first deform
the $\mathbf{k}^{(s)}$ matrices so that the $G$-operators become
definable. The deformation is carried out as follows
\begin{equation}
\begin{split}\mathbf{k}_{i,j}^{(s)}(u) & =\sum_{k,l}\mathbf{L}_{k,i}^{(s)}(u+(s-1))\epsilon_{k,l}\mathbf{L}_{l,j}^{(s)}(-u-(s-1)),\\
\tilde{\mathbf{k}}_{i,j}^{(s)}(u|\mathfrak{b}_{s}) & =\mathbf{k}_{i,j}^{(s)}(u)+\mathfrak{b}_{s}(u+(s-1/2))\mathbf{L}_{1,i}(u+(s-1))\mathbf{L}_{1,j}(-u-(s-1)).
\end{split}
\label{eq:defYm}
\end{equation}
For these deformed matrices, we can define the $G$-operators as
\begin{equation}
\begin{split}\tilde{\mathbf{G}}^{(2s-1)}(u|\mathfrak{b}_{s}) & =\tilde{\mathbf{k}}_{1,1}^{(s)}(u|\mathfrak{b}_{s}),\\
\tilde{\mathbf{G}}^{(2s)}(u|\mathfrak{b}_{s}) & =\tilde{\mathbf{k}}_{2,2}^{(s)}(u|\mathfrak{b}_{s})-\tilde{\mathbf{k}}_{2,1}^{(s)}(u|\mathfrak{b}_{s})\left[\tilde{\mathbf{k}}_{1,1}^{(s)}(u|\mathfrak{b}_{s})\right]^{-1}\tilde{\mathbf{k}}_{1,2}^{(s)}(u|\mathfrak{b}_{s}).
\end{split}
\end{equation}
Based on the previous reasoning, the eigenvalues of the $\tilde{\mathbf{G}}^{(2s-1)}$
operators are proportional, while the eigenvalues of the $\tilde{\mathbf{G}}^{(2s)}$
operators are inversely proportional to $\mathfrak{b}_{s}$. Thus,
the dependence of the eigenvalues of the $F$-operators
\begin{equation}
\begin{split}\mathbf{F}^{(2s-1)}(u|\mathfrak{b}_{s}) & =\left[\tilde{\mathbf{G}}^{(2s-1)}(u|\mathfrak{b}_{s})\right]^{-1}\tilde{\mathbf{G}}^{(2s)}(u|\mathfrak{b}_{s}),\\
\mathbf{F}^{(2s)}(u|\mathfrak{b}_{s},\mathfrak{b}_{s+1}) & =\left[\tilde{\mathbf{G}}^{(2s)}(u|\mathfrak{b}_{s})\right]^{-1}\tilde{\mathbf{G}}^{(2s+1)}(u|\mathfrak{b}_{s+1}),
\end{split}
\end{equation}
on $\mathfrak{b}_{s}$ is as follows
\begin{equation}
\begin{split}\mathcal{F}_{\ell}^{(2s-1)}(u|\mathfrak{b}_{s})= & \mathfrak{b}_{s}^{-2}\mathcal{F}_{\ell}^{(2s-1)}(u),\\
\mathcal{F}_{\ell}^{(2s)}(u|\mathfrak{b}_{s},\mathfrak{b}_{s+1})= & \mathfrak{b}_{s}\mathfrak{b}_{s+1}\mathcal{F}_{\ell}^{(2s)}(u),
\end{split}
\end{equation}
where $\mathcal{F}_{\ell}^{(k)}$ are the eigenvalues of the $\mathbf{F}^{(k)}$
operators in the limit $\mathfrak{b}_{s}\to1$. 

The deformed vacuum overlap is given by the formula in (\ref{eq:vacuumOv})
meaning that the dependence of the vacuum eigenvalues on $\mathfrak{b}_{s}$
is
\begin{equation}
\tilde{\beta}_{\ell}(\mathfrak{b}_{1},\dots,\mathfrak{b}_{n})=\prod_{s=1}^{n}\mathfrak{b}_{s}^{\frac{\Lambda_{2s-1}-\Lambda_{2s}}{2}}\beta_{\ell},
\end{equation}
where $\beta_{\ell}$ is the eigenvalue of the vacuum overlap in the
limit $\mathfrak{b}_{s}\to1$. Applying the overlap formula:
\[
\frac{\langle\mathrm{MPS}|\mathbb{B}(\bar{t})}{\sqrt{\mathbb{C}(\bar{t})\mathbb{B}(\bar{t})}}=\lim_{\mathfrak{b}_{s}\to0}\prod_{s=1}^{n}\mathfrak{b}_{s}^{\frac{\Lambda_{2s-1}-\Lambda_{2s}}{2}-r_{2s-1}+\frac{r_{2s-2}+r_{2s}}{2}}\sum_{\ell=1}^{d_{B}}\beta_{\ell}\tilde{\mathcal{F}}_{\ell}(\bar{t}^{+})\sqrt{\frac{\det G^{+}}{\det G^{-}}}.
\]
Now we can take the limit $\mathfrak{b}_{s}\to0$. The limit is non-vanishing
only when 
\begin{equation}
r_{2s-1}=\frac{\Lambda_{2s-1}-\Lambda_{2s}+r_{2s-2}+r_{2s}}{2},\label{eq:selrule-1}
\end{equation}
for $s=1,\dots,n$ and the $Y^{-}(2n)$ overlap is
\begin{equation}
\frac{\langle\mathrm{MPS}|\mathbb{B}(\bar{t})}{\sqrt{\mathbb{C}(\bar{t})\mathbb{B}(\bar{t})}}=\sum_{\ell=1}^{d_{B}}\beta_{\ell}\tilde{\mathcal{F}}_{\ell}(\bar{t}^{+})\sqrt{\frac{\det G^{+}}{\det G^{-}}}.
\end{equation}

\subsection{Generalization for $\mathcal{B}(N,M)$}

In the case of uncrossed $\mathcal{B}(N,M)$ $K$-matrices, the nesting
procedure described in \ref{subsec:The-uncrossed-overlaps} cannot
be continued beyond step $M$. The nesting in this case proceeds as
follows
\[
\begin{array}{ccccccccc}
\mathcal{B}(N,M) & \to & \mathcal{B}(2N-2,M-1) & \to & \dots & \to & \mathcal{B}(2N-2M+2,1) & \to & \mathcal{B}(2N-2M,0)\\
\mathbf{K}^{(1)}\equiv\mathbf{K} & \to & \mathbf{K}^{(2)} & \to & \dots & \to & \mathbf{K}^{(M)} & \to & \mathbf{K}^{(M+1)}\\
\downarrow &  & \downarrow &  & \dots &  & \downarrow\\
\mathbf{G}^{(1)} &  & \mathbf{G}^{(2)} &  & \dots &  & \mathbf{G}^{(M)}
\end{array}
\]
It is useful to first examine the case $M=0$ separately and then
generalize to arbitrary $M$.

\subsubsection{The $\mathcal{B}(N,0)$ overlaps}

For the $\mathcal{B}(N,0)$ $K$-matrices, the asymptotic expansion
begins with
\begin{equation}
\mathbf{K}_{i,j}(u)=\delta_{i,j}\mathbf{1}+\mathcal{O}(u^{-1}),
\end{equation}
where $i,j=1,\dots,N$. It is clear that the component $\mathbf{K}_{N,1}(u)$
is not invertible, so the previously used overlap formula cannot be
applied. For $N=2$, the uncrossed $\mathcal{B}(2,0)$ algebra is
isomorphic to the crossed $Y^{-}(2)$ algebra. In this case, for finite-dimensional
irreducible $K$-matrices, there exists a continuous deformation $\mathbf{K}\to\tilde{\mathbf{K}}$,
that still solves the reflection equation and makes the matrix elements
in the overlap formula invertible. The following generalizes this
procedure to arbitrary $N$. 

The deformation is based on the assumption that for every finite-dimensional
$\mathcal{B}(N,0)$ $K$-matrix, there exists a Lax operator (a representation
of $Y(N)$) such that
\begin{equation}
\mathbf{K}_{i,j}(u)=\sum_{k=1}^{N}\mathbf{L}'_{i,k}(u)\mathbf{L}_{k,j}(-u),
\end{equation}
where
\begin{equation}
\sum_{k=1}^{N}\mathbf{L}'_{i,k}(u)\mathbf{L}_{k,j}(u)=\delta_{i,j}\mathbf{1}.
\end{equation}
This is certainly true for $N=2$, but for $N>2$, it remains a conjecture.
The component 
\begin{equation}
\mathbf{K}_{N,1}(u)=\sum_{k=1}^{N}\mathbf{L}'_{N,k}(u)\mathbf{L}_{k,1}(-u),
\end{equation}
is a nilpotent operator.

The reflection equation has a special scalar solution (with $d_{B}=1$)
\begin{equation}
K_{i,j}(u|\bar{\mathfrak{b}})=\delta_{i,j}+u\sum_{i=1}^{N/2}\mathfrak{b}_{i}e_{N+1-i,i}.
\end{equation}
From this scalar representation and the Lax operator $\mathbf{L}$,
we can construct a deformed $K$-matrix:
\begin{equation}
\tilde{\mathbf{K}}_{0}(u|\bar{\mathfrak{b}})=\mathbf{L}'_{0}(u)K_{0}(u|\bar{\mathfrak{b}})\mathbf{L}_{0}(-u),
\end{equation}
that is
\begin{equation}
\tilde{\mathbf{K}}_{i,j}(u|\bar{\mathfrak{b}})=\mathbf{K}_{i,j}(u)+u\sum_{k=1}^{N/2}\mathfrak{b}_{k}\mathbf{L}'_{i,N+1-k}(u)\mathbf{L}_{k,j}(-u).
\end{equation}
This matrix satisfies the reflection equation, the component $\tilde{\mathbf{K}}_{N,1}(u|\bar{\mathfrak{b}})$
is invertible, and $\tilde{\mathbf{K}}_{i,j}(u|0)=\mathbf{K}_{i,j}(u)$. 

Let us now examine the component $\tilde{\mathbf{K}}_{N,1}$, i.e.,
the first $G$-operator:
\begin{align}
\tilde{\mathbf{G}}^{(1)}(u|\bar{\mathfrak{b}}) & =\tilde{\mathbf{K}}_{N,1}(u|\bar{\mathfrak{b}})=\mathfrak{b}_{1}u\mathbf{L}'_{N,N}(u)\mathbf{L}_{1,1}(-u)+\mathbf{K}_{N,1}(u)+u\sum_{k=2}^{N/2}\mathfrak{b}_{k}\mathbf{L}'_{N,N+1-k}(u)\mathbf{L}_{k,1}(-u)\nonumber \\
 & =\mathfrak{b}_{1}u\mathbf{L}'_{N,N}(u)\mathbf{L}_{1,1}(-u)+\check{\mathbf{K}}_{N,1}(u).
\end{align}
The eigenvalues of $\tilde{\mathbf{G}}^{(1)}$ coincide with those
of $\mathfrak{b}_{1}u\mathbf{L}'_{N,N}(u)\mathbf{L}_{1,1}(-u)$. Let
us continue with the second $G$-operator
\begin{equation}
\tilde{\mathbf{G}}^{(2)}(u|\bar{\mathfrak{b}})=\tilde{\mathbf{K}}_{2,2}(u|\bar{\mathfrak{b}})-\tilde{\mathbf{K}}_{2,1}(u|\bar{\mathfrak{b}})\left[\tilde{\mathbf{K}}_{1,1}(u|\bar{\mathfrak{b}})\right]^{-1}\tilde{\mathbf{K}}_{1,2}(u|\bar{\mathfrak{b}}).
\end{equation}
It can be shown that this operator can be expressed using the quantum
minor (see (\ref{eq:qminor}))
\begin{equation}
\tilde{\mathbf{k}}_{\bar{1},\bar{2}}^{1,2}(u|\bar{\mathfrak{b}})=\tilde{\mathbf{G}}^{(1)}(u+1|\bar{\mathfrak{b}})\tilde{\mathbf{G}}^{(2)}(u|\bar{\mathfrak{b}}),
\end{equation}
where $\bar{i}=N+1-i$. The quantum minor, however, factorizes
\begin{equation}
\tilde{\mathbf{k}}_{\bar{1},\bar{2}}^{1,2}(u|\bar{\mathfrak{b}})=\mathfrak{b}_{1}\mathfrak{b}_{2}\mathbf{l}_{1,2}^{1,2}(-u)\hat{\mathbf{l}}_{\bar{1},\bar{2}}^{\bar{1},\bar{2}}(u),
\end{equation}
meaning that the eigenvalues of $\tilde{\mathbf{G}}^{(2)}$ are proportional
to $\mathfrak{b}_{2}$. This can be generalized: the eigenvalues of
$\tilde{\mathbf{G}}^{(k)}$ are proportional to $\mathfrak{b}_{k}$
for $k\leq N/2$. Based on this, the dependence of the eigenvalues
of the operators
\begin{equation}
\mathbf{F}^{(k)}(u|\bar{\mathfrak{b}})=\left[\tilde{\mathbf{G}}^{(k)}(u|\bar{\mathfrak{b}})\right]^{-1}\tilde{\mathbf{G}}^{(k+1)}(u|\bar{\mathfrak{b}})
\end{equation}
on $\mathfrak{b}$ is extremely simple for $k<N/2$:
\begin{equation}
\mathcal{F}_{\ell}^{(k)}(u|\bar{\mathfrak{b}})=\frac{\mathfrak{b}_{k+1}}{\mathfrak{b}_{k}}\mathcal{F}_{\ell}^{(k)}(u)
\end{equation}
where the functions $\mathcal{F}_{\ell}(u)$ are the eigenvalues of
the operator
\begin{equation}
\mathbf{F}^{(k)}(u):=\mathbf{F}^{(k)}(u|\bar{\mathfrak{b}})\Biggr|_{\mathfrak{b}_{k}=1}.\label{eq:FYm-1}
\end{equation}
If $N=2n+1$, then the eigenvalues of $\tilde{\mathbf{G}}^{(n+1)}(u|\bar{\mathfrak{b}})$
do not depend on the $\mathfrak{b}_{k}$ parameters. If $N=2n$, then
the eigenvalues of $\tilde{\mathbf{G}}^{(n+1)}(u|\bar{\mathfrak{b}})$
are inversely proportional to $\mathfrak{b}_{n}$, i.e.,
\begin{equation}
\mathcal{F}_{\ell}^{(n)}(u|\bar{\mathfrak{b}})=\begin{cases}
\frac{1}{\mathfrak{b}_{n}}\mathcal{F}_{\ell}^{(n)}(u), & \text{if }N=2n+1,\\
\frac{1}{\mathfrak{b}_{n}^{2}}\mathcal{F}_{\ell}^{(n)}(u), & \text{if }N=2n.
\end{cases}
\end{equation}

The boundary state must also be deformed. The boundary state $\langle\tilde{\Psi}(\mathfrak{b})|$
corresponding to the deformed $K$-matrix can be obtained based on
the description in Appendix \ref{sec:Boundary-states-gen-rep}. Since
the eigenvalues of $\tilde{\mathbf{G}}^{(k)}$ are linear in $\mathfrak{b}_{k}$
the dependence of the eigenvalues of the $\tilde{\mathbf{B}}$ operator
on $\mathfrak{b}$ is as follows:
\begin{equation}
\tilde{\beta}_{\ell}(\bar{\mathfrak{b}})=\prod_{k=1}^{n}\mathfrak{b}_{k}^{\Lambda_{k}}\beta_{\ell},
\end{equation}
where the $\beta_{\ell}$ are the eigenvalues of the operator
\begin{equation}
\mathbf{B}:=\langle\tilde{\Psi}(\bar{\mathfrak{b}})|0\rangle\Biggr|_{\mathfrak{b}_{k}=1}.\label{eq:BYm-1}
\end{equation}
Applying the overlap formula:
\begin{equation}
\frac{\langle\mathrm{MPS}(\bar{\mathfrak{b}})|\mathbb{B}(\bar{t})}{\sqrt{\mathbb{C}(\bar{t})\mathbb{B}(\bar{t})}}=\prod_{k=1}^{n}\mathfrak{b}_{k}^{\Lambda_{k}-r_{k}+r_{k-1}}\sum_{\ell=1}^{d_{B}}\beta_{\ell}\tilde{\mathcal{F}}_{\ell}(\bar{t}^{+})\sqrt{\frac{\det G^{+}}{\det G^{-}}}.
\end{equation}
Now we can take the $\mathfrak{b}\to0$ limit. The limit is non-vanishing
only when 
\begin{equation}
\Lambda_{k}=r_{k}-r_{k-1},
\end{equation}
or equivalently
\begin{equation}
r_{k}=\sum_{l=1}^{k}\Lambda_{l}.\label{eq:selRuleBN0}
\end{equation}
In summary, the $\mathcal{B}(N,0)$ overlap is
\begin{equation}
\frac{\langle\mathrm{MPS}|\mathbb{B}(\bar{t})}{\sqrt{\mathbb{C}(\bar{t})\mathbb{B}(\bar{t})}}=\sum_{\ell=1}^{d_{B}}\beta_{\ell}\prod_{\nu=1}^{n}\tilde{\mathcal{F}}_{\ell}^{(\nu)}(\bar{t}^{+,\nu})\sqrt{\frac{\det G^{+}}{\det G^{-}}},
\end{equation}
with the selection rule (\ref{eq:selRuleBN0}).

\subsubsection{General $M$}

For general $M$, we define the nested $K$-matrices and $G$-operators
according to the procedure described in \ref{sec:Derivations-of-the-overlaps},
continuing until we reach the matrix $\mathbf{K}^{(M+1)}$, which
is a representation of $\tilde{\mathcal{B}}(N-2M,0)$. This matrix
is then deformed as described above, and the $G$-operators are computed
from the deformed matrices, following the same procedure:
\[
\begin{array}{ccccccccccccc}
\mathcal{B}(N,M) & \to & \dots & \to & \mathcal{B}(2N-2M+2,1) & \to & \mathcal{B}(2N-2M,0)\\
\mathbf{K}^{(1)}\equiv\mathbf{K} & \to & \dots & \to & \mathbf{K}^{(M)} & \to & \mathbf{K}^{(M+1)} & \rightarrow & \tilde{\mathbf{K}}^{(M+1)} & \to & \tilde{\mathbf{K}}^{(M+2)} & \to & \dots\\
\downarrow &  &  &  & \downarrow &  &  &  & \downarrow &  & \downarrow\\
\mathbf{G}^{(1)} &  & \dots &  & \mathbf{G}^{(M)} &  &  &  & \mathbf{G}^{(M+1)} &  & \mathbf{G}^{(M+2)} &  & \dots
\end{array}
\]
Specifically, the deformed $K$-matrix is given by
\begin{equation}
\begin{split}\mathbf{K}_{i,j}^{(M+1)}(u) & =\sum_{l=M+1}^{N-M}\mathbf{L}'_{i,l}(u)\mathbf{L}_{l,j}(-u),\\
\tilde{\mathbf{K}}_{i,j}^{(M+1)}(u) & =\mathbf{K}_{i,j}^{(M+1)}(u)+u\sum_{l=M+1}^{n}\mathbf{L}'_{i,N+1-l}(u)\mathbf{L}_{l,j}(-u),
\end{split}
\end{equation}
where $i,j=M+1,\dots,N-M$. The G-operators $\mathbf{G}^{(s)}$ are
derived from the matrix $\tilde{\mathbf{K}}^{(M+1)}$ in the usual
way for $s\geq M+1$-re. Non-zero overlaps require the following selection
rule:
\begin{equation}
r_{s}=r_{M}+\sum_{l=M+1}^{s}\Lambda_{l},
\end{equation}
for $s=M+1,\dots,n$.

\section{Generalization for the orthogonal and symplectic spin chains\label{sec:Orthogonal}}

In this section, we generalize the overlap formulas to orthogonal
and symplectic spin chains. Here, $n$ denotes the rank of the symmetry
algebra, i.e., $\mathfrak{sp}_{N}=\mathfrak{sp}_{2n}$ and
\[
\mathfrak{so}_{N}=\begin{cases}
\mathfrak{so}_{2n}, & \text{if }N\text{ is even},\\
\mathfrak{so}_{2n+1}, & \text{if }N\text{ is odd}.
\end{cases}
\]
We introduce the following index sets
\begin{equation}
I=\begin{cases}
\left\{ -n,\dots,,-2,-1,1,2,\dots,n\right\} , & \text{for }N=2n,\\
\left\{ -n,\dots,,-2,-1,0,1,2,\dots,n\right\} , & \text{for }N=2n+1,
\end{cases}
\end{equation}
and in summations, our convention is
\begin{equation}
\sum_{j=-n}^{n}a_{j}\equiv\sum_{j\in I}a_{j}.
\end{equation}

Assuming that the universal form of the overlaps exists (\ref{eq:onOV-1})
for both the orthogonal and symplectic cases, we determine the unknown
one-particle overlap functions $\tilde{\mathcal{F}}_{\ell}^{(\nu)}(u)$
for each node in the above subsections.

\subsection{$RTT$- and $KT$-relations}

The $R$-matrix for orthogonal and symplectic spin chains is given
by:
\begin{equation}
R(u)=\mathbf{1}+\frac{1}{u}\mathbf{P}-\frac{1}{u+\kappa_{N}}\mathbf{Q}\in\mathrm{End}(\mathbb{C}^{N}\otimes\mathbb{C}^{N}),
\end{equation}
where
\begin{equation}
\mathbf{P}=\sum_{i,j=-n}^{n}e_{i,j}\otimes e_{j,i},\quad\mathbf{Q}=\sum_{i,j=-n}^{n}\theta_{i}\theta_{j}e_{i,j}\otimes e_{-i,-j},
\end{equation}
and $\kappa_{N}=\frac{N\mp2}{2}$, where the minus and plus signs
correspond to the orthogonal and symplectic cases, respectively. In
the orthogonal case, $\theta_{i}=+1$ for all $i$, and in the symplectic
case, $\theta_{i}=\mathrm{sgn}(i)$. The $RTT$-relation defines the
Yangian algebras $Y(\mathfrak{g}_{N})$
\begin{equation}
R_{1,2}(u-v)T_{1}(u)T_{2}(v)=T_{2}(v)T_{1}(u)R_{1,2}(u-v).
\end{equation}
From this relation, it also follows that
\begin{equation}
\sum_{k=-n}^{n}\theta_{j}\theta_{k}T_{i,k}(u)T_{-j,-k}(u-\kappa_{N})=\delta_{i,j}\gamma(u),\label{eq:inv}
\end{equation}
where $\gamma(u)$ is a central element of the $Y(\mathfrak{g}_{N})$
algebra. This also implies that for the $\mathfrak{g}_{N}$ symmetric
spin chains, the crossed $KT$-relation is equivalent to the uncrossed
one, since $\widehat{T}(u)\sim T(u-\kappa_{N})$. 

The highest weight representations can be defined analogously as before,
meaning there exists a highest weight vector for which
\begin{equation}
\begin{split}T_{i,j}(u)\text{|}0\rangle & =0,\quad i>j,\\
T_{i,i}(u)\text{|}0\rangle & =\lambda_{i}(u)\text{|}0\rangle.
\end{split}
\end{equation}
Just like in the $\mathfrak{gl}_{N}$ case, the algebra $Y(\mathfrak{g}_{N})$
also contains a $\mathfrak{g}_{N}$ subalgebra, and the pseudo-vacuum
is the highest weight vector of this $\mathfrak{g}_{N}$ subalgebra
as well. These $\mathfrak{gl}_{N}$ weights can be obtained from the
asymptotic limit of the pseudo-vacuum eigenvalues
\begin{equation}
\lambda_{-i}(u)=1+\frac{1}{u}\Lambda_{-i+n+1}+\mathcal{O}(u^{-2}),
\end{equation}
for $i=1,\dots,n$.

For these representations, we can apply the inversion relation (\ref{eq:inv})
on the pseudo-vacuum 
\begin{equation}
\sum_{k=-n}^{n}\theta_{-n}\theta_{-k}T_{-n,-k}(u)T_{n,k}(u-\kappa_{N})\text{|}0\rangle=T_{-n,-n}(u)T_{n,n}(u-\kappa_{N})\text{|}0\rangle=\gamma(u)\text{|}0\rangle,
\end{equation}
i.e.,
\begin{equation}
\gamma(u)=\lambda_{-n}(u)\lambda_{n}(u-\kappa_{N}).
\end{equation}

The Bethe Ansatz equations are the following
\begin{equation}
\tilde{\alpha}_{a}(u_{j}^{a})=-\prod_{b=1}^{n}\prod_{k=1}^{r_{b}}\frac{u_{j}^{a}-u_{k}^{b}-\frac{i}{2}C_{a,b}}{u_{j}^{a}-u_{k}^{b}+\frac{i}{2}C_{a,b}},
\end{equation}
where $C_{a,b}$ is the symmetric Cartan matrix, $C_{a,b}=\vec{\rho}_{a}\cdot\vec{\rho}_{b}$,
with $\vec{\rho}_{a}\in\mathbb{R}^{n}$ being the simple roots for
$a=1,\dots,n$. In both the orthogonal and symplectic cases,
\begin{equation}
\rho_{a}=(\underbrace{0,\dots,0}_{a-1},1,-1,\underbrace{0,\dots,0}_{n-a-1}),
\end{equation}
for $a=1,\dots,n-1$, and the remaining root is
\begin{equation}
\rho_{n}=\begin{cases}
(0,\dots,0,0,1), & \text{for }\mathfrak{so}_{2n+1},\\
(0,\dots,0,0,2), & \text{for }\mathfrak{sp}_{2n},\\
(0,\dots,0,1,1), & \text{for }\mathfrak{so}_{2n}.
\end{cases}
\end{equation}

As previously established, for orthogonal and symplectic spin chains,
there is only one type of $KT$-relation. For simplicity, we use the
uncrossed convention:
\begin{equation}
\mathbf{K}_{0}(u)\langle\Psi|T_{0}(u)=\langle\Psi|T_{0}(-u)\mathbf{K}_{0}(u).
\end{equation}
From the $KT$-relation, it also follows that
\[
\lambda_{k}(u)=\lambda_{-k}(-u).
\]
As before, the $KT$-relation is compatible with associativity and
the $RTT$-relation if the $K$-matrix satisfies the reflection equation

\begin{equation}
R(v-u)\mathbf{K}_{1}(u)R(-u-v)\mathbf{K}_{2}(v)=\mathbf{K}_{2}(v)R(-u-v)\mathbf{K}_{1}(u)R(v-u).
\end{equation}

The $K$-matrices can be classified from the asymptotic limit of the
reflection equation. Using previous calculations, the asymptotic limit
of the $K$-matrix is
\begin{equation}
\mathbf{K}_{i,j}(u)=\mathcal{U}_{i,j}\otimes\mathbf{1}+\mathcal{O}(u^{-1}),
\end{equation}
where
\[
\mathcal{U}_{i,j}=\pm\theta_{i}\theta_{j}\mathcal{U}_{-j,-i},\quad\sum_{j=-n}^{n}\mathcal{U}_{i,j}\mathcal{U}_{j,k}=a\delta_{i,k},
\]
where $a\in\mathbb{C}$. If $a=0$, then the $K$-matrix is called
singular. If $a\neq0$, then the $K$-matrix can be normalized such
that $a=1$. From now on, we deal with non-singular $K$-matrices,
i.e., where $a=1$. From the second condition, it follows that the
eigenvalues of the $\mathcal{U}$ matrix are $\pm1$. If $\mathcal{U}_{i,j}=+\theta_{i}\theta_{j}\mathcal{U}_{-j,-i}$
then the symmetry algebras of the possible $\mathcal{U}$ matrices
are $\mathfrak{g}_{M}\oplus\mathfrak{g}_{N-M}$. The corresponding
$K$-matrices generate the reflection algebra $Y(\mathfrak{g}_{N},\mathfrak{g}_{M}\oplus\mathfrak{g}_{N-M})$
\cite{Guay_2016}. If $\mathcal{U}_{i,j}=-\theta_{i}\theta_{j}\mathcal{U}_{-j,-i}$
then $N$ is always even, and the symmetry algebra of the possible
$\mathcal{U}$ matrices is $\mathfrak{gl}_{n}$, with the corresponding
reflection algebra $Y(\mathfrak{g}_{2n},\mathfrak{gl}_{n})$ \cite{Guay_2016}.

In the following, we present examples of these algebras. In the case
of $Y(\mathfrak{so}_{N},\mathfrak{so}_{M}\oplus\mathfrak{so}_{N-M})$,
\begin{equation}
\mathcal{U}=\sum_{j=n-M+1}^{n}(e_{j,-j}+e_{-j,j})+\sum_{j=-n+M}^{n-M}e_{j,j}.\label{eq:soNsoM}
\end{equation}
In the case of $Y(\mathfrak{so}_{2n},\mathfrak{gl}_{n})$,
\begin{equation}
\mathcal{U}=\sum_{j=1}^{n/2}x_{j}(e_{n-2j+1,2j-2-n}-e_{n-2j+2,2j-1-n})+\sum_{j=1}^{n}(e_{j,j}-e_{-j,-j}).\label{eq:sogl}
\end{equation}
In the case of $Y(\mathfrak{sp}_{2n},\mathfrak{sp}_{2m}\oplus\mathfrak{sp}_{2n-2m})$,
\begin{align}
\mathcal{U} & =\sum_{j=1}^{m}x_{j}(e_{n-2j+1,2j-2-n}-e_{n-2j+2,2j-1-n})+\sum_{j=-n}^{n}\mathfrak{s}_{j}e_{j,j},\\
\mathfrak{s}_{j} & =\begin{cases}
+1, & |j|\leq n-2m,\\
(-1)^{n-j}, & |j|>n-2m.
\end{cases}\nonumber 
\end{align}
In the case of $Y(\mathfrak{sp}_{2n},\mathfrak{gl}_{n})$,
\begin{equation}
\mathcal{U}=\sum_{j=1}^{n}x_{j}e_{j,-j}+\sum_{j=1}^{n}(e_{j,j}-e_{-j,-j}).\label{eq:spgl}
\end{equation}

For better clarity, we provide a few explicit examples. In the case
of $Y(\mathfrak{so}_{4},\mathfrak{gl}_{2})$,
\begin{equation}
\mathcal{U}=\left(\begin{array}{cc|cc}
-1 & 0 & 0 & 0\\
0 & -1 & 0 & 0\\
\hline x_{1} & 0 & +1 & 0\\
0 & -x_{1} & 0 & +1
\end{array}\right),\label{eq:Uso4gl2}
\end{equation}
and in the case of $Y(\mathfrak{so}_{6},\mathfrak{gl}_{3})$,
\begin{equation}
\mathcal{U}=\left(\begin{array}{ccc|ccc}
-1 & 0 & 0 & 0 & 0 & 0\\
0 & -1 & 0 & 0 & 0 & 0\\
0 & 0 & -1 & 0 & 0 & 0\\
\hline 0 & 0 & 0 & +1 & 0 & 0\\
x_{1} & 0 & 0 & 0 & +1 & 0\\
0 & -x_{1} & 0 & 0 & 0 & +1
\end{array}\right),
\end{equation}
and in the case of $Y(\mathfrak{so}_{8},\mathfrak{gl}_{4})$,
\begin{equation}
\mathcal{U}=\left(\begin{array}{cccc|cccc}
-1 & 0 & 0 & 0 & 0 & 0 & 0 & 0\\
0 & -1 & 0 & 0 & 0 & 0 & 0 & 0\\
0 & 0 & -1 & 0 & 0 & 0 & 0 & 0\\
0 & 0 & 0 & -1 & 0 & 0 & 0 & 0\\
\hline 0 & 0 & x_{2} & 0 & +1 & 0 & 0 & 0\\
0 & 0 & 0 & -x_{2} & 0 & +1 & 0 & 0\\
x_{1} & 0 & 0 & 0 & 0 & 0 & +1 & 0\\
0 & -x_{1} & 0 & 0 & 0 & 0 & 0 & +1
\end{array}\right).
\end{equation}
In the case of $Y(\mathfrak{sp}_{8},\mathfrak{sp}_{2}\oplus\mathfrak{sp}_{6})$,
\begin{equation}
\mathcal{U}=\left(\begin{array}{cccc|cccc}
-1 & 0 & 0 & 0 & 0 & 0 & 0 & 0\\
0 & +1 & 0 & 0 & 0 & 0 & 0 & 0\\
0 & 0 & +1 & 0 & 0 & 0 & 0 & 0\\
0 & 0 & 0 & +1 & 0 & 0 & 0 & 0\\
\hline 0 & 0 & 0 & 0 & +1 & 0 & 0 & 0\\
0 & 0 & 0 & 0 & 0 & +1 & 0 & 0\\
x_{1} & 0 & 0 & 0 & 0 & 0 & +1 & 0\\
0 & -x_{1} & 0 & 0 & 0 & 0 & 0 & -1
\end{array}\right),
\end{equation}
and in the case of $Y(\mathfrak{sp}_{8},\mathfrak{sp}_{4}\oplus\mathfrak{sp}_{4})$,
\begin{equation}
\mathcal{U}=\left(\begin{array}{cccc|cccc}
-1 & 0 & 0 & 0 & 0 & 0 & 0 & 0\\
0 & +1 & 0 & 0 & 0 & 0 & 0 & 0\\
0 & 0 & -1 & 0 & 0 & 0 & 0 & 0\\
0 & 0 & 0 & +1 & 0 & 0 & 0 & 0\\
\hline 0 & 0 & x_{2} & 0 & +1 & 0 & 0 & 0\\
0 & 0 & 0 & -x_{2} & 0 & -1 & 0 & 0\\
x_{1} & 0 & 0 & 0 & 0 & 0 & +1 & 0\\
0 & -x_{1} & 0 & 0 & 0 & 0 & 0 & -1
\end{array}\right).
\end{equation}

\subsection{The $Y(n)$ subalgebra}

The algebra $Y(\mathfrak{g}_{N})$ contains a subalgebra $Y(n)$,
and the corresponding generators are denoted
\begin{equation}
T_{n+1-i,n+1-j}^{\mathfrak{gl}_{n}}(u)=T_{-i,-j}(u+\frac{\kappa_{N}}{2}),\label{eq:Tgl}
\end{equation}
where $i,j=1,\dots,n$. From this, it also follows that
\begin{equation}
\lambda_{1}^{\mathfrak{gl}_{n}}(u)=\lambda_{-n}(u+\frac{\kappa_{N}}{2})=\lambda_{n}(-u-\frac{\kappa_{N}}{2}).\label{eq:lamgl}
\end{equation}
The Bethe vectors corresponding to this subalgebra are those for which
$r_{n}=0$, meaning that in this sector, the Bethe vectors are generated
by the generators $T_{-i,-j}$ for $1\leq j<i\leq n$, i.e,
\begin{equation}
\mathbb{B}(\left\{ \bar{t}^{\nu}\right\} _{\nu=1}^{n-1},\emptyset)=\mathcal{P}(T_{-i<-j})|0\rangle,
\end{equation}
where $\mathcal{P}(T_{-i<-j})$ is a polynomial in the generators
$T_{-i,-j}$ for $1\leq j<i\leq n$. It follows that the Bethe vectors
in the $\mathfrak{gl}_{n}$ sector are annihilated by the generators
$T_{i,-j}$ if $0\leq i,j\leq n$ and $i\neq-j$. Based on this, we
introduce an equivalence relation in the $\mathfrak{gl}_{n}$ sector:
\begin{equation}
A\cong B,\quad\Longleftrightarrow\quad T_{i,-j}(A-B)=0,
\end{equation}
for some $0\leq i,j\leq n$ if $i\neq-j$. The inversion relation
(\ref{eq:inv}) simplifies in the $\mathfrak{gl}_{n}$ sector 
\begin{equation}
\sum_{k=1}^{n}T_{-i,-k}(u+\frac{\kappa_{N}}{2})T_{j,k}(u-\frac{\kappa_{N}}{2})\cong\delta_{i,j}\lambda_{-n}(u+\frac{\kappa_{N}}{2})\lambda_{n}(u-\frac{\kappa_{N}}{2}),
\end{equation}
when $i,j=1,\dots,n$. Here, we used that $\theta_{-j}\theta_{-k}=1$
for $i,j=1,\dots,n$. Using the correspondences (\ref{eq:inv}) and
(\ref{eq:lamgl}) we get
\begin{equation}
\sum_{k=1}^{n}T_{i,k}^{\mathfrak{gl}_{n}}(u)T_{n+1-j,n+1-k}(u-\frac{\kappa_{N}}{2})\cong\delta_{i,j}\lambda_{1}^{\mathfrak{gl}_{n}}(u)\lambda_{1}^{\mathfrak{gl}_{n}}(-u).
\end{equation}
We can see that this matches the inversion relation (\ref{eq:Thatdef})
used for the $Y(n)$ algebras , i.e.,
\begin{equation}
\widehat{T}_{j,k}^{\mathfrak{gl}_{n}}(u)=T_{n+1-j,n+1-k}(u-\frac{\kappa_{N}}{2}).
\end{equation}
The $KT$-relation in the $\mathfrak{gl}_{n}$ sector also simplifies
accordingly
\begin{equation}
\sum_{k=1}^{n}\mathbf{K}_{n+1-i,k-n-1}(u+\frac{\kappa_{N}}{2})\langle\Psi|T_{k,j}^{\mathfrak{gl}_{n}}(u)\cong\sum_{k=1}^{n}\langle\Psi|\widehat{T}_{i,k}^{\mathfrak{gl}_{n}}(-u)\mathbf{K}_{n+1-k,j-n-1}(u+\frac{\kappa_{N}}{2}).
\end{equation}
This is the $KT$-relation for the $\mathfrak{gl}_{n}$ spin chain,
and the $K$-matrix in the $\mathfrak{gl}_{n}$ sector under the $\mathfrak{gl}_{n}$
convention is 
\begin{equation}
\mathbf{K}_{i,j}^{\mathfrak{gl}_{n}}(u)=\mathbf{K}_{n+1-i,j-n-1}(u+\frac{\kappa_{N}}{2}),\label{eq:Kcon}
\end{equation}
where $i,j=1,\dots,n$.

Based on this, we can determine the relationship between the one-particle
overlap functions $\tilde{\mathcal{F}}_{k}^{\mathfrak{g}_{N},(s)}$
and $\tilde{\mathcal{F}}_{k}^{\mathfrak{gl}_{n},(s)}$ for $s=1,\dots,n-1$.
For this, we use the formula of the $\mathfrak{gl}_{n}$ $G$-operators
(\ref{eq:GG}). Based on this, in the $\mathfrak{g}_{N}$ convention,
the $G$-operators are
\begin{equation}
\mathbf{G}^{(s)}=\left(\begin{array}{cccc}
\mathbf{K}_{n,-n} & \dots & \mathbf{K}_{n,-n+s-2} & \mathbf{K}_{n,-n+s-1}\\
\vdots & \ddots & \vdots & \vdots\\
\mathbf{K}_{n-s+2,-n} & \dots & \mathbf{K}_{n-s+2,-n+s-2} & \mathbf{K}_{n-s+2,-n+s-1}\\
\mathbf{K}_{n-s+1,-n} & \dots & \mathbf{K}_{n-s+1,-n+s-2} & \boxed{\mathbf{K}_{n-s+1,-n+s-1}}
\end{array}\right).\label{eq:Ggn}
\end{equation}
Using the correspondence (\ref{eq:Kcon}), we can easily see that
\begin{equation}
\mathbf{G}^{\mathfrak{gl}_{n},(s)}(u)=\mathbf{G}^{(s)}(u+\frac{\kappa_{N}}{2}),
\end{equation}
for $s=1,\dots,n$. Based on this, we can define the $\mathfrak{g}_{N}$
$F$-operators for $s=1,\dots,n-1$ in the usual way
\begin{equation}
\mathbf{F}^{(s)}(u):=\left[\mathbf{G}^{(s)}(u)\right]^{-1}\mathbf{G}^{(s+1)}(u)=\mathbf{F}^{\mathfrak{gl}_{n},(s)}(u-\frac{\kappa_{N}}{2}).
\end{equation}
The one-particle overlap functions in the $\mathfrak{\ensuremath{gl}_{n}}$
sector are given by (\ref{eq:Ftdef})
\begin{equation}
\tilde{\mathcal{F}}_{\ell}^{(s)}(u)=\mathcal{F}_{\ell}^{\mathfrak{gl}_{n},(s)}(iu-s/2)\sqrt{\frac{u^{2}}{u^{2}+1/4}}=\mathcal{F}_{\ell}^{(s)}(iu+\frac{\kappa_{N}}{2}-s/2)\sqrt{\frac{u^{2}}{u^{2}+1/4}},
\end{equation}
for $s=1,\dots,n-1$.

\subsection{Nested $\mathfrak{g}_{N}$ $K$-matrices}

Using equation (\ref{eq:Kmqdet}) we can define nested $K$-matrices
in the $\mathfrak{\ensuremath{gl}_{n}}$ sector. Rewriting the quasi-determinant
in the $\mathfrak{g}_{N}$ convention, we obtain
\begin{equation}
\mathbf{K}_{i,j}^{(s)}=\left(\begin{array}{cccc}
\mathbf{K}_{n,-n} & \dots & \mathbf{K}_{n,-n+s-2} & \mathbf{K}_{n,j}\\
\vdots & \ddots & \vdots & \vdots\\
\mathbf{K}_{n-s+2,-n} & \dots & \mathbf{K}_{n-s+2,-n+s-2} & \mathbf{K}_{n-s+2,j}\\
\mathbf{K}_{i,-n} & \dots & \mathbf{K}_{i,-n+s-2} & \boxed{\mathbf{K}_{i,j}}
\end{array}\right).\label{eq:KembgN}
\end{equation}
This defines a $\mathfrak{gl}_{n-s+1}$ $K$-matrix for $i,-j=1,\dots,n-s+1$,
according to the correspondence (\ref{eq:Kcon}). This algebra embedding
can be extended to the full $\mathfrak{g}_{N}$ algebra, meaning that
(\ref{eq:KembgN}) defines a $\mathfrak{g}_{N-2s+2}$ $K$-matrix
for $i,j=-n+s-1,\dots,n-s+1$. The proof can be carried out based
on \cite{Jing_2018} or the proof of Theorem \ref{thm:uncrossed-nestedK}.

There are recursive definitions equivalent to quasi-determinants for
nested $K$-matrices. One such option is the following:
\begin{equation}
\mathbf{K}_{a,b}^{(s+1)}(u)=\mathbf{K}_{a,b}^{(s)}(u)-\mathbf{K}_{a,-n+s-1}^{(s)}(u)\left[\mathbf{K}_{n+1-s,-n+s-1}^{(s)}(u)\right]^{-1}\mathbf{K}_{n+1-s,b}^{(s)}(u),\label{eq:soNrec0}
\end{equation}
where $a,b=-n+s,\dots,n-s$. The $G$-operators can be expressed in
the following way.
\begin{equation}
\mathbf{G}^{(s)}(u)=\mathbf{K}_{n+1-s,-n-1+s}^{(s)}(u),
\end{equation}
for $s=1,\dots,n-1$. 

According to section \ref{subsec:Quasi-determinants}, another recursive
definition for odd-indexed $K$-matrices is:
\begin{equation}
\mathbf{K}_{a,b}^{(2s+1)}(u)=\mathbf{K}_{a,b}^{(2s-1)}(u)-\sum_{\alpha,\beta=n-2s+1}^{n-2s+2}\mathbf{K}_{a,-\alpha}^{(2s-1)}(u)\widehat{\mathbf{K}}_{-\alpha,\beta}^{(2s-1)}(u)\mathbf{K}_{a,\beta}^{(2s-1)}(u),\label{eq:soNrec}
\end{equation}
where $\widehat{\mathbf{K}}$ is a $2\times2$ inverse matrix defined
as follows
\[
\sum_{\alpha,\beta=n-2s+1}^{n-2s+2}\widehat{\mathbf{K}}_{-\alpha,\beta}^{(2s-1)}(u)\mathbf{K}_{\beta,-\gamma}^{(2s-1)}(u)=\delta_{\alpha,\gamma}.
\]
These matrices $\mathbf{K}^{(2s-1)}$ exist for $s=1,\dots,\left\lfloor \frac{n}{2}\right\rfloor $.
We can select $k=\left\lfloor \frac{n}{2}\right\rfloor $ number of
$\mathfrak{gl}_{2}$ $K$-matrices
\begin{equation}
\mathbf{k}_{\alpha,\beta}^{(s)}(u):=\mathbf{K}_{n-2s+3-\alpha,-n+2s-3+\beta}^{(2s-1)}(u),\label{eq:Ym2k}
\end{equation}
for $\alpha,\beta=1,2$. Using these, the $G$-operators can be defined
as
\begin{equation}
\begin{split}\mathbf{G}^{(2s-1)}(u) & =\mathbf{k}_{1,1}^{(s)}(u),\\
\mathbf{G}^{(2s)}(u) & =\mathbf{k}_{2,2}^{(s)}(u)-\mathbf{k}_{2,1}^{(s)}(u)\left[\mathbf{k}_{1,1}^{(s)}(u)\right]^{-1}\mathbf{k}_{1,2}^{(s)}(u),
\end{split}
\label{eq:Gop0}
\end{equation}
for $s=1,\dots,k$.

We can examine the leading order of nested $K$-matrices in the asymptotic
limit
\begin{equation}
\mathbf{K}_{i,j}^{(s)}(u)=\mathcal{U}_{i,j}^{(s)}\mathbf{1}+\mathcal{O}(u^{-1}).
\end{equation}
From the recursive equations (\ref{eq:soNrec0}), it is easily derived
that
\begin{equation}
\begin{split}\sum_{k=-n-1+s}^{n+1-s}\mathcal{U}_{i,k}^{(s)}\mathcal{U}_{k,l}^{(s)} & =\sum_{k=-n}^{n}\mathcal{U}_{i,k}\mathcal{U}_{k,l},\\
\sum_{k=-n-1+s}^{n+1-s}\mathcal{U}_{k,k}^{(s)} & =\sum_{k=-n}^{n}\mathcal{U}_{k,k}.
\end{split}
\label{eq:Urec-1}
\end{equation}
We can also see that for the algebras $Y(\mathfrak{so}_{2n},\mathfrak{so}_{M}\oplus\mathfrak{so}_{2n-M})$
and $Y(\mathfrak{sp}_{2n},\mathfrak{gl}_{n})$, the operator $\mathbf{K}_{n,-n}^{-1}$
exists, meaning that the first type of recursion (\ref{eq:soNrec0})
can be applied. This follows from the asymptotic limit of the $K$-matrices
(\ref{eq:soNsoM}), (\ref{eq:spgl}). Since in these cases $\mathcal{U}_{n,-n}\neq0$,
the inverse of the operator $\mathbf{K}_{n,-n}(u)$ can be defined
order by order in the expansion in $u^{-1}$. Based on this, we obtain
the following recursions for these reflection algebras.

\[
\begin{array}{ccccccc}
Y(\mathfrak{so}_{2n},\mathfrak{so}_{M}\oplus\mathfrak{so}_{2n-M}) & \to & Y(\mathfrak{so}_{2n-2},\mathfrak{so}_{M-1}\oplus\mathfrak{so}_{2n-M-1}) & \to & \dots & \to & Y(\mathfrak{so}_{2n-2M},\mathfrak{so}_{2n-2M})\\
Y(\mathfrak{sp}_{2n},\mathfrak{gl}_{n}) & \to & Y(\mathfrak{sp}_{2n-2},\mathfrak{gl}_{n-1}) & \to & \dots & \to & Y(\mathfrak{sp}_{2},\mathfrak{gl}_{1})
\end{array}
\]
In the case of $\mathfrak{so}_{2n}$, the recursion can be terminated
when $M=n$ or $M=n-1$. In these cases, we have
\[
\begin{array}{ccccccc}
Y(\mathfrak{so}_{2n},\mathfrak{so}_{n}\oplus\mathfrak{so}_{n}) & \to & Y(\mathfrak{so}_{2n-2},\mathfrak{so}_{n-1}\oplus\mathfrak{so}_{n-1}) & \to & \dots & \to & Y(\mathfrak{so}_{4},\mathfrak{so}_{2}\oplus\mathfrak{so}_{2})\\
Y(\mathfrak{so}_{2n},\mathfrak{so}_{n-1}\oplus\mathfrak{so}_{n+1}) & \to & Y(\mathfrak{so}_{2n-2},\mathfrak{so}_{n-2}\oplus\mathfrak{so}_{n}) & \to & \dots & \to & Y(\mathfrak{so}_{4},\mathfrak{so}_{3})
\end{array}
\]

We observe that for the algebras $Y(\mathfrak{so}_{2n},\mathfrak{gl}_{n})$
and $Y(\mathfrak{sp}_{2n},\mathfrak{sp}_{2m}\oplus\mathfrak{sp}_{2n-2m})$
the inverse $\mathbf{K}_{n,-n}^{-1}$ does not exist, so the first
recursion is not applicable. However, the second recursion (\ref{eq:soNrec})
is applicable. Based on this, we obtain the following recursions for
these reflection algebras.
\[
\begin{array}{ccccccc}
Y(\mathfrak{so}_{4k},\mathfrak{gl}_{2k}) & \to & Y(\mathfrak{so}_{4k-4},\mathfrak{gl}_{2k-2}) & \to & \dots & \to & Y(\mathfrak{so}_{4},\mathfrak{gl}_{2})\\
Y(\mathfrak{so}_{4k+2},\mathfrak{gl}_{2k+1}) & \to & Y(\mathfrak{so}_{4k-2},\mathfrak{gl}_{2k-1}) & \to & \dots & \to & Y(\mathfrak{so}_{6},\mathfrak{gl}_{3})\\
Y(\mathfrak{sp}_{2n},\mathfrak{sp}_{2m}\oplus\mathfrak{sp}_{2n-2m}) & \to & Y(\mathfrak{sp}_{2n-4},\mathfrak{sp}_{2m-2}\oplus\mathfrak{sp}_{2n-2m-2}) & \to & \dots & \to & Y(\mathfrak{sp}_{2n-4m},\mathfrak{sp}_{2n-4m})
\end{array}
\]

\subsection{Small rank algebras}

In this subsection, we determine the one-particle overlaps in low-rank
cases based on the correspondence $\mathfrak{sp}_{2}\cong\mathfrak{sl}_{2}$,
$\mathfrak{so}_{3}\cong\mathfrak{sl}_{2}$, $\mathfrak{so}_{4}\cong\mathfrak{sl}_{2}\oplus\mathfrak{sl}_{2}$.

\subsubsection{The $\mathfrak{sp}_{2}$ case}

The algebras $Y(\mathfrak{sp}_{2})\cong Y(2)$ are equivalent, since
the $R$-matrices satisfy the relation
\begin{equation}
\frac{2u+2}{2u+1}R^{\mathfrak{sp}_{2}}(2u)=R^{\mathfrak{gl}_{2}}(u),
\end{equation}
meaning that the $\mathfrak{sp}_{2}$ $K$-matrices can be obtained
from $\mathfrak{gl}_{2}$ $K$-matrices in the following way
\begin{equation}
\mathbf{K}^{\mathfrak{sp}_{2}}(u)=\mathbf{K}^{\mathfrak{gl}_{2}}(u/2),
\end{equation}
Based on this, the relation between the $F$-operators is:
\begin{equation}
\mathbf{F}^{\mathfrak{sp}_{2}}(u)=\mathbf{F}^{\mathfrak{gl}_{2}}(u/2).
\end{equation}
The Bethe Ansatz equations in the $\mathfrak{sp}_{2}$ case are
\begin{equation}
\tilde{\alpha}(u_{j})=\prod_{k\neq j}\frac{u_{j}-u_{k}-2i}{u_{j}-u_{k}+2i},
\end{equation}
meaning that the normalization of the Bethe roots differs between
the two conventions: $u_{j}^{\mathfrak{sp}_{2}}=2u_{j}^{\mathfrak{gl}_{2}}$.
The overlap function in the $\mathfrak{gl}_{2}$ convention is
\begin{align}
\tilde{\mathcal{F}}_{\ell}^{\mathfrak{gl}_{2}}(u_{j}^{\mathfrak{gl}_{2}}) & =\mathcal{F}_{\ell}^{\mathfrak{gl}_{2}}(iu_{j}^{\mathfrak{gl}_{2}})\sqrt{\frac{(u_{j}^{\mathfrak{gl}_{2}})^{2}}{(u_{j}^{\mathfrak{gl}_{2}})^{2}+1/4}}=\mathcal{F}_{\ell}^{\mathfrak{gl}_{2}}(iu_{j}^{\mathfrak{sp}_{2}}/2)\sqrt{\frac{(u_{j}^{\mathfrak{sp}_{2}})^{2}}{(u_{j}^{\mathfrak{sp}_{2}})^{2}+1}}\nonumber \\
 & =\mathcal{F}_{\ell}^{\mathfrak{sp}_{2}}(iu_{j}^{\mathfrak{sp}_{2}})\sqrt{\frac{(u_{j}^{\mathfrak{sp}_{2}})^{2}}{(u_{j}^{\mathfrak{sp}_{2}})^{2}+1}}.
\end{align}
We can see that in the $\mathfrak{sp}_{2}$ convention, the overlap
function takes the form
\begin{equation}
\tilde{\mathcal{F}}_{\ell}^{\mathfrak{sp}_{2}}(u)=\mathcal{F}_{\ell}^{\mathfrak{sp}_{2}}(iu)\sqrt{\frac{u^{2}}{u^{2}+1/4}}.
\end{equation}

The definitions of the $\mathfrak{sp}_{2}$ $G$-operators are
\begin{equation}
\begin{split}\mathbf{G}^{(1)} & =\mathbf{K}_{1,-1}.\\
\mathbf{G}^{(2)} & =\left(\begin{array}{cc}
\mathbf{K}_{1,-1} & \mathbf{K}_{1,1}\\
\mathbf{K}_{-1,-1} & \boxed{\mathbf{K}_{-1,1}}
\end{array}\right).
\end{split}
\end{equation}

\subsubsection{The $\mathfrak{so}_{3}$ case}

The $\mathfrak{so}_{3}$ $K$-matrices can be obtained from $\mathfrak{gl}_{2}$
$K$-matrices via fusion. Let the $\mathfrak{gl}_{2}$ $K$-matrix
be $\mathbf{k}_{i,j}$. The $\mathfrak{so}_{3}$ $K$-matrix can be
obtained as
\begin{multline}
R_{1,2}(1)\mathbf{k}_{1}(u-\frac{1}{2})R_{1,2}(-2u)\mathbf{k}_{2}(u+\frac{1}{2})=\mathbf{k}_{2}(u+\frac{1}{2})R_{1,2}(-2u)\mathbf{k}_{2}(u-\frac{1}{2})R_{1,2}(1)=\\
=\sum_{a_{i},b_{i}=1}^{2}e_{a_{1},b_{1}}\otimes e_{a_{2},b_{2}}\otimes\mathbf{k}_{a_{1},a_{2}}^{b_{1},b_{2}}(u),
\end{multline}
where $\mathbf{k}_{a_{1},a_{2}}^{b_{1},b_{2}}$ is symmetric in both
lower and upper indices. The $\mathfrak{so}_{3}$ $K$-matrix $\mathbf{K}_{i,j}$
can be obtained using the index correspondence $(1,1)\equiv-1$, $(1,2)\to0$,
$(2,2)\to+1$, i.e.,
\begin{equation}
\sqrt{2}^{|i|+|j|}\mathbf{K}_{i,j}(u)=\mathbf{k}_{a_{1},a_{2}}^{b_{1},b_{2}}(2u),
\end{equation}
where $(a_{1},a_{2})\equiv i$ and $(b_{1},b_{2})\equiv j$. The components
needed for the $F$- and $G$-operators are
\begin{equation}
\begin{split}\mathbf{K}_{1,-1}(u) & =\mathbf{k}_{2,1}^{+}\mathbf{k}_{2,1}^{-},\quad\mathbf{K}_{1,0}(u)=\sqrt{2}\mathbf{k}_{2,1}^{-}\mathbf{k}_{2,2}^{+},\\
\mathbf{K}_{0,-1}(u) & =\frac{1}{\sqrt{2}}\left(\frac{4u-1}{4u}\mathbf{k}_{2,1}^{-}\mathbf{k}_{1,1}^{+}+\mathbf{k}_{1,1}^{-}\mathbf{k}_{2,1}^{+}-\frac{1}{4u}\mathbf{k}_{2,2}^{-}\mathbf{k}_{2,1}^{+}\right),\\
\mathbf{K}_{0,0}(u) & =\frac{4u-1}{4u}\mathbf{k}_{2,1}^{-}\mathbf{k}_{1,2}^{+}+\mathbf{k}_{1,1}^{-}\mathbf{k}_{2,2}^{+}-\frac{1}{4u}\mathbf{k}_{2,2}^{-}\mathbf{k}_{2,2}^{+},
\end{split}
\end{equation}
where $\mathbf{k}_{i,j}^{\pm}\equiv\mathbf{k}_{i,j}(2u\pm\frac{1}{2})$.
Let us compute the following expression
\begin{equation}
\mathbf{K}_{0,-1}(u)\left[\mathbf{K}_{1,-1}(u)\right]^{-1}\mathbf{K}_{1,0}(u)=\frac{4u-1}{4u}\mathbf{k}_{2,1}^{-}\mathbf{k}_{1,1}^{+}\left(\mathbf{k}_{2,1}^{+}\right)^{-1}\mathbf{k}_{2,2}^{+}+\mathbf{k}_{1,1}^{-}\mathbf{k}_{2,2}^{+}-\frac{1}{4u}\mathbf{k}_{2,2}^{-}\mathbf{k}_{2,2}^{+}.
\end{equation}
We combine this formula with $\mathbf{K}_{0,0}$:
\begin{equation}
\mathbf{K}_{0,0}(u)-\mathbf{K}_{0,-1}(u)\left[\mathbf{K}_{1,-1}(u)\right]^{-1}\mathbf{K}_{1,0}(u)=\frac{4u-1}{4u}\mathbf{k}_{2,1}^{-}\left(\mathbf{k}_{1,2}^{+}-\mathbf{k}_{1,1}^{+}\left(\mathbf{k}_{2,1}^{+}\right)^{-1}\mathbf{k}_{2,2}^{+}\right).
\end{equation}
We can see that the $\mathfrak{gl}_{2}$ $G$-operator appears on
the right-hand side. Based on this, we define the $\mathfrak{so}_{3}$
$F$-operator as follows
\[
\mathbf{F}^{\mathfrak{so}_{3}}(u):=\left[\mathbf{K}_{1,-1}(u)\right]^{-1}\left(\mathbf{K}_{0,0}(u)-\mathbf{K}_{0,-1}(u)\left[\mathbf{K}_{1,-1}(u)\right]^{-1}\mathbf{K}_{1,0}(u)\right)=\frac{4u-1}{4u}\mathbf{F}^{\mathfrak{gl}_{2}}(2u+1/2).
\]
The Bethe Ansatz equations in the $\mathfrak{so}_{3}$ case are
\[
\tilde{\alpha}(u_{j})=\prod_{k\neq j}\frac{u_{j}-u_{k}-i/2}{u_{j}-u_{k}+i/2},
\]
meaning $u_{j}^{\mathfrak{so}_{3}}=u_{j}^{\mathfrak{gl}_{2}}/2$.
The overlap function in the $\mathfrak{gl}_{2}$ convention is
\begin{align}
\tilde{\mathcal{F}}_{\ell}^{\mathfrak{gl}_{2}}(u_{j}^{\mathfrak{gl}_{2}}) & =\mathcal{F}_{\ell}^{\mathfrak{gl}_{2}}(iu_{j}^{\mathfrak{gl}_{2}})\sqrt{\frac{(u_{j}^{\mathfrak{gl}_{2}})^{2}}{(u_{j}^{\mathfrak{gl}_{2}})^{2}+1/4}}=\mathcal{F}_{\ell}^{\mathfrak{gl}_{2}}(2iu_{j}^{\mathfrak{so}_{3}})\sqrt{\frac{(u_{j}^{\mathfrak{so}_{3}})^{2}}{(u_{j}^{\mathfrak{so}_{3}})^{2}+1/16}}\nonumber \\
 & =\mathcal{F}_{\ell}^{\mathfrak{so}_{3}}(iu_{j}^{\mathfrak{so}_{3}}-1/4)\frac{u_{j}^{\mathfrak{so}_{3}}+i/4}{u_{j}^{\mathfrak{so}_{3}}+i/2}\sqrt{\frac{(u_{j}^{\mathfrak{so}_{3}})^{2}}{(u_{j}^{\mathfrak{so}_{3}})^{2}+1/16}}.
\end{align}
We can see that in the $\mathfrak{so}_{3}$ convention, the overlap
function takes the following form
\begin{equation}
\tilde{\mathcal{F}}_{\ell}^{\mathfrak{so}_{3}}(u)=\mathcal{F}_{\ell}^{\mathfrak{so}_{3}}(iu-1/4)\frac{u+i/4}{u+i/2}\sqrt{\frac{u^{2}}{u^{2}+1/16}}.
\end{equation}

\subsubsection{The $\mathfrak{so}_{4}$ case}

The algebras $Y(\mathfrak{so}_{4})\cong Y(2)\oplus Y(2)$ are equivalent,
meaning the $\mathfrak{so}_{4}$ $R$-matrix has a tensor product
form:
\begin{equation}
R^{\mathfrak{so}_{4}}(u)=\frac{u}{u+1}R^{\mathfrak{gl}_{2}}(u)\otimes R^{\mathfrak{gl}_{2}}(u).
\end{equation}
The $\mathfrak{so}_{4}$ reflection equation has factorizable solutions
\begin{equation}
\mathbf{K}(u)=\mathbf{k}^{L}(u)\otimes\mathbf{k}^{R}(u),
\end{equation}
where $\mathbf{k}^{L/R}$ are solutions to the uncrossed $\mathfrak{gl}_{2}$
reflection equation. In this case, we have two non-interacting $\mathfrak{gl}_{2}$
spin chains, meaning that the overlap is the product of the overlaps
of the two $\mathfrak{gl}_{2}$ subsystems
\begin{equation}
\frac{\langle\mathrm{MPS}|\bar{u}^{L},\bar{u}^{R}\rangle}{\sqrt{\langle\bar{u}^{L},\bar{u}^{R}|\bar{u}^{L},\bar{u}^{R}\rangle}}=\frac{\langle\mathrm{MPS}_{L}|\bar{u}^{L}\rangle}{\sqrt{\langle\bar{u}^{L}|\bar{u}^{L}\rangle}}\frac{\langle\mathrm{MPS}_{R}|\bar{u}^{R}\rangle}{\sqrt{\langle\bar{u}^{R}|\bar{u}^{R}\rangle}},
\end{equation}
where $\langle\mathrm{MPS}_{L/R}|$ are the MPSs corresponding to
the $\mathbf{k}^{L/R}(u)$ $K$-matrices. The $K$-matrix written
explicitly is
\begin{equation}
\mathbf{K}=\left(\begin{array}{cccc}
\mathbf{k}_{1,1}^{L}\mathbf{k}_{1,1}^{R} & \mathbf{k}_{1,1}^{L}\mathbf{k}_{1,2}^{R} & \mathbf{k}_{1,2}^{L}\mathbf{k}_{1,1}^{R} & \mathbf{k}_{1,2}^{L}\mathbf{k}_{1,2}^{R}\\
\mathbf{k}_{1,1}^{L}\mathbf{k}_{2,1}^{R} & \mathbf{k}_{1,1}^{L}\mathbf{k}_{2,2}^{R} & \mathbf{k}_{1,2}^{L}\mathbf{k}_{2,1}^{R} & \mathbf{k}_{1,2}^{L}\mathbf{k}_{2,2}^{R}\\
\mathbf{k}_{2,1}^{L}\mathbf{k}_{1,1}^{R} & \mathbf{k}_{2,1}^{L}\mathbf{k}_{1,2}^{R} & \mathbf{k}_{2,2}^{L}\mathbf{k}_{1,1}^{R} & \mathbf{k}_{2,2}^{L}\mathbf{k}_{1,2}^{R}\\
\mathbf{k}_{2,1}^{L}\mathbf{k}_{2,1}^{R} & \mathbf{k}_{2,1}^{L}\mathbf{k}_{2,2}^{R} & \mathbf{k}_{2,2}^{L}\mathbf{k}_{2,1}^{R} & \mathbf{k}_{2,2}^{L}\mathbf{k}_{2,2}^{R}
\end{array}\right),
\end{equation}
We define the $\mathfrak{so}_{4}$ $G$-operators as follows:
\begin{equation}
\begin{split}\mathbf{G}^{(R)}(u) & :=\mathbf{K}_{1,-1}(u)-\mathbf{K}_{1,-2}(u)\mathbf{K}_{2,-2}^{-1}(u)\mathbf{K}_{2,-1}(u)=\mathbf{k}_{2,1}^{L}(\mathbf{k}_{1,2}^{R}-\mathbf{k}_{1,1}^{R}\left[\mathbf{k}_{2,1}^{R}\right]^{-1}\mathbf{k}_{2,2}^{R}),\\
\mathbf{G}^{(L)}(u) & :=\mathbf{K}_{-1,1}(u)-\mathbf{K}_{-1,-2}(u)\mathbf{K}_{2,-2}^{-1}(u)\mathbf{K}_{2,1}(u)=\mathbf{k}_{2,1}^{R}(\mathbf{k}_{1,2}^{L}-\mathbf{k}_{1,1}^{L}\left[\mathbf{k}_{2,1}^{L}\right]^{-1}\mathbf{k}_{2,2}^{L}).
\end{split}
\end{equation}
It is evident that the $\mathfrak{gl}_{2}$ $G$-operators appear
on the right-hand side. We define the $\mathfrak{so}_{4}$ $F$-operators
as follows
\begin{equation}
\mathbf{F}^{L/R}(u):=\left[\mathbf{K}_{2,-2}(u)\right]^{-1}\mathbf{G}^{(L/R)}(u)=\mathbf{F}^{\mathfrak{gl}_{2},L/R}(u).
\end{equation}
We can see that the $\mathfrak{so}_{4}$ $F$-operators are equal
to the $\mathfrak{gl}_{2}$ $F$-operators. In these calculations,
we assumed that $\mathbf{K}_{2,-2}(u)$ is invertible, which implies
that $\mathbf{k}_{2,1}^{L}$ and $\mathbf{k}_{2,1}^{R}$ are also
invertible, which holds for $\mathcal{B}(2,1)$ $K$-matrices. Therefore,
the above formulas are valid for the $Y(\mathfrak{so}_{4},\mathfrak{so}_{2}\oplus\mathfrak{so}_{2})\cong\mathcal{B}(2,1)\oplus\mathcal{B}(2,1)$
$K$-matrices. 

We can also say something about the case $Y(\mathfrak{so}_{4},\mathfrak{gl}_{2})\cong\mathcal{B}(2,1)\oplus\mathcal{B}(2,0)$.
In this case, $\mathbf{K}_{2,-2}(u)$ is not invertible (see also
(\ref{eq:Uso4gl2})), but $\mathbf{K}_{2,-1}(u)=\mathbf{k}_{2,1}^{L}\mathbf{k}_{2,2}^{R}$
is! Using this, $\mathbf{F}^{\mathfrak{gl}_{2},L}(u)$ can be expressed
as
\begin{equation}
\mathbf{F}^{L}(u):=\left[\mathbf{K}_{2,-1}(u)\right]^{-1}\left(\mathbf{K}_{-1,2}(u)-\mathbf{K}_{-1,-1}(u)\left[\mathbf{K}_{2,-1}(u)\right]^{-1}\mathbf{K}_{2,2}(u)\right)=\mathbf{F}^{\mathfrak{gl}_{2},L}(u).\label{eq:Fsogl}
\end{equation}

The $\mathfrak{so}_{4}$ reflection algebra also has achiral solutions,
meaning that the $K$-matrix does not factorize. The achiral solutions
are representations of the $Y(\mathfrak{so}_{4},\mathfrak{so}_{3})$
reflection algebra. These $K$-matrices can be obtained from representations
of the $Y(2)$ algebra as follows:
\begin{equation}
\mathbf{K}(u)=(\widehat{\mathbf{L}}(u)\otimes\mathbf{1})\check{R}(2u)(\mathbf{L}(u)\otimes\mathbf{1}),\label{eq:so4K}
\end{equation}
where $\mathbf{L}(u)$ is a $Y(2)$ Lax-operator and $\widehat{\mathbf{L}}(u)$
is the usual inverse-transposed Lax-operator. In the above equation,
we used the notation $\check{R}(u)=PR(u)$. In this case, the two
non-interacting $\mathfrak{gl}_{2}$ systems are coupled through the
boundary state. This results in an achiral pair structure, i.e., $\bar{u}^{L}=-\bar{u}^{R}$.
In this case, the Gaudin determinants are equal:
\begin{equation}
\det G^{+}=\det G^{-}=\det G^{R}=\det G^{L},\quad\rightarrow\quad\frac{\det G^{+}}{\det G^{-}}=1.
\end{equation}
The overlap can be reduced to a matrix element between Bethe vectors
of a $Y(2)$ transfer matrix, i.e., the overlap can be written as
\begin{equation}
\frac{\langle\mathrm{MPS}|\bar{u}^{L},\bar{u}^{R}\rangle}{\sqrt{\langle\bar{u}^{L},\bar{u}^{R}|\bar{u}^{L},\bar{u}^{R}\rangle}}=\sum_{\ell=1}^{d_{B}}\beta_{\ell}\mathcal{\tilde{F}}_{\ell}^{\mathfrak{gl}_{2}}(\bar{u}^{R}),
\end{equation}
where $\tilde{\mathcal{F}}_{\ell}^{\mathfrak{gl}_{2}}(u)=\mathcal{F}_{\ell}^{\mathfrak{gl}_{2}}(iu)$
are the eigenvalues of the $F$-operator
\begin{equation}
\mathbf{F}^{\mathfrak{gl}_{2}}(u)=\left(\mathbf{L}_{2,1}(u)\right)^{-1}\left(\mathbf{L}_{1,2}(u)-\mathbf{L}_{1,1}(u)\left(\mathbf{L}_{2,1}(u)\right)^{-1}\mathbf{L}_{2,2}(u)\right).
\end{equation}
Now we express $\mathbf{F}^{\mathfrak{gl}_{2}}$ in terms of the matrix
elements of the $\mathfrak{so}_{4}$ $K$-matrix (\ref{eq:so4K})
The components needed for the $G$- and $F$-operators are
\begin{equation}
\begin{split}\mathbf{K}_{2,-2}(u) & =\widehat{\mathbf{L}}_{2,1}(u)\mathbf{L}_{2,1}(u),\quad\mathbf{K}_{2,1}(u)=\widehat{\mathbf{L}}_{2,1}(u)\mathbf{L}_{2,2}(u),\\
\mathbf{K}_{1,-2}(u) & =\frac{2u-1}{2u}\widehat{\mathbf{L}}_{2,1}(u)\mathbf{L}_{1,1}(u)-\frac{1}{2u}\widehat{\mathbf{L}}_{2,2}(u)\mathbf{L}_{2,1}(u),\\
\mathbf{K}_{1,1}(u) & =\frac{2u-1}{2u}\widehat{\mathbf{L}}_{2,1}(u)\mathbf{L}_{1,2}(u)-\frac{1}{2u}\widehat{\mathbf{L}}_{2,2}(u)\mathbf{L}_{2,2}(u).
\end{split}
\end{equation}
The achiral $\mathfrak{so}_{4}$ $G$-operator is defined as follows:
\begin{align}
\mathbf{G}^{(2)}(u) & :=\mathbf{K}_{1,1}(u)-\mathbf{K}_{1,-2}(u)\left[\mathbf{K}_{2,-2}(u)\right]^{-1}\mathbf{K}_{2,1}(u)\nonumber \\
 & =\frac{2u-1}{2u}\widehat{\mathbf{L}}_{2,1}(u)\left(\mathbf{L}_{1,2}(u)-\mathbf{L}_{1,1}(u)\left(\mathbf{L}_{2,1}(u)\right)^{-1}\mathbf{L}_{2,2}(u)\right).
\end{align}
The achiral $\mathfrak{so}_{4}$ $F$-operator is defined as:
\begin{equation}
\mathbf{F}^{\mathfrak{so}_{4}}(u):=\left[\mathbf{K}_{2,-2}(u)\right]^{-1}\left(\mathbf{K}_{1,1}(u)-\mathbf{K}_{1,-2}(u)\left[\mathbf{K}_{2,-2}(u)\right]^{-1}\mathbf{K}_{2,1}(u)\right)=\frac{u-1/2}{u}\mathbf{F}^{\mathfrak{gl}_{2}}(u),
\end{equation}
meaning that the one-particle overlap functions are
\begin{equation}
\tilde{\mathcal{F}}_{\ell}^{\mathfrak{so}_{4}}(u)=\mathcal{F}_{\ell}^{\mathfrak{so}_{4}}(iu)\frac{u}{u+i/2}.
\end{equation}

\subsection{$F$-operators without extra selections rules}

In the previous sections, we determined the $F$-operators and functions
for the $\mathfrak{gl}_{n}$ sector and for the remaining nodes separately.
Based on these, we can now obtain the full overlap functions. The
results are summarized below.

For the $\mathfrak{so}_{2n+1}$ models
\begin{equation}
\tilde{\mathcal{F}}_{\ell}^{(s)}(u)=\begin{cases}
\mathcal{F}_{\ell}^{(s)}(iu-\frac{1}{4}+\frac{n-s}{2})\sqrt{\frac{u^{2}}{u^{2}+1/4}}, & \text{for }s=1,\dots,n-1,\\
\mathcal{F}_{\ell}^{(n)}(iu-1/4)\frac{u+i/4}{u+i/2}\sqrt{\frac{u^{2}}{u^{2}+1/16}}, & \text{for }s=n.
\end{cases}
\end{equation}

For the $\mathfrak{sp}_{2n}$ models
\begin{equation}
\tilde{\mathcal{F}}_{\ell}^{(s)}(u)=\begin{cases}
\mathcal{F}_{\ell}^{(s)}(iu+\frac{n+1-s}{2})\sqrt{\frac{u^{2}}{u^{2}+1/4}}, & \text{for }s=1,\dots,n-1,\\
\mathcal{F}_{\ell}^{(n)}(iu)\sqrt{\frac{u^{2}}{u^{2}+1}}, & \text{for }s=n.
\end{cases}
\end{equation}

For the chiral $\mathfrak{so}_{2n}$ models
\begin{equation}
\tilde{\mathcal{F}}_{\ell}^{(s)}(u)=\begin{cases}
\mathcal{F}_{\ell}^{(s)}(iu+\frac{n-1-s}{2})\sqrt{\frac{u^{2}}{u^{2}+1/4}}, & \text{for }s=1,\dots,n-2,\\
\mathcal{F}_{\ell}^{(s)}(iu)\sqrt{\frac{u^{2}}{u^{2}+1}}, & \text{for }s=n-1,n.
\end{cases}
\end{equation}

For the achiral $\mathfrak{so}_{2n}$ models
\begin{equation}
\tilde{\mathcal{F}}_{\ell}^{(s)}(u)=\begin{cases}
\mathcal{F}_{\ell}^{(s)}(iu+\frac{n-1-s}{2})\sqrt{\frac{u^{2}}{u^{2}+1/4}}, & \text{for }s=1,\dots,n-2,\\
\mathcal{F}_{\ell}^{(n-1)}(iu)\frac{u}{u+i/2}, & \text{for }s=n-1.
\end{cases}
\end{equation}
The $F$-operators are defined via the $G$-operators:
\begin{equation}
\mathbf{F}^{(s)}(u)=\left[\mathbf{G}^{(s)}(u)\right]^{-1}\mathbf{G}^{(s+1)}(u).\label{eq:Fdef}
\end{equation}
In the cases of $\mathfrak{so}_{2n+1}$ and $\mathfrak{sp}_{2n}$
this applies for $s=1,\dots,n$ and in the case of $\mathfrak{so}_{2n}$
for $s=1,\dots,n-1$. In the chiral $\mathfrak{so}_{2n}$ case, the
remaining $F$-operator is:
\begin{equation}
\mathbf{F}^{(n)}(u)=\left[\mathbf{G}^{(n-1)}(u)\right]^{-1}\mathbf{G}^{(n+1)}(u).
\end{equation}
The $G$-operators can be defined through nested $K$-matrices using
equation (\ref{eq:soNrec0}). For $\mathfrak{so}_{2n+1}$ and $\mathfrak{sp}_{2n}$,
the matrices $\mathbf{K}^{(s)}$ are defined for $s=1,\dots,n$ and
for $\mathfrak{so}_{2n}$, for $s=1,\dots,n-1$. These define the
$G$-operators for the corresponding values of $s$ as
\begin{equation}
\mathbf{G}^{(s)}(u)=\mathbf{K}_{n+1-s,-n-1+s}^{(s)}(u).
\end{equation}
The remaining $G$-operators are defined as follows.

For the $\mathfrak{so}_{2n+1}$ models
\begin{equation}
\mathbf{G}^{(n+1)}(u):=\mathbf{K}_{0,0}^{(n)}(u)-\mathbf{K}_{0,-1}^{(n)}(u)\left[\mathbf{K}_{1,-1}^{(n)}(u)\right]^{-1}\mathbf{K}_{1,0}^{(n)}(u).
\end{equation}

For the $\mathfrak{sp}_{2n}$ models

\begin{equation}
\mathbf{G}^{(n+1)}(u):=\mathbf{K}_{-1,1}^{(n)}(u)-\mathbf{K}_{-1,-1}^{(n)}(u)\left[\mathbf{K}_{1,-1}^{(n)}(u)\right]^{-1}\mathbf{K}_{1,1}^{(n)}(u),
\end{equation}

For the chiral $\mathfrak{so}_{2n}$ models
\begin{equation}
\begin{split}\mathbf{G}^{(n)}(u) & :=\mathbf{K}_{1,-1}^{(n-1)}(u)-\mathbf{K}_{1,-2}^{(n-1)}(u)\left[\mathbf{K}_{2,-2}^{(n-1)}(u)\right]^{-1}\mathbf{K}_{2,-1}^{(n-1)}(u),\\
\mathbf{G}^{(n+1)}(u) & :=\mathbf{K}_{-1,1}^{(n-1)}(u)-\mathbf{K}_{-1,-2}^{(n-1)}(u)\left[\mathbf{K}_{2,-2}^{(n-1)}(u)\right]^{-1}\mathbf{K}_{2,1}^{(n-1)}(u).
\end{split}
\end{equation}

For the achiral $\mathfrak{so}_{2n}$ models
\begin{equation}
\mathbf{G}^{(n)}(u):=\mathbf{K}_{1,1}^{(n-1)}(u)-\mathbf{K}_{1,-2}^{(n-1)}(u)\left[\mathbf{K}_{2,-2}^{(n-1)}(u)\right]^{-1}\mathbf{K}_{2,1}^{(n-1)}(u).\label{eq:Gn}
\end{equation}

The $G$-operators can also be expressed using quasi-determinants.
In both chiral and achiral cases, equation (\ref{eq:Ggn}) defines
the $\mathbf{G}^{(s)}$ operators for $s=1,\dots,n$ or $s=1,\dots,n-1$.
The missing $G$-operators are as follows.

For the $\mathfrak{so}_{2n+1}$ models
\begin{equation}
\mathbf{G}^{(n+1)}=\left(\begin{array}{cccc}
\mathbf{K}_{n,-n} & \dots & \mathbf{K}_{n,-1} & \mathbf{K}_{n,0}\\
\vdots & \ddots & \vdots & \vdots\\
\mathbf{K}_{1,-n} & \dots & \mathbf{K}_{1,-1} & \mathbf{K}_{1,0}\\
\mathbf{K}_{0,-n} & \dots & \mathbf{K}_{0,-1} & \boxed{\mathbf{K}_{0,0}}
\end{array}\right),
\end{equation}

For the $\mathfrak{sp}_{2n}$ models
\begin{equation}
\mathbf{G}^{(n+1)}=\left(\begin{array}{cccc}
\mathbf{K}_{n,-n} & \dots & \mathbf{K}_{n,-1} & \mathbf{K}_{n,1}\\
\vdots & \ddots & \vdots & \vdots\\
\mathbf{K}_{1,-n} & \dots & \mathbf{K}_{1,-1} & \mathbf{K}_{1,1}\\
\mathbf{K}_{-1,-n} & \dots & \mathbf{K}_{-1,-1} & \boxed{\mathbf{K}_{-1,1}}
\end{array}\right),
\end{equation}

For the chiral $\mathfrak{so}_{2n}$ models
\begin{equation}
\mathbf{G}^{(n+1)}=\left(\begin{array}{cccc}
\mathbf{K}_{n,-n} & \dots & \mathbf{K}_{n,-2} & \mathbf{K}_{n,1}\\
\vdots & \ddots & \vdots & \vdots\\
\mathbf{K}_{2,-n} & \dots & \mathbf{K}_{2,-2} & \mathbf{K}_{2,1}\\
\mathbf{K}_{-1,-n} & \dots & \mathbf{K}_{-1,-2} & \boxed{\mathbf{K}_{-1,1}}
\end{array}\right).
\end{equation}

For the achiral $\mathfrak{so}_{2n}$ models
\begin{equation}
\mathbf{G}^{(n)}=\left(\begin{array}{cccc}
\mathbf{K}_{n,-n} & \dots & \mathbf{K}_{n,-2} & \mathbf{K}_{n,1}\\
\vdots & \ddots & \vdots & \vdots\\
\mathbf{K}_{2,-n} & \dots & \mathbf{K}_{2,-2} & \mathbf{K}_{2,1}\\
\mathbf{K}_{1,-n} & \dots & \mathbf{K}_{1,-2} & \boxed{\mathbf{K}_{1,1}}
\end{array}\right).\label{eq:achiralG}
\end{equation}

These constructions assume the existence of inverses for certain operators.
This is valid for the following reflection algebras
\begin{itemize}
\item $Y(\mathfrak{so}_{2n+1},\mathfrak{so}_{n}\oplus\mathfrak{so}_{n+1})$
\item $Y(\mathfrak{sp}_{2n},\mathfrak{gl}_{n})$
\item $Y(\mathfrak{so}_{2n},\mathfrak{so}_{n}\oplus\mathfrak{so}_{n})$,
which lead to chiral pair structures
\item $Y(\mathfrak{so}_{2n},\mathfrak{so}_{n-1}\oplus\mathfrak{so}_{n+1})$,
which lead to achiral pair structures
\end{itemize}
For the remaining reflection algebras, the above recursive equations
cannot be applied directly. However, the techniques described in section
\ref{sec:Other-reflection-algebras} still work, meaning that after
suitable deformations, the $F$-operators can be defined. From here
on, each case is treated separately. A common feature in these cases
is the need for certain deformations of the nested $K$-matrices.
These are continuous deformations, and in the final overlap functions,
the zero limit of the deformation parameter must be taken. In this
limit, the overlap vanishes for certain quantum numbers, leading to
additional selection rules.

\subsection{$F$-operators with extra selections rules}

\subsubsection{$Y(\mathfrak{so}_{2n},\mathfrak{gl}_{n})$}

In this case, the $K$-matrix of the $\mathfrak{gl}_{n}$ sector is
a representation of $Y^{-}(n)$ (see equations (\ref{eq:Kcon}) and
(\ref{eq:sogl})), so we apply the trick described in section \ref{subsec:Generalization-Ym}.
The previous $\mathbf{K}^{(2s)}$ matrices do not exist, but the $\mathbf{K}^{(2s-1)}$
matrices do. For these matrices, we can use the recursion (\ref{eq:soNrec}).
From these, we select $k=\left\lfloor \frac{n}{2}\right\rfloor $
number of $Y^{-}(2)$ $K$-matrices (\ref{eq:Ym2k}). The $G$-operator
formulas (\ref{eq:Gop0}) cannot be applied directly. Therefore, we
deform these $Y^{-}(2)$ $K$-matrices in the way described in section
\ref{subsec:Generalization-Ym}, i.e.,
\begin{equation}
\begin{split}\mathbf{k}_{i,j}^{(s)}(u) & =\sum_{k,l=1}^{2}\mathbf{L}_{k,i}^{(s)}(u-\kappa_{N-4s+4})\epsilon_{k,l}\mathbf{L}_{l,j}^{(s)}(-u),\\
\tilde{\mathbf{k}}_{i,j}^{(s)}(u) & =\mathbf{k}_{i,j}^{(s)}(u)+(u+1/2-\kappa_{N-4s+4}/2)\mathbf{L}_{1,i}^{(s)}(u-\kappa_{N-4s+4})\mathbf{L}_{1,j}^{(s)}(-u),
\end{split}
\label{eq:kdef}
\end{equation}
and for these, the $G$-operators are well-defined
\begin{equation}
\begin{split}\mathbf{G}^{(2s-1)}(u) & =\tilde{\mathbf{k}}_{1,1}^{(s)}(u),\\
\mathbf{G}^{(2s)}(u) & =\tilde{\mathbf{k}}_{2,2}^{(s)}(u)-\tilde{\mathbf{k}}_{2,1}^{(s)}(u)\left[\tilde{\mathbf{k}}_{1,1}^{(s)}(u)\right]^{-1}\tilde{\mathbf{k}}_{1,2}^{(s)}(u),
\end{split}
\label{eq:GYm2}
\end{equation}
for $s=1,\dots,k$. These define the $F$-operators $\mathbf{F}^{(s)}$
in the usual way for $s=1,\dots,2k-1$. From here on, it is practical
to treat even and odd $n$ separately.

\paragraph{$Y(\mathfrak{so}_{4k},\mathfrak{gl}_{2k})$}

In this case, the nesting proceeds as follows.
\[
\begin{array}{ccccccc}
Y(\mathfrak{so}_{4k},\mathfrak{gl}_{2k}) & \to & Y(\mathfrak{so}_{4k-4},\mathfrak{gl}_{2k-2}) & \to & \dots & \to & Y(\mathfrak{so}_{4},\mathfrak{gl}_{2})\\
\mathbf{K}^{(1)}\equiv\mathbf{K} & \to & \mathbf{K}^{(3)} & \to & \dots & \to & \mathbf{K}^{(2k-1)}\\
\downarrow &  & \downarrow &  & \dots &  & \downarrow\\
\mathbf{k}^{(1)} &  & \mathbf{k}^{(2)} &  & \dots &  & \mathbf{k}^{(k)}
\end{array}
\]
We see that in the final step, we obtain a $Y(\mathfrak{so}_{4},\mathfrak{gl}_{2})$
$K$-matrix, which has a chiral pair structure. In this case, we have
$n=2k$ $F$-operators, of which the first $2k-1$ were already given.
We still need to define the missing $\mathbf{F}^{(2k)}$ operator.
In the final step, the $K$-matrix is a representation of $Y(\mathfrak{so}_{4},\mathfrak{gl}_{2})\cong\mathcal{B}(2,1)\oplus\mathcal{B}(2,0)$.
The $F$-operator corresponding to the $\mathcal{B}(2,0)$ subalgebra
was already given, that was $\mathbf{F}^{(2k-1)}$. The $F$-operator
corresponding to the $\mathcal{B}(2,1)$ subalgebra will be the missing
$\mathbf{F}^{(2k)}$, and it can be determined using equation (\ref{eq:Fsogl}),
i.e.,
\begin{equation}
\mathbf{F}^{(2k)}(u)=\left[\mathbf{K}_{2,-1}^{(2k-1)}(u)\right]^{-1}\left(\mathbf{K}_{-1,2}^{(2k-1)}(u)-\mathbf{K}_{-1,-1}^{(2k-1)}(u)\left[\mathbf{K}_{2,-1}^{(2k-1)}(u)\right]^{-1}\mathbf{K}_{2,2}^{(2k-1)}(u)\right).
\end{equation}
Since the $K$-matrix of the $\mathfrak{gl}_{n}$ sector is a representation
of $Y^{-}(n)$, there are extra selection rules (see (\ref{eq:selrule-1})).
The overlap is non-zero only if
\begin{equation}
\begin{split}r_{2s-1} & =\frac{\Lambda_{2s-1}-\Lambda_{2s}+r_{2s-2}+r_{2s}}{2},\quad\text{for }s=1,\dots,\frac{n}{2}-1,\\
r_{n-1} & =\frac{\Lambda_{n-1}-\Lambda_{n}+r_{n-2}}{2}.
\end{split}
\end{equation}

\paragraph{$Y(\mathfrak{so}_{4k+2},\mathfrak{gl}_{2k+1})$}

In this case, the nesting proceeds as follows.
\[
\begin{array}{ccccccc}
Y(\mathfrak{so}_{4k+2},\mathfrak{gl}_{2k+1}) & \to & Y(\mathfrak{so}_{4k-2},\mathfrak{gl}_{2k-1}) & \to & \dots & \to & Y(\mathfrak{so}_{6},\mathfrak{gl}_{3})\\
\mathbf{K}^{(1)}\equiv\mathbf{K} & \to & \mathbf{K}^{(3)} & \to & \dots & \to & \mathbf{K}^{(2k-1)}\\
\downarrow &  & \downarrow &  & \dots &  & \downarrow\\
\mathbf{k}^{(1)} &  & \mathbf{k}^{(2)} &  & \dots &  & \mathbf{k}^{(k)}
\end{array}
\]
We see that in the final step, we obtain a $Y(\mathfrak{so}_{6},\mathfrak{gl}_{3})$
$K$-matrix. This algebra is equivalent to the $\mathcal{B}(4,1)$
reflection algebra, for which the pair structure is achiral. In the
achiral case, we need $n=2k+1$ $G$-operators, of which the first
$2k$ were already given, and $\mathbf{G}^{(2k+1)}(u)$ is still missing.
The definition (\ref{eq:Gn}) cannot be used, but (\ref{eq:achiralG})
can. Additionally, we can use another equivalent definition:
\begin{equation}
\mathbf{G}^{(2k+1)}(u)=\mathbf{K}_{1,1}^{(2k-1)}(u)-\sum_{\alpha,\beta=2}^{3}\mathbf{K}_{1,-\alpha}^{(2k-1)}(u)\widehat{\mathbf{K}}_{-\alpha,\beta}^{(2k-1)}(u)\mathbf{K}_{\beta,1}^{(2k-1)}(u).
\end{equation}
Since the $K$-matrix of the $\mathfrak{gl}_{n-1}$ sector is a representation
of $Y^{-}(n-1)$, there are extra selection rules (see (\ref{eq:selrule-1})).
The overlap is non-zero only if
\begin{equation}
\begin{split}r_{2s-1} & =\frac{\Lambda_{2s-1}-\Lambda_{2s}+r_{2s-2}+r_{2s}}{2},\quad\text{for }s=1,\dots,\frac{n-3}{2},\\
r_{n-2} & =\frac{\Lambda_{n-2}-\Lambda_{n-1}+r_{n-3}+r_{n-1}+r_{n}}{2}.
\end{split}
\end{equation}

\subsubsection{$Y(\mathfrak{so}_{2n},\mathfrak{o}_{M}\oplus\mathfrak{o}_{2n-M})$}

In this case, the nesting proceeds as follows:
\[
\begin{array}{ccccccc}
Y(\mathfrak{so}_{2n},\mathfrak{o}_{M}\oplus\mathfrak{o}_{2n-M}) & \to & Y(\mathfrak{so}_{2n-2},\mathfrak{o}_{M-1}\oplus\mathfrak{o}_{2n-M-1}) & \to & \dots & \to & Y(\mathfrak{so}_{2n-2M},\mathfrak{so}_{2n-2M})\\
\mathbf{K}^{(1)}\equiv\mathbf{K} & \to & \mathbf{K}^{(2)} & \to & \dots & \to & \mathbf{K}^{(M+1)}\\
\downarrow &  & \downarrow &  & \dots\\
\mathbf{G}^{(1)} &  & \mathbf{G}^{(2)} &  & \dots
\end{array}
\]
In the $M+1$-th step, we obtain a $Y(\mathfrak{so}_{2n-2M},\mathfrak{so}_{2n-2M})$
$K$-matrix. This means that the MPS is a singlet state with respect
to the $\mathfrak{so}_{2n-2M}$ subalgebra, i.e., the overlap is non-zero
only if
\begin{equation}
\begin{split}r_{s} & =r_{M}+\sum_{k=M+1}^{s}\Lambda_{k},\quad\text{for }s=M+1,\dots,n-2,\\
r_{n-1} & =\frac{1}{2}r_{M}+\frac{1}{2}\sum_{k=M+1}^{n-1}\Lambda_{k}-\frac{1}{2}\Lambda_{n},\\
r_{n} & =\frac{1}{2}r_{M}+\frac{1}{2}\sum_{k=M+1}^{n-1}\Lambda_{k}+\frac{1}{2}\Lambda_{n}.
\end{split}
\end{equation}
The corresponding $\mathbf{K}_{j,-j}^{(M+1)}$ components are not
invertible, so the nesting cannot continue, and the operators $\mathbf{G}^{(M+1)},\dots$
are not defined. However, the trick described for the $\mathcal{B}(N,M)$
reflection algebras can be applied here as well, i.e., we deform the
factorized form of the $K$-matrix with a singular scalar $K$-matrix.
This singular solution is
\begin{equation}
K(u)=1+u\sum_{j=1}^{\tilde{n}/2}(e_{\tilde{n}-2j+1,2j-2-\tilde{n}}-e_{\tilde{n}-2j+2,2j-1-\tilde{n}}),\label{eq:Ksing}
\end{equation}
where $\tilde{n}=n-M$. Assuming that the $Y(\mathfrak{so}_{2n-2M},\mathfrak{so}_{2n-2M})$
matrix can be factorized into
\begin{equation}
\mathbf{K}_{a,b}^{(M+1)}(u)=\sum_{c=-\tilde{n}}^{\tilde{n}}\mathbf{L}_{-c,-a}(u-\kappa_{N-2M})\mathbf{L}_{c,b}(-u),\label{eq:Kfactor}
\end{equation}
the deformed $K$-matrix is
\begin{align}
\tilde{\mathbf{K}}_{a,b}^{(1)}(u) & =\mathbf{K}_{a,b}^{(M+1)}(u)+\label{eq:Kdef}\\
 & +u\sum_{l=1}^{k}\left(\mathbf{L}_{-(\tilde{n}-2l+1),-a}(u-\kappa_{N-2M})\mathbf{L}_{-\tilde{n}+2l-2,b}(-u)-\mathbf{L}_{-(\tilde{n}-2l+2),-a}(u-\kappa_{N-2M})\mathbf{L}_{-\tilde{n}+2l-1,b}(-u)\right).\nonumber 
\end{align}
It is clear that this deformed $K$-matrix forms a $Y^{-}(\tilde{n})$
subalgebra in the $\mathfrak{gl}_{\tilde{n}}$ sub-sector, i.e., the
same nesting can be carried out as for the $Y(\mathfrak{so}_{2\tilde{n}},\mathfrak{gl}_{\tilde{n}})$
algebra. We can define the nested $K$-matrices in the manner of (\ref{eq:soNrec}).
We can select $k=\left\lfloor \frac{\tilde{n}}{2}\right\rfloor $
number of $Y^{-}(2)$ $K$-matrices $\mathbf{k}^{(s)}$ for $s=1,\dots,k$.
These can be deformed in the manner of (\ref{eq:kdef}), and for these,
the $G$-operators are well-defined
\begin{equation}
\begin{split}\mathbf{G}^{(M+2s-1)}(u) & =\tilde{\mathbf{k}}_{1,1}^{(s)}(u),\\
\mathbf{G}^{(M+2s)}(u) & =\tilde{\mathbf{k}}_{2,2}^{(s)}(u)-\tilde{\mathbf{k}}_{2,1}^{(s)}(u)\left[\tilde{\mathbf{k}}_{1,1}^{(s)}(u)\right]^{-1}\tilde{\mathbf{k}}_{1,2}^{(s)}(u),
\end{split}
\label{eq:GYm2-1}
\end{equation}
for $s=1,\dots,k$. 

If $n-M$ is \textbf{even}, i.e., $\tilde{n}=2k$, then the pair structure
is \textbf{chiral}. The $G$-operators $\mathbf{G}^{(s)}(u)$ previously
defined for $s=1,\dots,n$ immediately define the first $n-1$ $F$-operators.
The missing $F$-operator is
\begin{equation}
\mathbf{F}^{(n)}(u)=\left[\tilde{\mathbf{K}}_{2,-1}^{(2k-1)}(u)\right]^{-1}\left(\tilde{\mathbf{K}}_{-1,2}^{(2k-1)}(u)-\tilde{\mathbf{K}}_{-1,-1}^{(2k-1)}(u)\left[\tilde{\mathbf{K}}_{2,-1}^{(2k-1)}(u)\right]^{-1}\tilde{\mathbf{K}}_{2,2}^{(2k-1)}(u)\right).
\end{equation}

If $n-M$ is \textbf{odd}, i.e., $\tilde{n}=2k+1$, then the pair
structure is \textbf{achiral}. The $G$-operators $\mathbf{G}^{(s)}(u)$
previously defined for $s=1,\dots,n-1$. The missing $G$-operator
is 
\begin{equation}
\mathbf{G}^{(n)}(u)=\tilde{\mathbf{K}}_{1,1}^{(2k-1)}(u)-\sum_{\alpha,\beta=2}^{3}\tilde{\mathbf{K}}_{1,-\alpha}^{(2k-1)}(u)\widehat{\mathbf{K}}_{-\alpha,\beta}^{(2k-1)}(u)\tilde{\mathbf{K}}_{\beta,1}^{(2k-1)}(u).
\end{equation}

\subsubsection{$Y(\mathfrak{so}_{2n+1},\mathfrak{o}_{M}\oplus\mathfrak{o}_{2n+1-M})$}

This case is very similar to the earlier $Y(\mathfrak{so}_{2n},\mathfrak{o}_{M}\oplus\mathfrak{o}_{2n-M})$
case. Now, the nesting proceeds as follows.
\[
\begin{array}{ccccccc}
Y(\mathfrak{so}_{2n+1},\mathfrak{o}_{M}\oplus\mathfrak{o}_{2n+1-M}) & \to & Y(\mathfrak{so}_{2n-1},\mathfrak{o}_{M-1}\oplus\mathfrak{o}_{2n-M}) & \to & \dots & \to & Y(\mathfrak{so}_{2n+1-2M},\mathfrak{so}_{2n+1-2M})\\
\mathbf{K}^{(1)}\equiv\mathbf{K} & \to & \mathbf{K}^{(2)} & \to & \dots & \to & \mathbf{K}^{(M+1)}\\
\downarrow &  & \downarrow &  & \dots\\
\mathbf{G}^{(1)} &  & \mathbf{G}^{(2)} &  & \dots
\end{array}
\]
In the $M+1$-th step, we obtain a $Y(\mathfrak{so}_{2n+1-2M},\mathfrak{so}_{2n+1-2M})$
$K$-matrix. This means that the MPS is a singlet state with respect
to the $\mathfrak{so}_{2n+1-2M}$ subalgebra, i.e., the overlap is
non-zero only if
\begin{equation}
r_{s}=r_{M}+\sum_{k=M+1}^{s}\Lambda_{k},
\end{equation}
for $s=M+1,\dots,n$.

To compute the $F$-operators, we again use the trick described in
the previous section: we deform the factorized form of the $K$-matrix
(\ref{eq:Kfactor}) with a singular scalar $K$-matrix (\ref{eq:Ksing})
as in (\ref{eq:Kdef}), where $\tilde{n}=n-M$. It is evident that
the deformed $K$-matrix again forms a $Y^{-}(\tilde{n})$ subalgebra
in the $\mathfrak{gl}_{\tilde{n}}$ sub-sector. Then we can define
the nested $K$-matrices using (\ref{eq:soNrec}). We select $k=\left\lfloor \frac{\tilde{n}}{2}\right\rfloor $
number of $Y^{-}(2)$ $K$-matrices $\mathbf{k}^{(s)}$ for $s=1,\dots,k$.
These can be deformed using (\ref{eq:kdef}), and the $G$-operators
are well-defined for them (\ref{eq:GYm2-1}). 

If $n-M$ is \textbf{even}, i.e., $\tilde{n}=2k$, then the $G$-operators
$\mathbf{G}^{(s)}(u)$ have already been defined for $s=1,\dots,n$.
The missing $G$-operator is
\begin{equation}
\mathbf{G}^{(n+1)}(u)=\tilde{\mathbf{K}}_{0,0}^{(2k-1)}(u)-\sum_{\alpha,\beta=1}^{2}\tilde{\mathbf{K}}_{0,-\alpha}^{(2k-1)}(u)\widehat{\mathbf{K}}_{-\alpha,\beta}^{(2k-1)}(u)\tilde{\mathbf{K}}_{\beta,0}^{(2k-1)}(u).
\end{equation}

If $n-M$ is \textbf{odd}, i.e., $\tilde{n}=2k+1$, then the $G$-operators
$\mathbf{G}^{(s)}(u)$ have already been defined for $s=1,\dots,n-1$.
Additionally, we need $\mathbf{G}^{(n)}$ and $\mathbf{G}^{(n+1)}$.
The $K$-matrix $\tilde{\mathbf{K}}^{(2k+1)}$ is a representation
of $Y(\mathfrak{so}_{3},\mathfrak{so}_{3})$. This must again be deformed
using the following singular $\mathfrak{so}_{3}$ $K$-matrix:
\begin{equation}
K(u)=\left(\begin{array}{ccc}
1 & 0 & 0\\
u & 1 & 0\\
-\frac{1}{2}u(u+\frac{1}{4}) & -u & 1
\end{array}\right).
\end{equation}
The deformed $\mathfrak{so}_{3}$ $K$-matrix $\check{\mathbf{K}}$
is obtained as
\begin{equation}
\begin{split}\tilde{\mathbf{K}}_{a,b}^{(k+1)}(u) & =\sum_{c=-1}^{1}\mathbf{L}_{-c,-a}(u-\kappa_{3})\mathbf{L}_{c,b}(-u),\\
\check{\mathbf{K}}_{a,b}^{(k+1)}(u) & =\tilde{\mathbf{K}}_{a,b}^{(k+1)}(u)-\frac{1}{2}u(u+\frac{1}{4})\mathbf{L}_{-1,-a}(u-\kappa_{3})\mathbf{L}_{-1,b}(-u)\\
 & +u\left(\mathbf{L}_{0,-a}(u-\kappa_{3})\mathbf{L}_{-1,b}(-u)-\mathbf{L}_{-1,-a}(u-\kappa_{3})\mathbf{L}_{0,b}(-u)\right),
\end{split}
\label{eq:KO4k1}
\end{equation}
and the remaining $G$-operators $\mathbf{G}^{(n)}$, $\mathbf{G}^{(n+1)}$
are given by
\begin{equation}
\begin{split}\mathbf{G}^{(n)}(u) & =\check{\mathbf{K}}_{1,-1}^{(k+1)}(u),\\
\mathbf{G}^{(n+1)}(u) & =\check{\mathbf{K}}_{0,0}^{(k+1)}(u)-\check{\mathbf{K}}_{0,-1}^{(k+1)}(u)\left[\check{\mathbf{K}}_{1,-1}^{(k+1)}(u)\right]^{-1}\check{\mathbf{K}}_{1,0}^{(k+1)}(u).
\end{split}
\label{eq:G4k1}
\end{equation}

\subsubsection{$Y(\mathfrak{sp}_{2n},\mathfrak{sp}_{2m}\oplus\mathfrak{sp}_{2n-2m})$}

This case is very similar to $Y(\mathfrak{so}_{2n},\mathfrak{gl}_{n})$
(see the asymptotic expansions). In this case, the nesting proceeds
as follows.
\[
\begin{array}{ccccccc}
Y(\mathfrak{sp}_{2n},\mathfrak{sp}_{2m}\oplus\mathfrak{sp}_{2n-2m}) & \to & Y(\mathfrak{sp}_{2n-4},\mathfrak{sp}_{2m-2}\oplus\mathfrak{sp}_{2n-2m-2}) & \to & \dots & \to & Y(\mathfrak{sp}_{2n-4m},\mathfrak{sp}_{2n-4m})\\
\mathbf{K}^{(1)}\equiv\mathbf{K} & \to & \mathbf{K}^{(3)} & \to & \dots & \to & \mathbf{K}^{(2m+1)}\\
\downarrow &  & \downarrow &  & \dots\\
\mathbf{k}^{(1)} &  & \mathbf{k}^{(2)} &  & \dots
\end{array}
\]
In the $m+1$-th step, we obtain a $Y(\mathfrak{sp}_{2n-4m},\mathfrak{sp}_{2n-4m})$
$K$-matrix. This means that the MPS is a singlet state with respect
to the $\mathfrak{sp}_{2n-4m}$ subalgebra. The $K$-matrix of the
$\mathfrak{gl}_{2m}$ sector is a representation of $Y^{-}(2m)$,
so there are selection rules similar to (\ref{eq:selrule-1}). Based
on these, the overlap is non-zero only if
\begin{equation}
\begin{split}r_{2s-1} & =\frac{\Lambda_{2s-1}-\Lambda_{2s}+r_{2s-2}+r_{2s}}{2},\quad\text{for }s=1,\dots,m,\\
r_{s} & =r_{2m}+\sum_{k=2m+1}^{s}\Lambda_{k},\quad\text{for }s=2m+1,\dots,n-1,\\
r_{n} & =\frac{1}{2}r_{2m}+\frac{1}{2}\sum_{k=2m+1}^{n}\Lambda_{k}.
\end{split}
\end{equation}

The nested $K$-matrices $\mathbf{K}^{(3)},\dots,\mathbf{K}^{(2m+1)}$
can be obtained using the recursion (\ref{eq:soNrec}). From these,
we select $m$ number of $Y^{-}(2)$ $K$-matrices (\ref{eq:Ym2k}).
The $G$-operator formulas (\ref{eq:Gop0}) cannot be applied directly.
Therefore, we deform these $Y^{-}(2)$ $K$-matrices using (\ref{eq:kdef}),
and the $G$-operators are well-defined for them using (\ref{eq:GYm2})
for $s=1,\dots,m$. These define the $G$-operators $\mathbf{G}^{(s)}$
for $s=1,\dots,2m$.

The operator $\mathbf{K}^{(2m+1)}$ can be deformed using the following
scalar solution of the $\mathfrak{sp}_{2\tilde{n}}$ reflection equation:
\begin{equation}
K(u)=1+u\sum_{j=1}^{\tilde{n}}e_{j,-j},
\end{equation}
where $\tilde{n}=n-2m$. The deformed $K$-matrix is
\begin{equation}
\begin{split}\mathbf{K}_{a,b}^{(2m+1)}(u) & =\sum_{c=-\tilde{n}}^{\tilde{n}}\theta_{a}\theta_{c}\mathbf{L}_{-c,-a}(u-\kappa_{2\tilde{n}})\mathbf{L}_{c,b}(-u),\\
\tilde{\mathbf{K}}_{a,b}^{(1)}(u) & =\mathbf{K}_{a,b}^{(2m+1)}(u)+u\sum_{c=1}^{k}\theta_{a}\theta_{c}\mathbf{L}_{-c,-a}(u-\kappa_{2\tilde{n}})\mathbf{L}_{-c,b}(-u).
\end{split}
\end{equation}
From here, the nested $K$-matrices $\tilde{\mathbf{K}}^{(2)},\dots,\tilde{\mathbf{K}}^{(\tilde{n})}$
can be defined using the recursion (\ref{eq:soNrec0}), and the remaining
$G$-operators are given by
\begin{equation}
\mathbf{G}^{(2m+s)}(u)=\tilde{\mathbf{K}}_{\tilde{n}-s+1,-\tilde{n}+s-1}^{(s)}(u),
\end{equation}
for $s=1,\dots,\tilde{n}$ and
\begin{equation}
\mathbf{G}^{(n+1)}(u)=\tilde{\mathbf{K}}_{-1,1}^{(\tilde{n})}(u)-\tilde{\mathbf{K}}_{-1,-1}^{(\tilde{n})}(u)\left[\tilde{\mathbf{K}}_{1,-1}^{(\tilde{n})}(u)\right]^{-1}\tilde{\mathbf{K}}_{1,1}^{(\tilde{n})}(u).
\end{equation}

\section{Conclusions}

In this paper, we presented the universal formula (\ref{eq:onOV-1})
that provides the overlaps between integrable MPSs and Bethe states.
The formula consists of a ratio of Gaudin determinants and a prefactor.
The prefactor contains the one-particle overlap functions $\mathcal{F}_{\ell}^{(\nu)}(u)$,
which are the eigenvalues of the commuting $F$-operators $\mathbf{F}^{(\nu)}(u)$.
The formula is surprisingly broadly applicable, as the $F$-operators
involved depend only on the quasi-determinants of the $K$-matrix
associated with the MPS.

The main result of the paper is a precise proof of the universal formula
for those MPSs of $\mathfrak{gl}_{N}$ symmetric spin chains where
all quasi-determinants exist. This applies to a large family of possible
integrable MPSs. However, we observed that for certain reflection
algebras, these quasi-determinants are not well-defined because some
operators are non-invertible. In such cases, we showed that through
continuous deformations, these operators can be made invertible, allowing
the quantities in the universal overlap formula to be determined.
We also generalized the definitions of the $F$-operators to orthogonal
and symplectic spin chains.

In the future, it would be worthwhile to extend the proofs to the
orthogonal and symplectic cases as well. This requires knowledge of
recursion, action, and coproduct formulas for the Bethe vectors. With
these formulas in hand, the proofs in the current paper could likely
be applied without difficulty. For $\mathfrak{so}_{2n+1}$ symmetric
spin chains, these formulas are already available \cite{Liashyk:2024hhz,Liashyk:2025oao}.
Additionally, it may be worthwhile to extend the results to graded
and trigonometric spin chains. The necessary formulas for Bethe vectors
are already known in certain cases \cite{Hutsalyuk:2017tcx,Liashyk:2021tth}. 

It is worth noting that the formulas presented in the paper are not
applicable in the case of twisted boundary conditions. In these cases,
integrability does not manifest at the level of Bethe roots (pair
structure), but rather in the full $Q$-system \cite{Gombor:2021uxz}.
In such cases, the overlaps of the tensor product states cannot be
expressed with Gaudin determinants but rather with a determinant involving
$Q$-functions \cite{Ekhammar:2023iph}. Currently, one systematic
method is available for computing MPS overlaps, which is based on
searching for generalized dressing formulas \cite{deLeeuw:2019ebw,Gombor:2024zru}.
The advantage of this method is that it can be applied even in the
presence of twists; the disadvantage is that it must be carried out
separately for each MPS, making it less general than the result of
the present paper. A potential research direction could be the search
for a similarly universal overlap formula under twisted boundary conditions.

\section*{Acknowledgments}

The author would like to thank Vidas Regelskis for the helpful discussions
about reflection algebras. This paper was supported by the NKFIH grant
PD142929 and the János Bolyai Research Scholarship of the Hungarian
Academy of Science.

\appendix

\section{Boundary states for arbitrary representations\label{sec:Boundary-states-gen-rep}}

In this section, we focus on finite-dimensional representations of
the Lie algebra $\mathfrak{gl}_{N}$, meaning that $\Lambda_{k}-\Lambda_{k+1}\in\mathbb{N}$.
Without loss of generality, we assume $\Lambda_{N}=0$, so $\Lambda_{k}\in\mathbb{N}$.
Every finite-dimensional representation corresponds to a Young diagram
$Y(\Lambda)$ , where the $k$-th row contains $\Lambda_{k}$ boxes.
The total number of boxes in the diagram is $n=\sum_{k=1}^{N-1}\Lambda_{k}$.
We introduce an $n$-fold tensor product space $(\mathbb{C}^{N})^{\otimes n}$,
indexed by $1,2,\dots,n$. A box $x\in Y(\Lambda)$ in the Young diagram
is identified by its row $i$ and column $j$, i.e., $x=(i,j)$. Additionally,
we can label the boxes as $x_{j}\in Y(\Lambda)$, where $j=1,\dots,n$.
The label of the box at position $(a,b)$ is given by $\sum_{j=1}^{b-1}w_{j}+a$,
where $w_{j}$ is the height of the $j$-th column. For example, the
Young diagram corresponding to $\Lambda=(3,2,2,1,0)$ can be labeled
in the following way.
\begin{center}
\begin{tabular}{|c|c|c}
\hline 
1 & 5 & \multicolumn{1}{c|}{8}\tabularnewline
\hline 
2 & 6 & \tabularnewline
\cline{1-2} \cline{2-2} 
3 & 7 & \tabularnewline
\cline{1-2} \cline{2-2} 
4 & \multicolumn{1}{c}{} & \tabularnewline
\cline{1-1} 
\end{tabular}
\par\end{center}

Each space also has an associated inhomogeneity $s_{j}=k-l$, if $x_{j}=(k,l)$.
We also introduce a shorthand notation for the $R$-matrices:
\begin{equation}
R_{j,k}\equiv R_{j,k}(u_{j}-u_{k}).
\end{equation}
From these, a projection operator can be constructed
\begin{equation}
\Pi_{1,2,\dots,,n}^{\Lambda}=\frac{1}{A}\lim_{u_{j}\to s_{j}}\left(R_{n-1,n}\right)\left(R_{n-2,n}R_{n-2,n-1}\right)\dots\left(R_{1,n}\dots R_{1,3}R_{1,2}\right),\to(\Pi_{1,2,\dots,,n}^{\Lambda})^{2}=\Pi_{1,2,\dots,,n}^{\Lambda},
\end{equation}
where $A$ is a normalization constant
\begin{equation}
A=\prod_{x\in Y(\Lambda)}\text{hook}(x).
\end{equation}
The subspace of the operator $\Pi_{1,2,\dots,,n}^{\Lambda}$ with
eigenvalue $+1$ is isomorphic to the irreducible representation $\Lambda$.

\subsection{Fusion of the $R$-matrix}

The two types of $R$-matrices (\ref{eq:Rm}), (\ref{eq:Rb}) are
representations of the Yangian algebra $Y(N)$, i.e.,
\begin{equation}
R_{1,2}(u_{1}-u_{2})\bar{R}_{1,3}(u_{1})\bar{R}_{2,3}(u_{2})=\bar{R}_{2,3}(u_{2})\bar{R}_{1,3}(u_{1})R_{1,2}(u_{1}-u_{2}).
\end{equation}
We can define fused $R$-matrices as follows
\begin{equation}
\begin{split}\bar{R}_{I,0}^{\Lambda,\square}(u) & =\left(\bar{R}_{n,0}(u+s_{n})\dots\bar{R}_{2,0}(u+s_{2})\bar{R}_{1,0}(u+s_{1})\right)\Pi_{1,2,\dots,n}^{\Lambda}=\\
 & =\Pi_{1,2,\dots,n}^{\Lambda}\left(\bar{R}_{1,0}(u+s_{1})\bar{R}_{2,0}(u+s_{2})\dots\bar{R}_{n,0}(u+s_{n})\right),
\end{split}
\end{equation}
and
\begin{equation}
\begin{split}\bar{R}_{0,I}^{\square,\Lambda}(u) & =\left(\bar{R}_{0,1}(u-s_{1})\bar{R}_{0,2}(u-s_{2})\dots\bar{R}_{0,n}(u-s_{n})\right)\Pi_{1,2,\dots,n}^{\Lambda}=\\
 & =\Pi_{1,2,\dots,n}^{\Lambda}\left(\bar{R}_{0,n}(u-s_{n})\dots\bar{R}_{0,2}(u-s_{2})\bar{R}_{0,1}(u-s_{1})\right),
\end{split}
\end{equation}
where $I$ denotes the $+1$ eigenspace of the operator $\Pi_{1,2,\dots,,n}^{\Lambda}$.
The matrices $R^{\Lambda,\square}$ and $R^{\square,\Lambda}$ satisfy
the following relations
\begin{equation}
\begin{split}R_{1,2}(u_{1}-u_{2})R_{1,3}^{\square,\Lambda}(u_{1}-u_{3})R_{2,3}^{\square,\Lambda}(u_{2}-u_{3}) & =R_{2,3}^{\square,\Lambda}(u_{2}-u_{3})R_{1,3}^{\square,\Lambda}(u_{1}-u_{3})R_{1,2}(u_{1}-u_{2}),\\
R_{1,2}^{\Lambda,\square}(u_{1}-u_{2})R_{1,3}^{\Lambda,\square}(u_{1}-u_{3})R_{2,3}(u_{2}-u_{3}) & =R_{2,3}(u_{2}-u_{3})R_{1,3}^{\Lambda,\square}(u_{1}-u_{3})R_{1,2}^{\Lambda,\square}(u_{1}-u_{2}),\\
R_{1,2}^{\square,\Lambda}(u_{1}-u_{2})R_{1,3}(u_{1}-u_{3})R_{2,3}^{\Lambda,\square}(u_{2}-u_{3}) & =R_{2,3}^{\Lambda,\square}(u_{2}-u_{3})R_{1,3}(u_{1}-u_{3})R_{1,2}^{\square,\Lambda}(u_{1}-u_{2}),
\end{split}
\end{equation}
and unitarity
\begin{equation}
\begin{split}R_{1,2}^{\square,\Lambda}(u)R_{2,1}^{\Lambda,\square}(-u) & =\frac{u-(N-1)}{u}\prod_{k=1}^{N-1}\frac{u+(\Lambda_{k}-k+1)}{u+(\Lambda_{k}-k)},\\
R_{1,2}^{\Lambda,\square}(u)R_{2,1}^{\square,\Lambda}(-u) & =\frac{u+(N-1)}{u}\prod_{k=1}^{N-1}\frac{u-(\Lambda_{k}-k+1)}{u-(\Lambda_{k}-k)}.
\end{split}
\end{equation}
The matrix $\widehat{R}^{\square,\Lambda}$ is also a representation
of the Yangian algebra:
\begin{equation}
\begin{split}R_{1,2}(u_{1}-u_{2})\widehat{R}_{1,3}^{\square,\Lambda}(u_{1}-u_{3})\widehat{R}_{2,3}^{\square,\Lambda}(u_{2}-u_{3}) & =\widehat{R}_{2,3}^{\square,\Lambda}(u_{2}-u_{3})\widehat{R}_{1,3}^{\square,\Lambda}(u_{1}-u_{3})R_{1,2}(u_{1}-u_{2}),\\
\widehat{R}_{1,2}(u_{1}-u_{2})\widehat{R}_{1,3}^{\square,\Lambda}(u_{1}-u_{3})R_{2,3}^{\square,\Lambda}(u_{2}-u_{3}) & =R_{2,3}^{\square,\Lambda}(u_{2}-u_{3})\widehat{R}_{1,3}^{\square,\Lambda}(u_{1}-u_{3})\widehat{R}_{1,2}(u_{1}-u_{2}).
\end{split}
\end{equation}
The partial transpose of $\widehat{R}^{\square,\Lambda}$ can be expressed
using the matrix $R_{I,0}^{\Lambda,\square}$
\begin{equation}
\left(\widehat{R}_{0,I}^{\square,\Lambda}(u)\right)^{t_{0}}=\left(R_{0,n}(-u+s_{n})\dots R_{0,2}(-u+s_{2})R_{0,1}(-u+s_{1})\right)\Pi_{1,2,\dots,n}^{\Lambda}=R_{I,0}^{\Lambda,\square}(-u),
\end{equation}
i.e.,
\begin{equation}
R_{1,2}^{\square,\Lambda}(u)\left(\widehat{R}_{1,2}^{\square,\Lambda}(u)\right)^{t_{1}}=\frac{u-(N-1)}{u}\prod_{k=1}^{N-1}\frac{u+(\Lambda_{k}-k+1)}{u+(\Lambda_{k}-k)}.
\end{equation}
The fused $R$-matrix $R^{\square,\Lambda}$ coincides with the previously
defined Lax operator (\ref{eq:Lax})
\begin{equation}
L^{\Lambda}(u)=R^{\square,\Lambda}(u).
\end{equation}
This is a highest-weight representation of the Yangian, and the highest-weight
vector is
\begin{equation}
|0_{\Lambda}\rangle=\Pi_{1,2,\dots,,n}^{\Lambda}|1,2,\dots,w_{1},1,2,\dots,w_{2},\dots,1,2,\dots,w_{m}\rangle.
\end{equation}
The highest weights can be obtained from the action of $L_{i,i}^{\Lambda}(u)$:
\begin{equation}
L_{i,i}^{\Lambda}(u)|0_{\Lambda}\rangle=\lambda_{i}^{\Lambda}(u)|0_{\Lambda}\rangle\to\lambda_{i}^{\Lambda}(u)=\frac{u+\Lambda_{i}}{u}.
\end{equation}
From the inversion relation, we get that
\begin{equation}
\widehat{L}^{\Lambda}(u)=\frac{u+\Lambda_{1}}{u}\frac{u-\Lambda_{1}}{u-(N-1)}\prod_{k=1}^{N-1}\frac{u+(\Lambda_{k}-k)}{u+(\Lambda_{k}-k+1)}\widehat{R}^{\square,\Lambda}(u),
\end{equation}
where $\widehat{L}^{\Lambda}$ is defined as in (\ref{eq:Thatdef}).
The vector $|0_{\Lambda}\rangle$ is the lowest-weight vector of the
operator $\widehat{L}^{\Lambda}$, and the lowest weight is given
by (\ref{eq:lamhat})
\begin{equation}
\begin{split}\hat{\lambda}_{i}^{\Lambda}(u) & =\frac{\lambda_{1}^{\Lambda}(u)\lambda_{1}^{\Lambda}(-u)}{\lambda_{i}^{\Lambda}(u-(i-1))}\prod_{k=1}^{i-1}\frac{\lambda_{k}^{\Lambda}(u-k)}{\lambda_{k}^{\Lambda}(u-(k-1))}\\
 & =\frac{u+\Lambda_{1}}{u+(\Lambda_{i}-i+1)}\frac{u-\Lambda_{1}}{u}\prod_{k=1}^{i-1}\frac{u+(\Lambda_{k}-k)}{u+(\Lambda_{k}-k+1)}.
\end{split}
\label{eq:suly}
\end{equation}
For the Lax operator $\left(\widehat{L}_{0,I}^{\Lambda}\right)^{t_{I}}$,
the vector $|0_{\Lambda}\rangle$ becomes the highest-weight vector,
and the highest weight is given by (\ref{eq:suly}).

\subsection{Fusion of the $K$-matrix}

Fusion can also be applied to the reflection equation
\begin{equation}
R_{1,2}(v-u)\mathbf{K}_{1}(u)\bar{R}_{1,2}(-u-v)\mathbf{K}_{2}(v)=\mathbf{K}_{2}(v)\bar{R}_{1,2}(-u-v)\mathbf{K}_{1}(u)R_{1,2}(v-u).
\end{equation}
The fused reflection matrix is defined as
\begin{equation}
\begin{split}\mathbf{K}_{I}^{\Lambda}(u) & =\lim_{u_{j}\to u+s_{j}}\mathbf{K}_{1}\left(\bar{R}_{1,2}\mathbf{K}_{2}\right)\left(\bar{R}_{1,3}\bar{R}_{2,3}\mathbf{K}_{3}\right)\dots\left(\bar{R}_{1,n}\bar{R}_{2,n}\dots\bar{R}_{n-1,n}\mathbf{K}_{n}\right)\Pi_{1,2,\dots,n}^{\Lambda}\\
 & =\Pi_{1,2,\dots,n}^{\Lambda}\left(\mathbf{K}_{n}\bar{R}_{n-1,n}\dots\bar{R}_{2,n}\bar{R}_{1,n}\right)\dots\left(\mathbf{K}_{3}\bar{R}_{2,3}\bar{R}_{1,3}\right)\left(\mathbf{K}_{2}\bar{R}_{1,2}\right)\mathbf{K}_{1},
\end{split}
\end{equation}
where
\[
\mathbf{K}_{j}\equiv\mathbf{K}_{j}(u_{j}),\quad\bar{R}_{j,k}\equiv\bar{R}_{j,k}(-u_{j}-u_{k}).
\]
This $K$-matrix satisfies the reflection equation
\begin{equation}
L_{0,1}^{\Lambda}(u-\theta)\mathbf{K}_{1}^{\Lambda}(\theta)\bar{L}_{0,1}^{\Lambda}(-u-\theta)\mathbf{K}_{0}(u)=\mathbf{K}_{0}(u)\bar{L}_{0,1}^{\Lambda}(-u-\theta)\mathbf{K}_{1}^{\Lambda}(\theta)L_{0,1}^{\Lambda}(u-\theta).
\end{equation}
This equation is equivalent to the relation
\begin{equation}
\mathbf{K}_{0}(u)\psi_{2,1}^{\Lambda}(\theta)\left[\bar{L}_{0,2}^{\Lambda}(-u-\theta)\right]^{t_{2}}L_{0,1}^{\Lambda}(u-\theta)=\psi_{2,1}^{\Lambda}(\theta)\left[L_{0,2}^{\Lambda}(u-\theta)\right]^{t_{2}}\bar{L}_{0,1}^{\Lambda}(-u-\theta)\mathbf{K}_{0}(u),
\end{equation}
where we introduce a matrix-valued two-site state.
\begin{equation}
\psi_{1,2}^{\Lambda}(\theta)=\sum_{i,j=1}^{d_{\Lambda}}\langle i,j|\otimes\mathbf{K}_{i,j}^{\Lambda}(\theta).
\end{equation}

\subsection{Crossed KT-relation}

In the crossed case, the fused $K$-matrix satisfies the following
equation
\begin{equation}
\mathbf{K}_{0}(u)\psi_{2,1}^{\Lambda}(\theta)\left[\widehat{L}_{0,2}^{\Lambda}(-u-\theta)\right]^{t_{2}}L_{0,1}^{\Lambda}(u-\theta)=\psi_{2,1}^{\Lambda}(\theta)\left[L_{0,2}^{\Lambda}(u-\theta)\right]^{t_{2}}\widehat{L}_{0,1}^{\Lambda}(-u-\theta)\mathbf{K}_{0}(u).
\end{equation}
This crossed $KT$-relation applies to a two-site spin chain, where
the monodromy matrices are
\begin{equation}
T_{0}(u)=\left[\widehat{L}_{0,2}^{\Lambda}(-u-\theta)\right]^{t_{2}}L_{0,1}^{\Lambda}(u-\theta),\quad\widehat{T}_{0}(u)=\left[L_{0,2}^{\Lambda}(-u-\theta)\right]^{t_{2}}\widehat{L}_{0,1}^{\Lambda}(u-\theta).
\end{equation}
It is easy to verify that these matrices satisfy the required relations,
namely the $RTT$-relation (\ref{eq:RTT}) and the inversion relation
(\ref{eq:Thatdef}). The pseudo-vacuum eigenvalues are
\begin{equation}
\lambda_{k}(u)=\lambda_{k}^{\Lambda}(u-\theta)\hat{\lambda}_{k}^{\Lambda}(-u-\theta),\quad\hat{\lambda}_{k}(u)=\hat{\lambda}_{k}^{\Lambda}(u-\theta)\lambda_{k}^{\Lambda}(-u-\theta).
\end{equation}
In this two-site quantum space, the pseudo-vacuum is
\begin{equation}
|0\rangle=|0_{\Lambda}\rangle\otimes|0_{\Lambda}\rangle.
\end{equation}
If $\Lambda$ is a fundamental representation, i.e., $\Lambda=\mu^{(k)}$
where
\begin{equation}
\mu_{j}^{(k)}=\begin{cases}
1, & j\leq k,\\
0, & j>k,
\end{cases}
\end{equation}
then the highest weight vector is
\begin{equation}
|0_{\mu^{(k)}}\rangle=\sum_{i_{1},\dots,i_{k}}\epsilon_{i_{1},\dots,i_{k}}|i_{1},\dots,i_{k}\rangle,
\end{equation}
where $\epsilon$ a Levi-Civita tensor. In these cases, the fused
transfer matrices coincide with the Sklyanin minors
\begin{equation}
\mathbf{K}^{\mu^{(k)}}\to\mathbf{k}_{i_{1},i_{2},\dots,i_{k}}^{j_{1},j_{2},\dots,j_{k}},
\end{equation}
which are defined as
\begin{equation}
\Pi_{1,2,\dots,,k}^{\mu^{(k)}}\left(\mathbf{K}_{k}\widehat{R}_{k-1,k}\dots\widehat{R}_{2,k}\widehat{R}_{1,k}\right)\dots\left(\mathbf{K}_{3}\widehat{R}_{2,3}\widehat{R}_{1,3}\right)\left(\mathbf{K}_{2}\widehat{R}_{1,2}\right)\mathbf{K}_{1}=\sum_{a_{i},b_{i}}e_{a_{1},b_{1}}\otimes\dots\otimes e_{a_{k},b_{k}}\otimes\mathbf{k}_{a_{1},\dots,a_{k}}^{b_{1},\dots,b_{k}},\label{eq:SklyaninMin}
\end{equation}
where $\mathbf{k}$ is antisymmetric under the exchanges $i_{k}\leftrightarrow i_{l}$
and $j_{k}\leftrightarrow j_{l}$. The pseudo-vacuum overlap is
\begin{equation}
\psi_{2,1}^{\mu^{(k)}}(\theta)|0\rangle=\mathbf{k}_{1,2,\dots,k}^{1,2,\dots,k}(\theta),
\end{equation}
which is the Sklyanin determinant in the subalgebra generated by $\left\{ \mathbf{K}_{i,j}\right\} _{i,j=1}^{k}$.
For a general representation $\Lambda=\sum_{k=1}^{N-1}d_{k}\mu^{(k)}$,
the pseudo-vacuum overlap is
\begin{equation}
\mathbf{B}^{\Lambda}(\theta)=\psi_{2,1}^{\Lambda}(\theta)|0\rangle=B(\theta)\prod_{k=1}^{N-1}\prod_{l=1}^{d_{k}}\mathbf{k}_{1,2,\dots,k}^{1,2,\dots,k}(\theta-(D_{k}+l-1)),
\end{equation}
where $D_{k}=\sum_{l=k+1}^{N-1}d_{k}$, and $B(\theta)$ is a scalar
depending on $\theta$. Since the normalization of the MPS is arbitrary,
the specific value of $B(\theta)$ is irrelevant.  

Using the coproduct property of boundary states (Lemma \ref{lem:co-prod}),
we can construct more general boundary states
\begin{equation}
\langle\Psi|=\psi_{2J,2J-1}^{\Lambda^{(J)}}(\theta_{J})\dots\psi_{4,3}^{\Lambda^{(2)}}(\theta_{2})\psi_{2,1}^{\Lambda^{(1)}}(\theta_{1}),
\end{equation}
where the $\Lambda^{(j)}$ are arbitrary representations and the $\theta_{j}$
are arbitrary inhomogeneities. The monodromy matrices are
\begin{equation}
\begin{split}T_{0}(u) & =\left[\widehat{L}_{0,2J}^{\Lambda^{(J)}}(-u-\theta_{J})\right]^{t_{2J}}L_{0,2J-1}^{\Lambda^{(J)}}(u-\theta_{J})\dots\left[\widehat{L}_{0,2}^{\Lambda^{(1)}}(-u-\theta_{1})\right]^{t_{2}}L_{0,1}^{\Lambda^{(1)}}(u-\theta_{1}),\\
\widehat{T}_{0}(u) & =\left[L_{0,2J}^{\Lambda^{(J)}}(-u-\theta_{J})\right]^{t_{2J}}\widehat{L}_{0,2J-1}^{\Lambda^{(J)}}(u-\theta_{J})\dots\left[L_{0,2}^{\Lambda^{(1)}}(-u-\theta_{1})\right]^{t_{2}}\widehat{L}_{0,1}^{\Lambda^{(1)}}(u-\theta_{1}).
\end{split}
\end{equation}
The pseudo-vacuum eigenvalues are
\begin{equation}
\lambda_{k}(u)=\prod_{j=1}^{J}\lambda_{k}^{\Lambda^{(j)}}(u-\theta_{j})\hat{\lambda}_{k}^{\Lambda^{(j)}}(-u-\theta_{j}),\quad\hat{\lambda}_{k}(u)=\prod_{j=1}^{J}\hat{\lambda}_{k}^{\Lambda^{(j)}}(u-\theta_{j})\lambda_{k}^{\Lambda^{(j)}}(-u-\theta_{j}).
\end{equation}
We observe a symmetry property
\begin{equation}
\lambda_{k}(u)=\hat{\lambda}_{k}(-u).
\end{equation}
The pseudo-vacuum overlap is
\begin{equation}
\mathbf{B}(\theta)=\langle\Psi|0\rangle=\prod_{j=1}^{J}\mathbf{B}^{\Lambda^{(j)}}(\theta_{j}).\label{eq:vacuumOv}
\end{equation}

\subsection{Uncrossed KT-relation}

In the uncrossed case, the fused $K$-matrix satisfies the following
equation
\begin{equation}
\mathbf{K}_{0}(u)\psi_{2,1}^{\Lambda}(\theta)\left[L_{0,2}^{\Lambda}(-u-\theta)\right]^{t_{2}}L_{0,1}^{\Lambda}(u-\theta).=\psi_{2,1}^{\Lambda}(\theta)\left[L_{0,2}^{\Lambda}(u-\theta)\right]^{t_{2}}L_{0,1}^{\Lambda}(-u-\theta)\mathbf{K}_{0}(u).\label{eq:uncrKT}
\end{equation}
We evaluate the transposed Lax operator.
\begin{align}
\left[L_{1,2}^{\Lambda}(-u)\right]^{t_{2}} & =e_{i,j}\otimes\left(\delta_{i,j}\mathbf{1}-\frac{1}{u}(E_{j,i}^{\Lambda})^{t}\right)=e_{i,j}\otimes\left(\delta_{i,j}\mathbf{1}+\frac{1}{u}E_{j,i}^{\Lambda^{cg}}\right)\nonumber \\
 & =L_{1,2}^{\Lambda^{cg}}(u).
\end{align}
Here we introduced the contra-gradient representation
\begin{equation}
E_{i,j}^{\Lambda^{cg}}:=-(E_{i,j}^{\Lambda})^{t}.
\end{equation}
Let the highest weight vector of the contra-gradient representation
be $|\bar{0}\rangle^{\Lambda}$, for which
\begin{equation}
\begin{split}E_{i,j}^{\Lambda^{cg}}|\bar{0}_{\Lambda}\rangle & =0,\quad\text{for }i<j,\\
E_{i,i}^{\Lambda^{cg}}|\bar{0}_{\Lambda}\rangle & =-\Lambda_{N+1-i}|\bar{0}_{\Lambda}\rangle.
\end{split}
\end{equation}
We see that (\ref{eq:uncrKT}) is an uncrossed $KT$-relation for
a two-site spin chain, where the monodromy matrix is
\begin{equation}
T_{0}(u)=L_{0,2}^{\Lambda^{cg}}(u+\theta)L_{0,1}^{\Lambda}(u-\theta).
\end{equation}
In this two-site quantum space, the pseudo-vacuum is
\begin{equation}
|0\rangle=|0_{\Lambda}\rangle\otimes|\bar{0}_{\Lambda}\rangle.
\end{equation}
The pseudo-vacuum eigenvalue is
\begin{equation}
\lambda_{k}(u)=\frac{u-\theta+\Lambda_{k}}{u-\theta}\frac{u+\theta-\Lambda_{N+1-k}}{u+\theta}.
\end{equation}
For fundamental representations, we can introduce the same notation
for the components of the fused $K$-matrix as before. The highest
weight vectors are
\begin{equation}
\begin{split}|0_{\mu^{(k)}}\rangle & =\sum_{i_{1},\dots,i_{k}=1}^{k}\epsilon_{i_{1},\dots,i_{k}}|i_{1},\dots,i_{k}\rangle,\\
|\bar{0}_{\mu^{(k)}}\rangle & =\sum_{i_{1},\dots,i_{k}=1}^{k}\epsilon_{i_{1},\dots,i_{k}}|\bar{i}_{1},\dots,\bar{i}_{k}\rangle,
\end{split}
\end{equation}
where $\bar{i}=N+1-i$. Since the construction is symmetric under
the exchange $\Lambda\leftrightarrow\Lambda^{cg}$ we can assume without
loss of generality that $k\leq N/2$. The pseudo-vacuum overlap is
\begin{equation}
\psi_{2,1}^{\mu^{(k)}}(\theta)|0\rangle=\mathbf{k}_{\bar{1},\bar{2},\dots,\bar{k}}^{1,2,\dots,k}(\theta),
\end{equation}
where the quantum minor is defined as follows
\begin{equation}
\Pi_{1,2,\dots,,k}^{\mu^{(k)}}\left(\mathbf{K}_{k}R_{k-1,k}\dots R_{2,k}R_{1,k}\right)\dots\left(\mathbf{K}_{3}R_{2,3}R_{1,3}\right)\left(\mathbf{K}_{2}R_{1,2}\right)\mathbf{K}_{1}=\sum_{a_{i},b_{i}}e_{a_{1},b_{1}}\otimes\dots\otimes e_{a_{k},b_{k}}\otimes\mathbf{k}_{a_{1},\dots,a_{k}}^{b_{1},\dots,b_{k}}.
\end{equation}
If $k\leq N/2$, then $\left\{ \mathbf{K}_{\bar{i},j}\right\} _{i,j=1}^{k}$
forms a $Y(k)$ subalgebra, so $\mathbf{k}_{\bar{1},\bar{2},\dots,\bar{k}}^{1,2,\dots,k}(\theta)$
is a quantum determinant in the $Y(k)$ subalgebra. For a general
representation $\Lambda=\sum_{k=1}^{N/2}d_{k}\mu^{(k)}$, the pseudo-vacuum
overlap is
\begin{equation}
\mathbf{B}^{\Lambda}(\theta)=\psi_{2,1}^{\Lambda}(\theta)|0\rangle=B(\theta)\prod_{k=1}^{N/2}\prod_{l=1}^{d_{k}}\mathbf{k}_{\bar{1},\bar{2},\dots,\bar{k}}^{1,2,\dots,k}(\theta-(D_{k}+l-1)),
\end{equation}
where $B(\theta)$ is a scalar depending on $\theta$.

Using the coproduct property of boundary states (\ref{eq:Psi_Factor}),
we can construct more general boundary states
\begin{equation}
\langle\Psi|=\psi_{2J,2J-1}^{\Lambda^{(J)}}(\theta_{J})\dots\psi_{4,3}^{\Lambda^{(2)}}(\theta_{2})\psi_{2,1}^{\Lambda^{(1)}}(\theta_{1}),
\end{equation}
where the $\Lambda^{(j)}$ are arbitrary representations and the $\theta_{j}$
are arbitrary inhomogeneities. The monodromy matrix is
\begin{equation}
T_{0}(u)=L_{0,2J}^{\Lambda^{(J),cg}}(u+\theta_{J})L_{0,2J-1}^{\Lambda^{(J)}}(u-\theta_{J})\dots L_{0,2}^{\Lambda^{(1),cg}}(u+\theta_{1})L_{0,1}^{\Lambda^{(1)}}(u-\theta_{1}).
\end{equation}
The pseudo-vacuum eigenvalues are as follows
\begin{equation}
\lambda_{k}(u)=\prod_{j=1}^{J}\frac{u-\theta_{j}+\Lambda_{k}^{(j)}}{u-\theta_{j}}\frac{u+\theta_{j}-\Lambda_{N+1-k}^{(j)}}{u+\theta_{j}}.
\end{equation}
We observe a symmetry property
\begin{equation}
\lambda_{k}(u)=\lambda_{k}(-u).
\end{equation}
The pseudo-vacuum overlap is
\begin{equation}
\mathbf{B}(\theta)=\langle\Psi|0\rangle=\prod_{j=1}^{J}\mathbf{B}^{\Lambda^{(j)}}(\theta_{j}).
\end{equation}

\section{Recursions for the off-shell Bethe states\label{sec:Reqursions-for-BS}}

In this section, we summarize the formulas for the off-shell Bethe
vectors that are necessary for our purposes. These formulas can be
found in the papers \cite{Hutsalyuk:2017tcx,Hutsalyuk:2020dlw}. The
Bethe vectors satisfy the following recursive equations

\begin{equation}
\mathbb{B}(\{z,\bar{t}^{1}\},\left\{ \bar{t}^{k}\right\} _{k=2}^{N-1})=\sum_{j=2}^{N}\frac{T_{1,j}(z)}{\lambda_{2}(z)}\sum_{\mathrm{part}(\bar{t})}\mathbb{B}(\bar{t}^{1},\left\{ \bar{t}_{\textsc{ii}}^{k}\right\} _{k=2}^{j-1},\left\{ \bar{t}^{k}\right\} _{k=j}^{N-1})\frac{\prod_{\nu=2}^{j-1}\alpha_{\nu}(\bar{t}_{\textsc{i}}^{\nu})g(\bar{t}_{\textsc{i}}^{\nu},\bar{t}_{\textsc{i}}^{\nu-1})f(\bar{t}_{\textsc{ii}}^{\nu},\bar{t}_{\textsc{i}}^{\nu})}{\prod_{\nu=1}^{j-1}f(\bar{t}^{\nu+1},\bar{t}_{\textsc{i}}^{\nu})},
\end{equation}
where the sum goes over the partitions $\bar{t}^{s}\vdash\{\bar{t}_{\textsc{i}}^{s},\bar{t}_{\textsc{ii}}^{s}\}$
for $s=2,\dots,j-1$ such that $\#\bar{t}_{\textsc{i}}^{s}=1$. We
have another recursion
\begin{equation}
\mathbb{B}(\left\{ \bar{t}^{k}\right\} _{k=1}^{N-2},\{z,\bar{t}^{N-1}\})=\sum_{j=1}^{N-1}\frac{T_{j,N}(z)}{\lambda_{N}(z)}\sum_{\mathrm{part}(\bar{t})}\mathbb{B}(\left\{ \bar{t}^{k}\right\} _{k=1}^{j-1},\left\{ \bar{t}_{\textsc{ii}}^{k}\right\} _{k=j}^{N-2},\bar{t}^{N-1})\frac{\prod_{\nu=j}^{N-2}g(\bar{t}_{\textsc{i}}^{\nu+1},\bar{t}_{\textsc{i}}^{\nu})f(\bar{t}_{\textsc{i}}^{\nu},\bar{t}_{\textsc{ii}}^{\nu})}{\prod_{\nu=j}^{N-1}f(\bar{t}_{\textsc{i}}^{\nu},\bar{t}^{\nu-1})},\label{eq:rec2}
\end{equation}
where the sum goes over the partitions $\bar{t}^{s}\vdash\{\bar{t}_{\textsc{i}}^{s},\bar{t}_{\textsc{ii}}^{s}\}$
for $s=j,\dots,N-2$ such that $\#\bar{t}_{\textsc{i}}^{s}=1$.

We also use the following action formula
\begin{multline}
T_{i,j}(z)\mathbb{B}(\bar{t})=\lambda_{N}(z)\sum_{\mathrm{part}(\bar{w})}\mathbb{B}(\bar{w}_{\textsc{ii}})\frac{\prod_{s=j}^{i-1}f(\bar{w}_{\textsc{i}}^{s},\bar{w}_{\textsc{iii}}^{s})}{\prod_{s=j}^{i-2}f(\bar{w}_{\textsc{i}}^{s+1},\bar{w}_{\textsc{iii}}^{s})}\times\\
\prod_{s=1}^{i-1}\frac{f(\bar{w}_{\textsc{i}}^{s},\bar{w}_{\textsc{ii}}^{s})}{h(\bar{w}_{\textsc{i}}^{s},\bar{w}_{\textsc{i}}^{s-1})f(\bar{w}_{\textsc{i}}^{s},\bar{w}_{\textsc{ii}}^{s-1})}\prod_{s=j}^{N-1}\frac{\alpha_{s}(\bar{w}_{\textsc{iii}}^{s})f(\bar{w}_{\textsc{ii}}^{s},\bar{w}_{\textsc{iii}}^{s})}{h(\bar{w}_{\textsc{iii}}^{s+1},\bar{w}_{\textsc{iii}}^{s})f(\bar{w}_{\textsc{ii}}^{s+1},\bar{w}_{\textsc{iii}}^{s})},\label{eq:act}
\end{multline}
where $\bar{w}^{\nu}=\{z,\bar{t}^{\nu}\}$. The sum goes over all
the partitions of $\bar{w}^{\nu}\vdash\left\{ \bar{w}_{\textsc{i}}^{\nu},\bar{w}_{\textsc{ii}}^{\nu},\bar{w}_{\textsc{iii}}^{\nu}\right\} $
where $\#\bar{w}_{\textsc{i}}^{\nu}=\Theta(i-1-\nu)$ and $\#\bar{w}_{\textsc{iii}}^{\nu}=\Theta(\nu-j)$.
We also set $\bar{w}_{\textsc{i}}^{0}=\bar{w}_{\textsc{iii}}^{N}=\{z\}$
and $\bar{w}_{\textsc{ii}}^{0}=\bar{w}_{\textsc{iii}}^{0}=\bar{w}_{\textsc{i}}^{N}=\bar{w}_{\textsc{ii}}^{N}=\emptyset$.
We also used the unit step function 
\[
\Theta(k)=\begin{cases}
1, & k\geq0,\\
0, & k<0.
\end{cases}
\]
We also need the action formula for the crossed monodromy matrix \cite{Liashyk:2018egk,Gombor:2021hmj}
\begin{multline}
\widehat{T}_{i,j}(z)\mathbb{B}(\bar{t})=(-1)^{i-j}\hat{\lambda}_{1}(z)\sum_{\mathrm{part}(\bar{w})}\mathbb{B}(\bar{w}_{\textsc{ii}})\frac{\prod_{s=2}^{N-1}f(\bar{t}^{s-1}-1,\bar{t}^{s})}{\prod_{s=2}^{N-1}f(\bar{w}_{\textsc{ii}}^{s-1}-1,\bar{w}_{\textsc{ii}}^{s})}\frac{\prod_{s=i}^{j-1}f(\bar{w}_{\textsc{i}}^{s},\bar{w}_{\textsc{iii}}^{s})}{\prod_{s=i+1}^{j-1}f(\bar{w}_{\textsc{i}}^{s-1}-1,\bar{w}_{\textsc{iii}}^{s})}\times\\
\prod_{s=i}^{N-1}\frac{f(\bar{w}_{\textsc{i}}^{s},\bar{w}_{\textsc{ii}}^{s})}{h(\bar{w}_{\textsc{i}}^{s},\bar{w}_{\textsc{i}}^{s+1}+1)f(\bar{w}_{\textsc{i}}^{s},\bar{w}_{\textsc{ii}}^{s+1}+1)}\prod_{s=1}^{j-1}\frac{\alpha_{s}(\bar{w}_{\textsc{iii}}^{s})f(\bar{w}_{\textsc{ii}}^{s},\bar{w}_{\textsc{iii}}^{s})}{h(\bar{w}_{\textsc{iii}}^{s-1}-1,\bar{w}_{\textsc{iii}}^{s})f(\bar{w}_{\textsc{ii}}^{s-1}-1,\bar{w}_{\textsc{iii}}^{s})},\label{eq:actTw}
\end{multline}
where $\bar{w}^{\nu}=\{z-\nu,\bar{t}^{\nu}\}$. The sum goes over
all the partitions of $\bar{w}^{\nu}\vdash\left\{ \bar{w}_{\textsc{i}}^{\nu},\bar{w}_{\textsc{ii}}^{\nu},\bar{w}_{\textsc{iii}}^{\nu}\right\} $
where $\#\bar{w}_{\textsc{i}}^{\nu}=\Theta(\nu-i)$, $\#\bar{w}_{\textsc{iii}}^{\nu}=\Theta(j-1-\nu)$.
We also set $\bar{w}_{\textsc{iii}}^{0}=\{z\}$, $\bar{w}_{\textsc{i}}^{N}=\{z-N\}$
and $\bar{w}_{\textsc{i}}^{0}=\bar{w}_{\textsc{ii}}^{0}=\bar{w}_{\textsc{ii}}^{N}=\bar{w}_{\textsc{iii}}^{N}=\emptyset$. 

Finally, we also need the co-product formula of the Bethe states.
Let $\mathcal{H}^{(1)},\mathcal{H}^{(2)}$ be two quantum spaces for
which $\mathcal{H}=\mathcal{H}^{(1)}\otimes\mathcal{H}^{(2)}$ and
the corresponding off-shell states are $\mathbb{B}^{(1)}(\bar{t}),\mathbb{B}^{(2)}(\bar{t})$.
The co-product formula reads as
\begin{equation}
\mathbb{B}(\bar{t})=\sum_{\mathrm{part}(\bar{t})}\frac{\prod_{\nu=1}^{N-1}\alpha_{\nu}^{(2)}(\bar{t}_{\textsc{i}}^{\nu})f(\bar{t}_{\textsc{ii}}^{\nu},\bar{t}_{\textsc{i}}^{\nu})}{\prod_{\nu=1}^{N-2}f(\bar{t}_{\textsc{ii}}^{\nu+1},\bar{t}_{\textsc{i}}^{\nu})}\mathbb{B}^{(1)}(\bar{t}_{\textsc{i}})\mathbb{B}^{(2)}(\bar{t}_{\textsc{ii}}),\label{eq:Bcoprod}
\end{equation}
where the sum goes over the partitions $\bar{t}^{s}\vdash\{\bar{t}_{\textsc{i}}^{s},\bar{t}_{\textsc{ii}}^{s}\}$.

\section{Theorems for the nested $K$-matrices}

In this section, we prove the theorems concerning nested $K$-matrices
from subsection \ref{subsec:Theorems-for-K}.

\subsection{Crossed $K$-matrices}
\begin{lem}
\label{lem:KK}The operator $\mathbf{K}_{1,1}(u)$ satisfies the following
commutation relations
\begin{equation}
\left[\mathbf{K}_{1,1}(u),\mathbf{B}\right]=\left[\mathbf{K}_{1,1}(u),\mathbf{K}_{1,1}(v)\right]=0.\label{eq:comKK}
\end{equation}
\end{lem}
\begin{proof}
Let us get the $(1,1)$ component of the $KT$-relation
\begin{equation}
\sum_{k=1}^{N}\mathbf{K}_{1,k}(z)\langle\Psi|T_{k,1}(z)=\sum_{k=1}^{N}\langle\Psi|\widehat{T}_{1,k}(-z)\mathbf{K}_{k,1}(z).
\end{equation}
Let us apply it on the pseudo-vacuum:
\begin{equation}
\mathbf{K}_{1,1}(z)\langle\Psi|T_{1,1}(z)|0\rangle=\langle\Psi|\widehat{T}_{1,1}(-z)|0\rangle\mathbf{K}_{1,1}(z),
\end{equation}
where we used that $T_{i,j}(z)|0\rangle=\widehat{T}_{j,i}(z)|0\rangle=0$
for $i>j$. The equation above simplifies as
\begin{equation}
\mathbf{K}_{1,1}(z)\mathbf{B}\lambda_{1}(z)=\mathbf{B}\mathbf{K}_{1,1}(z)\hat{\lambda}_{1}(-z).\label{eq:KTN1}
\end{equation}
Using the symmetry property (\ref{eq:lamProp}), we obtain that
\begin{equation}
[\mathbf{K}_{1,1}(z),\mathbf{B}]=0.\label{eq:commKB}
\end{equation}
From the nested $KT$-relations we can derive analogous equations
as (\ref{eq:commKB}). For the nested $K$-matrices we have
\begin{equation}
[\mathbf{K}_{k,k}^{(k)}(z),\mathbf{B}]=0.\label{eq:KB}
\end{equation}

For the other relation, we can use the $(1,1),(1,1)$ component of
the reflection equation (\ref{eq:crBY}). After simplifications
\begin{equation}
\left[\mathbf{K}_{1,1}(u),\mathbf{K}_{1,1}(v)\right]=0.
\end{equation}
\end{proof}
\begin{lem}
\label{lem:KG}The component $\mathbf{K}_{1,1}(u)$ commutes with
the entire nested $K$-matrix, i.e.,
\begin{equation}
\left[\mathbf{K}_{1,1}(v),\mathbf{K}_{a,b}^{(2)}(u)\right]=0,
\end{equation}
for $a,b=2,\dots,N$.
\end{lem}
\begin{proof}
Let us get the ($(1,1),(1,b)$), ($(1,1),(a,1))$ and $((1,1),(a,b))$
components of the boundary Yang-Baxter equation (\ref{eq:crBY}) ($a,b>1$)
\begin{align}
\frac{u-v-1}{u-v}\frac{u+v+1}{u+v}\mathbf{K}_{1,1}(u)\mathbf{K}_{1,b}(v) & =\mathbf{K}_{1,b}(v)\mathbf{K}_{1,1}(u)-\frac{1}{u-v}\frac{u+v+1}{u+v}\mathbf{K}_{1,1}(v)\mathbf{K}_{1,b}(u)\nonumber \\
 & +\frac{1}{u+v}\mathbf{K}_{1,1}(v)\mathbf{K}_{b,1}(u),\label{eq:Nb}\\
\mathbf{K}_{1,1}(u)\mathbf{K}_{a,1}(v) & =\frac{u-v-1}{u-v}\frac{u+v+1}{u+v}\mathbf{K}_{a,1}(v)\mathbf{K}_{1,1}(u)+\frac{1}{u-v}\frac{u+v+1}{u+v}\mathbf{K}_{a,1}(u)\mathbf{K}_{1,1}(v)\nonumber \\
 & -\frac{1}{u+v}\mathbf{K}_{1,a}(u)\mathbf{K}_{1,1}(v),\label{eq:a1}\\
\left[\mathbf{K}_{1,1}(u),\mathbf{K}_{a,b}(v)\right] & =\frac{1}{u-v}\frac{u+v+1}{u+v}\left(\mathbf{K}_{a,1}(u)\mathbf{K}_{1,b}(v)-\mathbf{K}_{a,1}(v)\mathbf{K}_{1,b}(u)\right)\nonumber \\
 & -\frac{1}{u+v}\left(\mathbf{K}_{1,a}(u)\mathbf{K}_{1,b}(v)-\mathbf{K}_{a,1}(v)\mathbf{K}_{b,1}(u)\right).\label{eq:ab}
\end{align}
Let us apply (\ref{eq:a1}) to obtain
\begin{align}
\mathbf{K}_{1,1}(u)\left(\mathbf{K}_{a,1}(v)\mathbf{K}_{1,1}^{-1}(v)\mathbf{K}_{1,b}(v)\right) & =\frac{u-v-1}{u-v}\frac{u+v+1}{u+v}\mathbf{K}_{a,1}(v)\mathbf{K}_{1,1}(u)\mathbf{K}_{1,1}^{-1}(v)\mathbf{K}_{1,b}(v)+\nonumber \\
 & +\frac{1}{u-v}\frac{u+v+1}{u+v}\mathbf{K}_{a,1}(u)\mathbf{K}_{1,b}(v)-\frac{1}{u+v}\mathbf{K}_{1,a}(u)\mathbf{K}_{1,b}(v).
\end{align}
On the first term on the right-hand side, we use the commutation relations
(\ref{eq:comKK}) and (\ref{eq:Nb})
\begin{multline}
\frac{u-v-1}{u-v}\frac{u+v+1}{u+v}\mathbf{K}_{a,1}(v)\mathbf{K}_{1,1}(u)\mathbf{K}_{1,1}^{-1}(v)\mathbf{K}_{1,b}(v)=\mathbf{K}_{a,1}(v)\mathbf{K}_{1,1}^{-1}(v)\left(\frac{u-v-1}{u-v}\frac{u+v+1}{u+v}\mathbf{K}_{1,1}(u)\mathbf{K}_{1,b}(v)\right)=\\
\left(\mathbf{K}_{a,1}(v)\mathbf{K}_{1,1}^{-1}(v)\mathbf{K}_{1,b}(v)\right)\mathbf{K}_{1,1}(u)-\frac{1}{u-v}\frac{u+v+1}{u+v}\mathbf{K}_{a,1}(v)\mathbf{K}_{1,b}(u)+\frac{1}{u+v}\mathbf{K}_{a,1}(v)\mathbf{K}_{b,1}(u).
\end{multline}
Combining the last two equations we have
\begin{align}
\left[\mathbf{K}_{1,1}(u),\mathbf{K}_{a,1}(v)\mathbf{K}_{1,1}^{-1}(v)\mathbf{K}_{1,b}(v)\right] & =\frac{1}{u-v}\frac{u+v+1}{u+v}\left(\mathbf{K}_{a,1}(u)\mathbf{K}_{1,b}(v)-\mathbf{K}_{a,1}(v)\mathbf{K}_{1,b}(u)\right)\nonumber \\
 & -\frac{1}{u+v}\left(\mathbf{K}_{1,a}(u)\mathbf{K}_{1,b}(v)-\mathbf{K}_{a,1}(v)\mathbf{K}_{b,1}(u)\right).
\end{align}
We can see the r.h.s. agrees with the r.h.s. of the commutation relation
(\ref{eq:ab}) therefore we just proved
\begin{equation}
\left[\mathbf{K}_{1,1}(u_{2}),\mathbf{K}_{a,b}^{(2)}(u_{1})\right]=0,
\end{equation}
for $a,b=2,\dots,N$. 
\end{proof}
Now we turn to the proof of Theorem \ref{thm:crossed-nested-K}.
\begin{proof}
First, we prove that the nested $K$-matrices $\mathbf{K}^{(k)}(u-(k-1)/2)$
form a representation of the algebra $Y^{+}(N+1-k)$. In the section
\ref{subsec:Quasi-determinants}, we saw that the nested $K$-matrices
can be expressed using quasi-determinants. Let us decompose the $K$-matrix
into block form
\begin{equation}
\mathbf{K}=\left(\begin{array}{cc}
\mathbf{A} & \mathbf{B}\\
\mathbf{C} & \mathbf{D}
\end{array}\right),
\end{equation}
where $\mathbf{A}$ is a $(k-1)\times(k-1)$ matrix and $\mathbf{D}$
is an $(N-k+1)\times(N-k+1)$ matrix. The nested $K$-matrices $\mathbf{K}^{(k)}$
can be expressed as follows
\begin{equation}
\mathbf{K}^{(k)}=\left(\begin{array}{cc}
\mathbf{A} & \mathbf{B}\\
\mathbf{C} & \boxed{\mathbf{D}}
\end{array}\right)=\mathbf{D}-\mathbf{C}\mathbf{A}^{-1}\mathbf{B}.\label{eq:nestedK}
\end{equation}
First, we invert the reflection equation
\begin{equation}
R_{1,2}(u-v)\mathbf{K}_{1}^{-1}(u)\widehat{R}_{1,2}(u+v+N)\mathbf{K}_{2}^{-1}(v)=\mathbf{K}_{2}^{-1}(v)\widehat{R}_{1,2}(u+v+N)\mathbf{K}_{1}^{-1}(u)R_{1,2}(u-v),
\end{equation}
where we used the unitarity relations of the $R$-matrices
\begin{equation}
R_{1,2}(u)R_{1,2}(-u)=\frac{u^{2}-1}{u^{2}},\quad\widehat{R}_{1,2}(u)\widehat{R}_{1,2}(-u+N)=1.
\end{equation}
We decompose the inverse matrix into block form
\begin{equation}
\mathbf{K}^{-1}=\left(\begin{array}{cc}
\tilde{\mathbf{A}} & \tilde{\mathbf{B}}\\
\tilde{\mathbf{C}} & \tilde{\mathbf{D}}
\end{array}\right),
\end{equation}
where $\tilde{\mathbf{D}}$ can be expressed in the following form
\begin{equation}
\tilde{\mathbf{D}}=\left(\mathbf{D}-\mathbf{C}\mathbf{A}^{-1}\mathbf{B}\right)^{-1}.\label{eq:Dtilde}
\end{equation}
The matrix $\tilde{\mathbf{D}}$ satisfies the following reflection
equation
\begin{equation}
R_{1,2}(u-v)\tilde{\mathbf{D}}_{1}(u)\widehat{R}_{1,2}(u+v+N)\tilde{\mathbf{D}}_{2}(v)=\tilde{\mathbf{D}}_{2}(v)\widehat{R}_{1,2}(u+v+N)\tilde{\mathbf{D}}_{1}(u)R_{1,2}(u-v),
\end{equation}
where $R$ and $\widehat{R}$ are $\mathfrak{gl}_{N-k+1}$ symmetric
$R$-matrices. We can invert the equation again
\begin{equation}
R_{1,2}(v-u)\tilde{\mathbf{D}}_{1}^{-1}(u)\widehat{R}_{1,2}(-u-v-k+1)\tilde{\mathbf{D}}_{2}^{-1}(v)=\tilde{\mathbf{D}}_{2}^{-1}(v)\widehat{R}_{1,2}(-u-v-k+1)\tilde{\mathbf{D}}_{1}^{-1}(u)R_{1,2}(v-u).\label{eq:BYBDti}
\end{equation}
We see that the operator $\tilde{\mathbf{D}}^{-1}(u-(k-1)/2)$ satisfies
the $\mathfrak{gl}_{N-k+1}$ reflection equation. Combining equations
(\ref{eq:nestedK}), (\ref{eq:Dtilde}) and (\ref{eq:BYBDti}), we
see that the operator $\mathbf{K}^{(k)}(u-(k-1)/2)$ is a representation
of the reflection algebra $Y^{+}(N-k+1)$.

Since the operator $\mathbf{K}^{(k)}$ satisfies the $\mathfrak{gl}_{N-k+1}$
reflection equation, the Lemmas \ref{lem:KK} and \ref{lem:KG} are
also applicable to the operator $\mathbf{K}^{(k)}$, i.e.,
\begin{equation}
\left[\mathbf{K}_{k,k}^{(k)}(v),\mathbf{K}_{k,k}^{(k)}(u)\right]=\left[\mathbf{K}_{k,k}^{(k)}(v),\mathbf{K}_{a,b}^{(k+1)}(u)\right]=0,
\end{equation}
where $a,b=k+1,\dots,N$.
\end{proof}

\subsection{Uncrossed $K$-matrices}
\begin{lem}
\label{lem:KK-1}The operator $\mathbf{K}_{N,1}(u)$ satisfies the
following commutation relations
\begin{equation}
\left[\mathbf{K}_{N,1}(u),\mathbf{B}\right]=\left[\mathbf{K}_{N,1}(u),\mathbf{K}_{N,1}(v)\right]=0.\label{eq:comKK-1}
\end{equation}
\end{lem}
\begin{proof}
Let us get the $(N,1)$ component of the $KT$-relation
\begin{equation}
\sum_{k=1}^{N}\mathbf{K}_{N,k}(z)\langle\Psi|T_{k,1}(z)=\sum_{k=1}^{N}\langle\Psi|T_{N,k}(-z)\mathbf{K}_{k,1}(z).
\end{equation}
Let us apply it on the pseudo-vacuum:
\begin{equation}
\mathbf{K}_{N,1}(z)\langle\Psi|T_{1,1}(z)|0\rangle=\langle\Psi|T_{N,N}(-z)|0\rangle\mathbf{K}_{N,1}(z),
\end{equation}
where we used that $T_{i,j}(z)|0\rangle=0$ for $i>j$. The equation
above simplifies as
\begin{equation}
\mathbf{K}_{N,1}(z)\mathbf{B}\lambda_{1}(z)=\mathbf{B}\mathbf{K}_{N,1}(z)\lambda_{N}(-z).\label{eq:KTN1-1}
\end{equation}
Using the symmetry property $\lambda_{1}(z)=\lambda_{N}(-z)$, we
obtain that
\begin{equation}
[\mathbf{K}_{N,1}(z),\mathbf{B}]=0.\label{eq:commKB-1}
\end{equation}
From the nested $KT$-relations we can derive analogous equations
as (\ref{eq:commKB-1}). For the nested $K$-matrices we have
\begin{equation}
[\mathbf{K}_{N+1-k,k}^{(k)}(z),\mathbf{B}]=0.\label{eq:KB-1}
\end{equation}

For the other relation we can use the $(N,1),(N,1)$ component of
the reflection equation (\ref{eq:crBY}). After simplifications
\[
\left[\mathbf{K}_{N,1}(u),\mathbf{K}_{N,1}(v)\right]=0.
\]
\end{proof}
\begin{lem}
\label{lem:KG-1}The component $\mathbf{K}_{N,1}(u)$ commutes with
the entire nested $K$-matrix, i.e.,
\begin{equation}
\left[\mathbf{K}_{N,1}(v),\mathbf{K}_{a,b}^{(2)}(u)\right]=0,\label{eq:comGK}
\end{equation}
for $a,b=2,\dots,N-1$. 
\end{lem}
\begin{proof}
Let us get the ($(N,1),(N,b)$), ($(N,1),(a,1))$ and $((N,1),(a,b))$
components of the boundary Yang-Baxter equation ($1<a,b<N$)
\begin{equation}
\frac{u-v-1}{u-v}\mathbf{K}_{N,1}(u)\mathbf{K}_{N,b}(v)=\mathbf{K}_{N,b}(v)\mathbf{K}_{N,1}(u)-\frac{1}{u-v}\mathbf{K}_{N,1}(v)\mathbf{K}_{1,b}(u)\label{eq:Nb-1}
\end{equation}
\begin{equation}
\mathbf{K}_{N,1}(u)\mathbf{K}_{a,1}(v)=\frac{u-v-1}{u-v}\mathbf{K}_{a,1}(v)\mathbf{K}_{N,1}(u)+\frac{1}{u-v}\mathbf{K}_{a,1}(u)\mathbf{K}_{N,1}(v)\label{eq:a1-1}
\end{equation}
\begin{equation}
\left[\mathbf{K}_{N,1}(u),\mathbf{K}_{a,b}(v)\right]=\frac{1}{u-v}\left(\mathbf{K}_{a,1}(u)\mathbf{K}_{N,b}(v)-\mathbf{K}_{a,1}(v)\mathbf{K}_{N,b}(u)\right)\label{eq:ab-1}
\end{equation}
Let us apply (\ref{eq:a1-1}) to obtain
\begin{align}
\mathbf{K}_{N,1}(u)\left(\mathbf{K}_{a,1}(v)\mathbf{K}_{N,1}^{-1}(v)\mathbf{K}_{N,b}(v)\right) & =\frac{u-v-1}{u-v}\mathbf{K}_{a,1}(v)\mathbf{K}_{N,1}(u)\mathbf{K}_{N,1}^{-1}(v)\mathbf{K}_{N,b}(v)+\nonumber \\
 & +\frac{1}{u-v}\mathbf{K}_{a,1}(u)\mathbf{K}_{N,b}(v).
\end{align}
On the first term on the right-hand side, we use the commutation relations
(\ref{eq:comKK-1}) and (\ref{eq:Nb-1}):
\begin{multline}
\frac{u-v-1}{u-v}\mathbf{K}_{a,1}(v)\mathbf{K}_{N,1}(u)\mathbf{K}_{N,1}^{-1}(v)\mathbf{K}_{1,b}(v)=\mathbf{K}_{a,1}(v)\mathbf{K}_{N,1}^{-1}(v)\left(\frac{u-v-1}{u-v}\mathbf{K}_{N,1}(u)\mathbf{K}_{N,b}(v)\right)\\
\left(\mathbf{K}_{a,1}(v)\mathbf{K}_{N,1}^{-1}(v)\mathbf{K}_{N,b}(v)\right)\mathbf{K}_{N,1}(u)-\frac{1}{u-v}\mathbf{K}_{a,1}(v)\mathbf{K}_{N,b}(u).
\end{multline}
Combining the last two equations we have
\begin{equation}
\left[\mathbf{K}_{N,1}(u),\mathbf{K}_{a,1}(v)\mathbf{K}_{N,1}^{-1}(v)\mathbf{K}_{N,b}(v)\right]=\frac{1}{u-v}\left(\mathbf{K}_{a,1}(u)\mathbf{K}_{N,b}(v)-\mathbf{K}_{a,1}(v)\mathbf{K}_{N,b}(u)\right).
\end{equation}
We can see the r.h.s. agrees with the r.h.s. of the commutation relation
(\ref{eq:ab-1}) therefore we just proved
\begin{equation}
\left[\mathbf{K}_{N,1}(u),\mathbf{K}_{a,b}^{(2)}(v)\right]=0,
\end{equation}
for $a,b=2,\dots,N-1$. 
\end{proof}
Now we turn to the proof of Theorem \ref{thm:uncrossed-nestedK}.
\begin{proof}
First, let us prove that the nested $K$-matrix $\mathbf{K}^{(2)}(u)$
satisfies the reflection equation. The proof is very similar to the
proof of Theorem 3.1 in \cite{Jing_2018}. We start from the reflection
equation
\begin{align}
\mathbf{k}_{1,2}(u)=R_{1,2}(-1)\mathbf{K}_{1}(u+1)R_{1,2}(-2u-1)\mathbf{K}_{2}(u) & =\mathbf{K}_{2}(u)R_{1,2}(-2u-1)\mathbf{K}_{1}(u+1)R_{1,2}(-1)\nonumber \\
 & =\sum_{a_{i},b_{i}}e_{a_{1},b_{1}}\otimes e_{a_{2},b_{2}}\otimes\mathbf{k}_{a_{1},a_{2}}^{b_{1},b_{2}}(u).\label{eq:qminor}
\end{align}
Since $R_{1,2}(-1)$ is a projection operator onto the antisymmetrized
subspace, the quantum minor $\mathbf{k}_{a_{1},a_{2}}^{b_{1},b_{2}}(u)$
is antisymmetric in both its upper and lower indices. Let us examine
the specific components $\mathbf{k}_{N,a}^{1,b}(u)$:
\begin{equation}
\mathbf{k}_{N,a}^{1,b}(u)=\mathbf{K}_{N,1}(u+1)\mathbf{K}_{a,b}(u)-\mathbf{K}_{a,1}(u+1)\mathbf{K}_{N,b}(u).
\end{equation}
Substitute $(u,v)\to(u+1,u)$ into (\ref{eq:a1-1}):
\begin{equation}
\mathbf{K}_{N,1}(u+1)\mathbf{K}_{a,1}(u)=\mathbf{K}_{a,1}(u+1)\mathbf{K}_{N,1}(u).
\end{equation}
Using this, we find that $\mathbf{k}_{N,a}^{1,b}$ can be expressed
in terms of $\mathbf{K}_{a,b}^{(2)}$:
\begin{equation}
\mathbf{k}_{N,a}^{1,b}(u)=\mathbf{K}_{N,1}(u+1)\mathbf{K}_{a,b}^{(2)}(u).
\end{equation}
We define an $(N-2)\times(N-2)$ matrix as follows:
\begin{equation}
\tilde{\mathbf{k}}_{1}(u)=\sum_{a_{1},b_{1}=2}^{N-1}\tilde{e}_{a_{1},b_{1}}\otimes\mathbf{k}_{N,a_{1}}^{1,b_{1}}(u),
\end{equation}
where $\tilde{e}_{a,b}$ are the unit matrices of size $(N-2)\times(N-2)$.
In this notation,
\begin{equation}
\tilde{\mathbf{k}}_{1}(u)=\mathbf{G}(u+1)\mathbf{K}_{1}^{(2)}(u).\label{eq:kK}
\end{equation}
Now consider four auxiliary spaces and define an $R$-matrix on them
\begin{align}
R_{1,2,3,4}(u) & =A_{1,2}A_{3,4}R_{4,1}(u-1)R_{3,1}(u)R_{4,2}(u)R_{3,2}(u+1)\nonumber \\
 & =R_{3,2}(u+1)R_{3,1}(u)R_{4,2}(u)R_{4,1}(u-1)A_{1,2}A_{3,4}\\
 & =\sum_{a_{i},b_{i}=1}^{N}e_{a_{1},b_{1}}\otimes e_{a_{2},b_{2}}\otimes e_{a_{3},b_{3}}\otimes e_{a_{4},b_{4}}\otimes R_{a_{1},a_{2},a_{3},a_{4}}^{b_{1},b_{2},b_{3},b_{4}}(u).\nonumber 
\end{align}
This $R$-matrix is antisymmetric in the first-second and third-fourth
indices, i.e.,
\begin{equation}
P_{1,2}R_{1,2,3,4}(u)=R_{1,2,3,4}(u)P_{1,2}=P_{3,4}R_{1,2,3,4}(u)=R_{1,2,3,4}(u)P_{3,4}=-R_{1,2,3,4}(u).
\end{equation}
The quantum minor satisfies the following equation:
\begin{equation}
R_{1,2,3,4}(v-u)\mathbf{k}_{1,2}(u)R_{1,2,3,4}(-u-v-1)\mathbf{k}_{3,4}(v)=\mathbf{k}_{3,4}(v)R_{1,2,4,3}(-u-v-1)\mathbf{k}_{1,2}(u)R_{1,2,3,4}(v-u).\label{eq:refl1234}
\end{equation}
We define operators where the components are fixed in the first and
third spaces
\begin{equation}
\begin{split}R_{1,2}^{\circ}(u) & =\sum_{a_{i},b_{i}=2}^{N-1}\tilde{e}_{a_{2},b_{2}}\otimes\tilde{e}_{a_{4},b_{4}}\otimes R_{1,a_{2},1,a_{4}}^{1,b_{2},1,b_{4}}(u),\\
R_{1,2}^{\bullet}(u) & =\sum_{a_{i},b_{i}=2}^{N-1}\tilde{e}_{a_{2},b_{2}}\otimes\tilde{e}_{a_{4},b_{4}}\otimes R_{1,a_{2},N,a_{4}}^{1,b_{2},N,b_{4}}(u),
\end{split}
\end{equation}
again using $\tilde{e}_{a,b}$ a $(N-2)\times(N-2)$ as the unit matrices.
Taking the $(N,1)$ components of equation (\ref{eq:refl1234}) in
the first and third spaces, we obtain a reflection equation for the
$\tilde{\mathbf{k}}$ matrices
\begin{equation}
R_{1,2}^{\circ}(v-u)\tilde{\mathbf{k}}_{1}(u)R_{1,2}^{\bullet}(-u-v-1)\tilde{\mathbf{k}}_{2}(v)=\tilde{\mathbf{k}}_{2}(v)R_{1,2}^{\bullet}(-u-v-1)\tilde{\mathbf{k}}_{1}(u)R_{1,2}^{\circ}(v-u),\label{eq:Rcirc}
\end{equation}
where we used the identities $R_{1,a_{2},1,a_{4}}^{1,b_{2},1,b_{4}}(u)=R_{N,a_{2},N,a_{4}}^{N.b_{2},N,b_{4}}(u)$
and $R_{1,a_{2},N,a_{4}}^{1,b_{2},N,b_{4}}(u)=R_{N,a_{2},1,a_{4}}^{N.b_{2},1,b_{4}}(u)$
and the fact that $R_{a_{1},a_{2},a_{3},a_{4}}^{b_{1},b_{2},b_{3},b_{4}}$
is antisymmetric in the first-second and third-fourth indices. The
matrices $R^{\circ}$ and $R^{\bullet}$ can be expressed using the
$\mathfrak{gl}_{N-2}$ symmetric $R$-matrix as follows:
\begin{equation}
\begin{split}R_{1,2}^{\circ}(u) & =\frac{u+2}{u+1}R_{1,2}^{(N-2)}(u),\\
R_{1,2}^{\bullet}(u) & =R_{1,2}^{(N-2)}(u+1).
\end{split}
\end{equation}
Substituting back into equation (\ref{eq:Rcirc}), we get:
\begin{equation}
R_{1,2}^{(N-2)}(v-u)\mathbf{K}_{1}^{(2)}(u)R_{1,2}^{(N-2)}(-u-v)\mathbf{K}_{2}^{(2)}(v)=\mathbf{K}_{2}^{(2)}(v)R_{1,2}^{(N-2)}(-u-v)\mathbf{K}_{1}^{(2)}(u)R_{1,2}^{(N-2)}(v-u),
\end{equation}
where we used the relation (\ref{eq:kK}) and the commutation relation
(\ref{eq:comGK}). We see that the nested $K$-matrix $\mathbf{K}^{(2)}$
satisfies the $\mathfrak{gl}_{N-2}$ reflection equation. By repeatedly
applying the above method, we see that $\mathbf{K}^{(k)}$ satisfies
the $\mathfrak{gl}_{N-2k+2}$ reflection equation. In \ref{subsec:The-uncrossed-overlaps}
, we saw that the asymptotic limit of the nested $K$-matrices satisfies
equation (\ref{eq:Urec}). Combining this with the reflection equation
just derived, we see that $\mathbf{K}^{(k)}$ is a representation
of the reflection algebra $\mathcal{B}(N-2k+2,M-k+1)$.

Since the operator $\mathbf{K}^{(k)}$ satisfies the $\mathfrak{gl}_{N-2k+2}$
reflection equation, the Lemmas \ref{lem:KK-1} and \ref{lem:KG-1}
are also applicable to the operator $\mathbf{K}^{(k)}$, i.e.,
\begin{equation}
\left[\mathbf{G}^{(k)}(v),\mathbf{G}^{(k)}(u)\right]=\left[\mathbf{G}^{(k)}(v),\mathbf{K}_{a,b}^{(k+1)}(u)\right]=0,
\end{equation}
where $a,b=k+1,\dots,N-k$.
\end{proof}

\section{Proofs for the theorems of overlaps\label{sec:Proofs-for-overlaps}}

In this section, we prove the theorems related to the overlap formula
from sections \ref{subsec:Sum-formulas} and \ref{subsec:On-shell-limit}.

\subsection{Proof for the sum formula\label{subsec:Proof-sum-formula}}
\begin{lem}
\label{lem:sum}The off-shell overlap has the sum formula
\begin{equation}
\mathbf{S}_{\bar{\alpha},\mathbf{B}}(\bar{t})=\sum_{\mathrm{part}(\bar{t})}\mathbf{W}_{\mathbf{B}}(\bar{t}_{\textsc{i}}|\bar{t}_{\textsc{ii}})\prod_{\nu=1}^{N-1}\alpha_{\nu}(\bar{t}_{\textsc{i}}^{\nu}),
\end{equation}
where the sum goes through the partitions $\bar{t}^{\nu}=\bar{t}_{\textsc{i}}^{\nu}\cup\bar{t}_{\textsc{ii}}^{\nu}$,
and \textup{$\mathbf{W}_{\mathbf{B}}$ is a matrix valued function
which depends only on }the variables $\bar{t}_{\textsc{i}},\bar{t}_{\textsc{ii}}$,
components of the $K$-matrices $\mathbf{K}_{i,j}$ and the vacuum
overlap $\mathbf{B}$.
\end{lem}
\begin{proof}
The proof is the same as Appendix C in \cite{Gombor:2021hmj}. In
that earlier case, the quantities $\mathbf{K}_{i,j}$, $\mathbf{W}$
and $\mathbf{S}$ were still scalars, but in the current paper, they
are matrices. However, this does not require any changes in the derivation
of the sum formula. Appendix C of \cite{Gombor:2021hmj} applies directly
even in the presence of nontrivial boundary spaces.
\end{proof}
\begin{lem}
\label{lem:weigths}The weights are factorized as 
\begin{equation}
\mathbf{W}_{\mathbf{B}}(\bar{t}_{i}|\bar{t}_{ii})=\frac{\prod_{\nu=1}^{N-1}f(\bar{t}_{\textsc{ii}}^{\nu},\bar{t}_{\textsc{i}}^{\nu})}{\prod_{\nu=1}^{N-2}f(\bar{t}_{\textsc{ii}}^{\nu+1},\bar{t}_{\textsc{i}}^{\nu})}\bar{\mathbf{Z}}(\bar{t}_{\textsc{ii}})\mathbf{B}\mathbf{Z}(\bar{t}_{\textsc{i}}),
\end{equation}
where the highest coefficients (HC) $\mathbf{Z}(\bar{t})$ and $\bar{\mathbf{Z}}(\bar{t})$
depend only on the variables $\bar{t}$ and the entries of the $K$-matrix.
\end{lem}
\begin{proof}
The proof is similar than Appendix D in \cite{Gombor:2021hmj}. The
derivation is based on co-product formula of the Bethe states (\ref{eq:Bcoprod}).
Let us renormalize the Bethe states and overlaps as
\begin{equation}
\tilde{\mathbb{B}}(\bar{t})=\prod_{\nu=1}^{N-1}\lambda_{\nu+1}(\bar{t}^{\nu})\mathbb{B}(\bar{t}),\quad\tilde{\mathbf{S}}_{\bar{\lambda},\mathbf{B}}(\bar{t})=\prod_{\nu=1}^{N-1}\lambda_{\nu+1}(\bar{t}^{\nu})\mathbf{S}_{\bar{\alpha},\mathbf{B}}(\bar{t}).
\end{equation}
The sum formula (Lemma \ref{lem:sum}) of the renormalized overlap
is 
\begin{equation}
\tilde{\mathbf{S}}_{\bar{\lambda},\mathbf{B}}(\bar{t})=\sum_{\mathrm{part}(\bar{t})}\mathbf{W}_{\mathbf{B}}(\bar{t}_{\textsc{i}}|\bar{t}_{\textsc{ii}})\prod_{\nu=1}^{N-1}\lambda_{\nu}(\bar{t}_{\textsc{i}}^{\nu})\lambda_{\nu+1}(\bar{t}_{\textsc{ii}}^{\nu}).\label{eq:rensum}
\end{equation}
Let $\mathcal{H}^{(1)},\mathcal{H}^{(2)}$ be two subsystems of the
quantum space for which $\mathcal{H}=\mathcal{H}^{(1)}\otimes\mathcal{H}^{(2)}$
and the renormalized co-product formula is
\begin{equation}
\tilde{\mathbb{B}}(\bar{t})=\sum\frac{\prod_{\nu=1}^{N-1}\lambda_{\nu}^{(2)}(\bar{t}_{\textsc{i}}^{\nu})\lambda_{\nu+1}^{(1)}(\bar{t}_{\textsc{ii}}^{\nu})f(\bar{t}_{\textsc{ii}}^{\nu},\bar{t}_{\textsc{i}}^{\nu})}{\prod_{\nu=1}^{N-2}f(\bar{t}_{\textsc{ii}}^{\nu+1},\bar{t}_{\textsc{i}}^{\nu})}\tilde{\mathbb{B}}^{(1)}(\bar{t}_{\textsc{i}})\tilde{\mathbb{B}}^{(2)}(\bar{t}_{\textsc{ii}}),
\end{equation}
where $\tilde{\mathbb{B}}^{(1/2)}$ and $\lambda_{\nu}^{(1/2)}$ are
the off-shell Bethe vectors and pseudo-vacuum eigenvalues on the subsystems
$\mathcal{H}^{(1/2)}$. We can also use the co-product property of
the boundary state, Lemma \ref{lem:co-prod}. Combining the co-product
formulas of the boundary state and the Bethe states, the of-shell
overlaps can be written as
\begin{equation}
\tilde{\mathbf{S}}_{\bar{\lambda},\mathbf{B}}(\bar{t})=\sum\frac{\prod_{\nu=1}^{N-1}\lambda_{\nu}^{(2)}(\bar{t}_{\textsc{i}}^{\nu})\lambda_{\nu+1}^{(1)}(\bar{t}_{\textsc{ii}}^{\nu})f(\bar{t}_{\textsc{ii}}^{\nu},\bar{t}_{\textsc{i}}^{\nu})}{\prod_{\nu=1}^{N-2}f(\bar{t}_{\textsc{ii}}^{\nu+1},\bar{t}_{\textsc{i}}^{\nu})}\tilde{\mathbf{S}}_{\bar{\lambda}^{(2)},\mathbf{B}^{(2)}}(\bar{t}_{\textsc{ii}})\tilde{\mathbf{S}}_{\bar{\lambda}^{(1)},\mathbf{B}^{(1)}}(\bar{t}_{\textsc{i}}),\label{eq:renco}
\end{equation}
where $\mathbf{B}^{(1/2)}$ are the pseudo-vacuum overlaps. 

Now let us fix a particular partition $\bar{t}=\bar{t}_{\textsc{i}}\cup\bar{t}_{\textsc{ii}}$
and choose the highest weights as
\begin{equation}
\begin{split}\lambda_{\nu+1}^{(1)}(u) & =0,\qquad\text{for all }u\in\bar{t}_{\textsc{i}}^{\nu},\\
\lambda_{\nu}^{(2)}(u) & =0,\qquad\text{for all }u\in\bar{t}_{\textsc{ii}}^{\nu},
\end{split}
\end{equation}
therefore in the sum rule of the overlap (\ref{eq:rensum}) and the
co-product formula (\ref{eq:renco}) there is only one non-vanishing
term:
\begin{equation}
\begin{split}\tilde{\mathbf{S}}_{\bar{\lambda},\mathbf{B}}(\bar{t}) & =\mathbf{W}_{\mathbf{B}}(\bar{t}_{\textsc{i}}|\bar{t}_{\textsc{ii}})\prod_{\nu=1}^{N-1}\lambda_{\nu}(\bar{t}_{\textsc{i}}^{\nu})\lambda_{\nu+1}(\bar{t}_{\textsc{ii}}^{\nu}),\\
\tilde{\mathbf{S}}_{\bar{\lambda}^{(1)},\mathbf{B}^{(1)}}(\bar{t}_{\textsc{i}}) & =\mathbf{W}_{\mathbf{B}^{(1)}}(\bar{t}_{\textsc{i}}|\emptyset)\prod_{\nu=1}^{N-1}\lambda_{\nu}^{(1)}(\bar{t}_{\textsc{i}}^{\nu}),\\
\tilde{\mathbf{S}}_{\bar{\lambda}^{(2)},\mathbf{B}^{(2)}}(\bar{t}_{\textsc{ii}}) & =\mathbf{W}_{\mathbf{B}^{(2)}}(\emptyset|\bar{t}_{\textsc{ii}})\prod_{\nu=1}^{N-1}\lambda_{\nu+1}^{(2)}(\bar{t}_{\textsc{ii}}^{\nu}).
\end{split}
\end{equation}
Substituting back, we obtain that
\begin{equation}
\mathbf{W}_{\mathbf{B}}(\bar{t}_{\textsc{i}}|\bar{t}_{\textsc{ii}})=\frac{\prod_{\nu=1}^{N-1}f(\bar{t}_{\textsc{ii}}^{\nu},\bar{t}_{\textsc{i}}^{\nu})}{\prod_{\nu=1}^{N-2}f(\bar{t}_{\textsc{ii}}^{\nu+1},\bar{t}_{\textsc{i}}^{\nu})}\mathbf{W}_{\mathbf{B}^{(2)}}(\emptyset|\bar{t}_{\textsc{ii}})\mathbf{W}_{\mathbf{B}^{(1)}}(\bar{t}_{\textsc{i}}|\emptyset).
\end{equation}
Applying the last equation for $\bar{t}_{\textsc{ii}}=\emptyset$
we have
\begin{equation}
\mathbf{W}_{\mathbf{B}}(\bar{t}|\emptyset)=\mathbf{W}_{\mathbf{B}^{(2)}}(\emptyset|\emptyset)\mathbf{W}_{\mathbf{B}^{(1)}}(\bar{t}|\emptyset).
\end{equation}
Since $\mathbf{W}_{\mathbf{B}^{(2)}}(\emptyset|\emptyset)$ is just
the vacuum overlap we have
\begin{equation}
\mathbf{W}_{\mathbf{B}}(\bar{t}|\emptyset)=\mathbf{B}^{(2)}\mathbf{W}_{\mathbf{B}^{(1)}}(\bar{t}|\emptyset).\label{eq:Wdec}
\end{equation}
From the co-product property of the boundary state (Lemma \ref{lem:co-prod})
we know that $\mathbf{B}=\mathbf{B}^{(2)}\mathbf{B}^{(1)}.$ Since
the equation (\ref{eq:Wdec}) is true for any decomposition of the
quantum space, the $\mathbf{B}$-dependence of the weights $\mathbf{W}_{\mathbf{B}}(\bar{t}|\emptyset)$
should be 
\begin{equation}
\mathbf{W}_{\mathbf{B}}(\bar{t}|\emptyset)=\mathbf{B}\mathbf{Z}(\bar{t}).
\end{equation}
Analogous way, for $\bar{t}_{i}=\emptyset$ we have
\begin{equation}
\mathbf{W}_{\mathbf{B}}(\emptyset|\bar{t})=\mathbf{W}_{\mathbf{B}^{(2)}}(\emptyset|\bar{t})\mathbf{B}^{(1)},
\end{equation}
therefore we can introduce another HC as
\begin{equation}
\mathbf{W}_{\mathbf{B}}(\emptyset|\bar{t})=\bar{\mathbf{Z}}(\bar{t})\mathbf{B}.
\end{equation}
Substituting back we obtain that
\begin{equation}
\mathbf{W}_{\mathbf{B}}(\bar{t}_{\textsc{i}}|\bar{t}_{\textsc{ii}})=\frac{\prod_{\nu=1}^{N-1}f(\bar{t}_{\textsc{ii}}^{\nu},\bar{t}_{\textsc{i}}^{\nu})}{\prod_{\nu=1}^{N-2}f(\bar{t}_{\textsc{ii}}^{\nu+1},\bar{t}_{\textsc{i}}^{\nu})}\bar{\mathbf{Z}}(\bar{t}_{\textsc{ii}})\mathbf{B}\mathbf{Z}(\bar{t}_{\textsc{i}}).
\end{equation}
\end{proof}
\begin{lem}
\label{lem:BtB}Let us introduce the following modified monodromy
matrix
\begin{equation}
\tilde{T}(u)=VT^{t}(-u)V,
\end{equation}
which also satisfies the $RTT$-relation. Denote the Bethe vector
defined with the components of $\tilde{T}$ as $\tilde{\mathbb{B}}(\bar{t})$.
The new Bethe vectors can be expressed in terms of the original Bethe
vectors.
\begin{equation}
\frac{1}{\prod_{\nu=1}^{N-1}\tilde{\alpha}_{\nu}(\bar{t}^{\nu})}\tilde{\mathbb{B}}(\bar{t})=\mathbb{B}(\pi^{a}(\bar{t})).
\end{equation}
\end{lem}
\begin{proof}
We can prove the statement by induction. The pseudo-vacuum is also
the highest weight state for $\tilde{T}$, since
\begin{equation}
\tilde{T}_{i,j}(u)|0\rangle=T_{N+1-j,N+1-i}(-u)|0\rangle=0,
\end{equation}
for $i>j$ and 
\begin{equation}
\tilde{T}_{i,i}(u)|0\rangle=T_{N+1-i,N+1-i}(-u)|0\rangle=\lambda_{N+1-i}(-u)|0\rangle,
\end{equation}
therefore
\begin{equation}
\tilde{\lambda}_{i}(u)=\lambda_{N+1-i}(-u).\label{eq:lamh}
\end{equation}
This proves the statement for $r_{j}=0$ for $j=1,\dots,N-1$.

Now assume the statement holds when the number of Bethe roots of the
first type is at most $r_{1}$. Now use the recursive equation for
$r_{1}+1$:
\begin{multline}
\frac{1}{\tilde{\alpha}_{1}(z)\prod_{\nu=1}^{N-1}\tilde{\alpha}_{\nu}(\bar{t}^{\nu})}\tilde{\mathbb{B}}(\{z,\bar{t}^{1}\},\left\{ \bar{t}^{k}\right\} _{k=2}^{N-1})=\\
\frac{1}{\tilde{\alpha}_{1}(z)\prod_{\nu=1}^{N-1}\tilde{\alpha}_{\nu}(\bar{t}^{\nu})}\sum_{j=2}^{N}\frac{\tilde{T}_{1,j}(z)}{\tilde{\lambda}_{2}(z)}\sum_{\mathrm{part}(\bar{t})}\tilde{\mathbb{B}}(\bar{t}^{1},\left\{ \bar{t}_{\textsc{ii}}^{k}\right\} _{k=2}^{j-1},\left\{ \bar{t}^{k}\right\} _{k=j}^{N-1})\frac{\prod_{\nu=2}^{j-1}\tilde{\alpha}_{\nu}(\bar{t}_{\textsc{i}}^{\nu})g(\bar{t}_{\textsc{i}}^{\nu},\bar{t}_{\textsc{i}}^{\nu-1})f(\bar{t}_{\textsc{ii}}^{\nu},\bar{t}_{\textsc{i}}^{\nu})}{\prod_{\nu=1}^{j-1}f(\bar{t}^{\nu+1},\bar{t}_{\textsc{i}}^{\nu})}.
\end{multline}
Apply the induction hypothesis to the right-hand side.
\begin{multline}
\frac{1}{\tilde{\alpha}_{1}(z)\prod_{\nu=1}^{N-1}\tilde{\alpha}_{\nu}(\bar{t}^{\nu})}\tilde{\mathbb{B}}(\{z,\bar{t}^{1}\},\left\{ \bar{t}^{k}\right\} _{k=2}^{N-1})=\\
\sum_{j=2}^{N}\frac{\tilde{T}_{1,j}(z)}{\tilde{\lambda}_{1}(z)}\sum_{\mathrm{part}(\bar{t})}\mathbb{B}(\left\{ -\bar{t}^{N-k}\right\} _{k=1}^{N-j},\left\{ -\bar{t}_{\textsc{ii}}^{N-k}\right\} _{k=N-j+1}^{N-2},-\bar{t}^{1})\frac{\prod_{\nu=2}^{j-1}g(\bar{t}_{\textsc{i}}^{\nu},\bar{t}_{\textsc{i}}^{\nu-1})f(\bar{t}_{\textsc{ii}}^{\nu},\bar{t}_{\textsc{i}}^{\nu})}{\prod_{\nu=1}^{j-1}f(\bar{t}^{\nu+1},\bar{t}_{\textsc{i}}^{\nu})}.
\end{multline}
Now we can express $\tilde{T}_{1,j}$ and $\tilde{\lambda}_{1}$ in
terms of the original quantities.
\begin{multline}
\frac{1}{\tilde{\alpha}_{1}(z)\prod_{\nu=1}^{N-1}\tilde{\alpha}_{\nu}(\bar{t}^{\nu})}\tilde{\mathbb{B}}(\{z,\bar{t}^{1}\},\left\{ \bar{t}^{k}\right\} _{k=2}^{N-1})=\\
\sum_{j=1}^{N-1}\frac{T_{j,N}(-z)}{\lambda_{N}(-z)}\sum_{\mathrm{part}(\bar{t})}\mathbb{B}(\left\{ -\bar{t}^{N-k}\right\} _{k=1}^{j-1},\left\{ -\bar{t}_{\textsc{ii}}^{N-k}\right\} _{k=j}^{N-2},-\bar{t}^{1})\frac{\prod_{\nu=2}^{N-j}g(\bar{t}_{\textsc{i}}^{\nu},\bar{t}_{\textsc{i}}^{\nu-1})f(\bar{t}_{\textsc{ii}}^{\nu},\bar{t}_{\textsc{i}}^{\nu})}{\prod_{\nu=1}^{N-j}f(\bar{t}^{\nu+1},\bar{t}_{\textsc{i}}^{\nu})}.
\end{multline}
After rearrangements, we get
\begin{multline}
\frac{1}{\tilde{\alpha}_{1}(z)\prod_{\nu=1}^{N-1}\tilde{\alpha}_{\nu}(\bar{t}^{\nu})}\tilde{\mathbb{B}}(\{z,\bar{t}^{1}\},\left\{ \bar{t}^{k}\right\} _{k=2}^{N-1})=\\
\sum_{j=1}^{N-1}\frac{T_{j,N}(-z)}{\lambda_{N}(-z)}\sum_{\mathrm{part}(\bar{t})}\mathbb{B}(\left\{ -\bar{t}^{N-k}\right\} _{k=1}^{j-1},\left\{ -\bar{t}_{\textsc{ii}}^{N-k}\right\} _{k=j}^{N-2},-\bar{t}^{1})\frac{\prod_{\nu=j}^{N-2}g(-\bar{t}_{\textsc{i}}^{N-\nu-1},-\bar{t}_{\textsc{i}}^{N-\nu})f(-\bar{t}_{\textsc{i}}^{N-\nu},-\bar{t}_{\textsc{ii}}^{N-\nu})}{\prod_{\nu=1}^{N-j}f(-\bar{t}_{\textsc{i}}^{N-\nu},-\bar{t}^{N-\nu+1})}.
\end{multline}
We see that the right-hand side matches the right-hand side of the
recursive equation (\ref{eq:rec2}), i.e.,
\begin{equation}
\frac{1}{\tilde{\alpha}_{1}(z)\prod_{\nu=1}^{N-1}\tilde{\alpha}_{\nu}(\bar{t}^{\nu})}\tilde{\mathbb{B}}(\{z,\bar{t}^{1}\},\left\{ \bar{t}^{k}\right\} _{k=2}^{N-1})=\mathbb{B}(\left\{ -\bar{t}^{N-k}\right\} _{k=1}^{N-2},\{-z,-\bar{t}^{1}\}),
\end{equation}
which proves the inductive step.
\end{proof}
\begin{lem}
Let us introduce the following modified monodromy matrix
\begin{equation}
\tilde{T}(u)=\widehat{T}^{t}(-u),
\end{equation}
which also satisfies the $RTT$-relation. Denote the Bethe vector
defined with the components of $\tilde{T}$ as $\tilde{\mathbb{B}}(\bar{t})$.
The new Bethe vectors can be expressed in terms of the original Bethe
vectors
\begin{equation}
\frac{1}{\prod_{\nu=1}^{N-1}\tilde{\alpha}_{\nu}(\bar{t}^{\nu})}\tilde{\mathbb{B}}(\bar{t})=A(\bar{t})\mathbb{B}(\pi^{c}(\bar{t})),
\end{equation}
where
\begin{equation}
A(\bar{t})=(-1)^{\#\bar{t}}\left(\prod_{s=1}^{N-2}f(\bar{t}^{s+1},\bar{t}^{s})\right)^{-1}.\label{eq:At}
\end{equation}
\end{lem}
\begin{proof}
The matrix $\tilde{T}$ can be expressed in two steps
\begin{equation}
\tilde{T}(u)=V\check{T}^{t}(-u)V.
\end{equation}
This is the relation that appeared in Lemma \ref{lem:BtB}, i.e.,
$\tilde{\mathbb{B}}$ can be expressed in terms of $\check{\mathbb{B}}$
as follows
\begin{equation}
\frac{1}{\prod_{\nu=1}^{N-1}\tilde{\alpha}_{\nu}(\bar{t}^{\nu})}\tilde{\mathbb{B}}(\bar{t})=\check{\mathbb{B}}(\pi^{a}(\bar{t})).
\end{equation}
At the same time, $\check{\mathbb{B}}$ can be expressed in terms
of the original Bethe vector as shown in \cite{Liashyk:2018egk}
\begin{equation}
\hat{\mathbb{B}}(\bar{t})=A(\bar{t})\mathbb{B}(\mu(\bar{t})),
\end{equation}
where
\begin{equation}
\mu(\bar{t})=\{\bar{t}^{N-1}-1,\bar{t}^{N-2}-2,\dots,\bar{t}^{1}-(N-1)\}.
\end{equation}
Substituting back, we get
\begin{equation}
\frac{1}{\prod_{\nu=1}^{N-1}\tilde{\alpha}_{\nu}(\bar{t}^{\nu})}\tilde{\mathbb{B}}(\bar{t})=A(\pi^{a}(\bar{t}))\mathbb{B}(\mu(\pi^{a}(\bar{t}))).
\end{equation}
It is easy to show that $A(\pi^{a}(\bar{t}))=A(\bar{t})$ and $\mu(\pi^{a}(\bar{t}))=\pi^{c}(\bar{t})$,
which completes the proof.
\end{proof}
\begin{lem}
The uncrossed HC-s have the following property\textup{
\begin{equation}
\mathbf{Z}^{\mathbf{K}}(\bar{t})=\left[\bar{\mathbf{Z}}^{\mathbf{K}^{\Pi}}(\pi^{a}(\bar{t}))\right]^{t_{B}},\label{eq:uncrZZbar}
\end{equation}
where we have indicated the $K$-matrix dependence of the HCs and
\begin{equation}
\mathbf{K}^{\Pi}(u)=V\mathbf{K}^{T}(u)V,
\end{equation}
where $^{T}$ denotes the transposition on the auxiliary and boundary
spaces.}
\end{lem}
\begin{proof}
By rearranging the $KT$-relation, we obtain:
\begin{equation}
\mathbf{K}^{\Pi}(u)\langle\Psi^{t_{B}}|\left[VT^{t}(-u)V\right]=\langle\Psi^{t_{B}}|\left[VT^{t_{0}}(u)V\right]\mathbf{K}^{\Pi}(u),
\end{equation}
where $^{t_{B}}$ denotes the transposition on the boundary space.
Let us introduce the following modified monodromy matrix:
\begin{equation}
\tilde{T}(u)=VT^{t}(-u)V.
\end{equation}
Denote the Bethe vectors defined by the components of $\tilde{T}$
as $\tilde{\mathbb{B}}(\bar{t})$. These Bethe vectors can be expressed
in terms of the original Bethe vectors, according to Lemma \ref{lem:BtB}:
\begin{equation}
\tilde{\mathbb{B}}(\bar{t})=\mathbb{B}(\pi^{a}(\bar{t}))\prod_{\nu=1}^{N-1}\alpha_{\nu}(\bar{t}^{\nu}).
\end{equation}
Since the modified monodromy matrix satisfies the $KT$-relation
\begin{equation}
\mathbf{K}^{\Pi}(u)\langle\Psi^{t_{B}}|\tilde{T}(u)=\langle\Psi^{t_{B}}|\tilde{T}(-u)\mathbf{K}^{\Pi}(u),
\end{equation}
the off-shell overlap satisfies the sum formula
\begin{equation}
\langle\Psi^{t_{B}}|\tilde{\mathbb{B}}(\bar{t})=\sum_{\mathrm{part}(\bar{t})}\mathbf{W}_{\mathbf{B}^{t_{B}}}^{\mathbf{K}^{\Pi}}(\bar{t}_{\textsc{i}}|\bar{t}_{\textsc{ii}})\prod_{\nu=1}^{N-1}\alpha_{\nu}(\bar{t}_{\textsc{i}}^{\nu}),
\end{equation}
where we used (\ref{eq:symProp}) and the symmetry property (\ref{eq:symProp}).
Now we can use the relation between Bethe overlaps, from which we
obtain:
\begin{equation}
\langle\Psi^{t_{B}}|\tilde{\mathbb{B}}(\bar{t})=\left(\langle\Psi|\mathbb{B}(\pi^{a}(\bar{t}))\right)^{t_{B}}\prod_{\nu=1}^{N-1}\alpha_{\nu}(\bar{t}^{\nu})=\sum_{\mathrm{part}(\bar{t})}\left(\mathbf{W}_{\mathbf{B}}^{\mathbf{K}}(\pi^{a}(\bar{t}_{\textsc{i}})|\pi^{a}(\bar{t}_{\textsc{ii}}))\right)^{t_{B}}\prod_{\nu=1}^{N-1}\alpha_{\nu}(\bar{t}_{\textsc{ii}}^{\nu}),
\end{equation}
that is,
\begin{equation}
\left(\mathbf{W}_{\mathbf{B}}^{\mathbf{K}}(\pi^{a}(\bar{t}_{\textsc{i}})|\pi^{a}(\bar{t}_{\textsc{ii}}))\right)^{t_{B}}=\mathbf{W}_{\mathbf{B}^{t_{B}}}^{\mathbf{K}^{\Pi}}(\bar{t}_{\textsc{ii}}|\bar{t}_{\textsc{i}}).
\end{equation}
Using the explicit form of the weights, we obtain the desired formula
(\ref{eq:uncrZZbar}).
\end{proof}
\begin{lem}
The crossed HC-s have the following property\textup{
\begin{equation}
\mathbf{Z}^{\mathbf{K}}(\bar{t})=(-1)^{\#\bar{t}}\left(\prod_{s=1}^{N-2}f(\bar{t}^{s+1},\bar{t}^{s})\right)^{-1}\left[\bar{\mathbf{Z}}^{\mathbf{K}^{\Pi}}(\pi^{c}(\bar{t}))\right]^{t_{B}},\label{eq:crZZbar}
\end{equation}
where we have indicated the $K$-matrix dependence of the HCs and
\begin{equation}
\mathbf{K}^{\Pi}(u)=\mathbf{K}^{T}(u),
\end{equation}
where $^{T}$ denotes the transposition on the auxiliary and boundary
spaces.}
\end{lem}
\begin{proof}
By rearranging the $KT$-relation, we obtain:
\begin{equation}
\mathbf{K}^{\Pi}(u)\langle\Psi^{t_{B}}|\left[\widehat{T}^{t}(-u)\right]=\langle\Psi^{t_{B}}|\left[T^{t}(u)\right]\mathbf{K}^{\Pi}(u).
\end{equation}
Let us introduce the following modified monodromy matrix:
\begin{equation}
\tilde{T}(u)=\widehat{T}^{t}(-u).
\end{equation}
Denote the Bethe vectors defined by the components of $\tilde{T}$
as $\tilde{\mathbb{B}}(\bar{t})$. These Bethe vectors can be expressed
in terms of the original Bethe vectors, according to Lemma \ref{lem:BtB}.
\begin{equation}
\tilde{\mathbb{B}}(\bar{t})=A(\bar{t})\mathbb{B}(\pi^{c}(\bar{t}))\prod_{\nu=1}^{N-1}\alpha_{\nu}(\bar{t}^{\nu}).
\end{equation}
Since the modified monodromy matrix satisfies the $KT$-relation,
the off-shell overlap satisfies the sum formula
\begin{equation}
\langle\Psi^{t_{B}}|\tilde{\mathbb{B}}(\bar{t})=\sum_{\mathrm{part}(\bar{t})}\mathbf{W}_{\mathbf{B}^{t_{B}}}^{\mathbf{K}^{\Pi}}(\bar{t}_{\textsc{i}}|\bar{t}_{\textsc{ii}})\prod_{\nu=1}^{N-1}\alpha_{\nu}(\bar{t}_{\textsc{i}}^{\nu}).
\end{equation}
Now we can use the relation between Bethe overlaps, from which we
obtain:
\begin{equation}
\langle\Psi^{t_{B}}|\tilde{\mathbb{B}}(\bar{t})=A(\bar{t})\left(\langle\Psi|\mathbb{B}(\pi^{c}(\bar{t}))\right)^{t_{B}}\prod_{\nu=1}^{N-1}\alpha_{\nu}(\bar{t}^{\nu})=A(\bar{t})\sum_{\mathrm{part}(\bar{t})}\left(\mathbf{W}_{\mathbf{B}}^{\mathbf{K}}(\pi^{c}(\bar{t}_{\textsc{i}})|\pi^{c}(\bar{t}_{\textsc{ii}}))\right)^{t_{B}}\prod_{\nu=1}^{N-1}\alpha_{\nu}(\bar{t}_{\textsc{ii}}^{\nu}),
\end{equation}
that is,
\begin{equation}
A(\bar{t})\left(\mathbf{W}_{\mathbf{B}}^{\mathbf{K}}(\pi^{c}(\bar{t}_{\textsc{i}})|\pi^{c}(\bar{t}_{\textsc{ii}}))\right)^{t_{B}}=\mathbf{W}_{\mathbf{B}^{t_{B}}}^{\mathbf{K}^{\Pi}}(\bar{t}_{\textsc{ii}}|\bar{t}_{\textsc{i}}).
\end{equation}
Using the explicit form of the weights and the identity $A(\pi^{c}(\bar{t}))^{-1}=A(\bar{t})$,
we obtain the desired formula (\ref{eq:crZZbar}).
\end{proof}
Combining the lemmas from this section, we complete the proof of Theorem
\ref{thm:The-sum-formula}.

\subsection{Poles of the HC-s}

In this subsection, we prove two lemmas that provide the HC $\bar{\mathbf{Z}}(\bar{t})$
recursively. The crossed and uncrossed cases must be treated separately.
In the proof, we follow the calculations in Appendices E and F of
\cite{Gombor:2021hmj}. There, the case $d_{B}=1$ was derived, i.e.,
when the quantities $\bar{\mathbf{Z}}$ and $\mathbf{K}_{i,j}$ are
scalars. Now, these calculations must be generalized to matrix-valued
quantities. The computation is essentially the same; we just need
to pay attention to the order of the matrices.

From the sum formula, we see that the off-shell overlap $\langle\Psi|\mathbb{B}(\bar{t})$
is simply equal to the HC $\bar{\mathbf{Z}}(\bar{t})$ when all $\alpha_{\nu}(t_{k}^{\nu})$
are set to zero. That is, in computing $\bar{\mathbf{Z}}(\bar{t})$,
it is sufficient to consider only the terms independent of $\alpha$.
We introduce the notation $\cong$ to mean that two quantities are
equal in the limit $\alpha_{\nu}(t_{k}^{\nu})\to0$, i.e.,
\begin{equation}
\langle\Psi|\mathbb{B}(\bar{t})\cong\bar{\mathbf{Z}}(\bar{t})\mathbf{B}.
\end{equation}

\begin{lem}
The crossed HC has the following recursion
\begin{multline}
\bar{\mathbf{Z}}(\{z,\bar{t}^{1}\},\left\{ \bar{t}^{\nu}\right\} _{\nu=2}^{N-1})=-\sum_{\mathrm{part}}\mathbf{K}_{1,1}^{-1}(z)\bar{\mathbf{Z}}(\bar{t}_{\textsc{ii}}^{1},\{\bar{t}^{\nu}\}_{\nu=2}^{N-1})\mathbf{G}^{(2)}(z)\frac{f(\bar{t}^{1},-z-1)}{f(\bar{t}^{2},-z-1)f(\bar{t}^{2},z)}\frac{f(\bar{t}_{\textsc{i}}^{1},\bar{t}_{\textsc{ii}}^{1})}{h(\bar{t}_{\textsc{i}}^{1},-z-1)}-\\
-\sum_{i=2}^{N}\mathbf{K}_{1,1}^{-1}(z)\mathbf{K}_{1,i}(z)\sum_{\mathrm{part}}\bar{\mathbf{Z}}(\bar{w}_{\textsc{ii}}^{1},\{\bar{t}_{\textsc{ii}}^{\nu}\}_{\nu=2}^{i-1},\{\bar{t}^{\nu}\}_{\nu=i}^{N-1})\frac{f(\bar{w}_{\textsc{i}}^{1},\bar{w}_{\textsc{ii}}^{1})}{h(\bar{w}_{\textsc{i}}^{1},z)}\frac{f(\bar{t}_{\textsc{i}}^{2},\bar{t}_{\textsc{ii}}^{2})}{h(\bar{t}_{\textsc{i}}^{2},\bar{w}_{\textsc{i}}^{1})f(\bar{t}_{\textsc{i}}^{2},\bar{w}_{\textsc{ii}}^{1})}\prod_{s=3}^{i-1}\frac{f(\bar{t}_{\textsc{i}}^{s},\bar{t}_{\textsc{ii}}^{s})}{h(\bar{t}_{\textsc{i}}^{s},\bar{t}_{\textsc{i}}^{s-1})f(\bar{t}_{\textsc{i}}^{s},\bar{t}_{\textsc{ii}}^{s-1})},\label{eq:rect1}
\end{multline}
In the first line the sum goes over all the partitions of $\bar{t}^{1}\vdash\left\{ \bar{t}_{\textsc{i}}^{1},\bar{t}_{\textsc{ii}}^{1}\right\} $where
$\#\bar{t}_{\textsc{i}}^{1}=1$. In the second line $\bar{w}^{1}=\{z,\bar{t}^{1}\}$
and the sum goes over all the partitions of $\bar{w}^{1}\vdash\left\{ \bar{w}_{\textsc{i}}^{1},\bar{w}_{\textsc{ii}}^{1}\right\} $,
$\bar{t}^{\nu}\vdash\left\{ \bar{t}_{\textsc{i}}^{\nu},\bar{t}_{\textsc{ii}}^{\nu}\right\} $
where $\#\bar{w}_{\textsc{i}}^{1}=\#\bar{t}_{\textsc{i}}^{\nu}=1$
and $\nu=2,\dots,i-1$. 
\end{lem}
\begin{proof}
We can see that this formula matches equation (4.14) from the paper
\cite{Gombor:2021hmj} for scalar quantities. We will follow the proof
described in section E.2 of \cite{Gombor:2021hmj}.

Let us start with the recurrence equation
\begin{equation}
\left\langle \Psi\right|\mathbb{B}(\{z,\bar{t}^{1}\},\left\{ \bar{t}^{\nu}\right\} _{\nu}^{N-1})\cong\left\langle \Psi\right|\frac{T_{1,2}(z)}{\lambda_{2}(z)f(\bar{t}^{2},z)}\mathbb{B}(\bar{t}),\label{eq:recZb}
\end{equation}
and the crossed $KT$-relation 
\begin{equation}
\mathbf{K}_{1,1}(z)\left\langle \Psi\right|T_{1,2}(z)=\sum_{j=1}^{N}\left\langle \Psi\right|\widehat{T}_{1,j}(-z)\mathbf{K}_{j,2}(z)-\sum_{i=2}^{N}\mathbf{K}_{1,i}(z)\left\langle \Psi\right|T_{i,2}(z).\label{eq:KTtw}
\end{equation}
The $T_{1,2}(z)$ term in the r.h.s. of (\ref{eq:recZb}) can be expressed
with $\widehat{T}_{1,k}(-z)$ and $T_{j,2}(z)$ where $k=1,\dots,N$
and $j=2,\dots,N$ therefore
\begin{multline}
\left\langle \Psi\right|\mathbb{B}(\{z,\bar{t}^{1}\},\left\{ \bar{t}^{\nu}\right\} _{\nu=2}^{N-1})\cong\\
\frac{1}{f(\bar{t}^{2},z)}\left[\sum_{j=1}^{N}\mathbf{K}_{1,1}^{-1}(z)\left\langle \Psi\right|\frac{\widehat{T}_{1,k}(-z)}{\lambda_{2}(z)}\mathbb{B}(\bar{t})\mathbf{K}_{k,2}(z)-\sum_{i=2}^{N}\mathbf{K}_{1,1}^{-1}(z)\mathbf{K}_{1,i}(z)\left\langle \Psi\right|\frac{T_{i,2}(z)}{\lambda_{2}(z)}\mathbb{B}(\bar{t})\right].\label{eq:temp3}
\end{multline}
The next step is to apply the action formulas to the expressions
\begin{equation}
\left\langle \Psi\right|\frac{\widehat{T}_{1,k}(-z)}{\lambda_{2}(z)}\mathbb{B}(\bar{t})\quad\text{and}\quad\left\langle \Psi\right|\frac{T_{i,2}(z)}{\lambda_{2}(z)}\mathbb{B}(\bar{t}).
\end{equation}
According to section E.2 of \cite{Gombor:2021hmj} only the terms
with $k=1,2$ contain $\alpha$-independent contributions, and these
are the following:
\begin{equation}
\begin{split}\left\langle \Psi\right|\frac{\widehat{T}_{1,1}(-z)}{\lambda_{2}(z)}\mathbb{B}(\bar{t}) & =\sum_{\mathrm{part}}\bar{\mathbf{Z}}(\bar{t}_{\textsc{ii}}^{1},\{\bar{t}^{\nu}\}_{\nu=2}^{N-1})\mathbf{B}\mathbf{Z}(\{-z-1\},\emptyset^{\times N-2})\frac{f(\bar{t}^{1},-z-1)}{f(\bar{t}^{2},-z-1)}\frac{f(\bar{t}_{\textsc{i}}^{1},\bar{t}_{\textsc{ii}}^{1})}{h(\bar{t}_{\textsc{i}}^{1},-z-1)},\\
\left\langle \Psi\right|\frac{\widehat{T}_{1,2}(-z)}{\lambda_{2}(z)}\mathbb{B}(\bar{t}) & =(-1)\sum_{\mathrm{part}}\bar{\mathbf{Z}}(\bar{t}_{\textsc{ii}}^{1},\{\bar{t}^{\nu}\}_{\nu=2}^{N-1})\mathbf{B}\frac{f(\bar{t}^{1},-z-1)}{f(\bar{t}^{2},-z-1)}\frac{f(\bar{t}_{\textsc{i}}^{1},\bar{t}_{\textsc{ii}}^{1})}{h(\bar{t}_{\textsc{i}}^{1},-z-c)},
\end{split}
\end{equation}
see equations (E.27), (E.29), (E.33), and (E.34) in \cite{Gombor:2021hmj}.
The sums goes over all the partitions of $\bar{t}^{1}\vdash\left\{ \bar{t}_{\textsc{i}}^{1},\bar{t}_{\textsc{ii}}^{1}\right\} $where
$\#\bar{t}_{\textsc{i}}^{1}=1$. The first formula contains the HC
$\mathbf{Z}(\{-z-1\},\emptyset^{\times N-2})$. We can calculate it
from the one-excitation overlap $\langle\Psi|\mathbb{B}(\{z\},\emptyset^{\times N-2})$.
First, we use the recursive formula for the Bethe vectors, then the
$KT$-relation (\ref{eq:KTtw})
\begin{equation}
\langle\Psi|\mathbb{B}(\{z\},\emptyset^{\times N-2})=\frac{1}{\lambda_{2}(z)}\left\langle \Psi\right|T_{1,2}(z)|0\rangle=\frac{\mathbf{K}_{1,1}^{-1}(z)}{\lambda_{2}(z)}\left(\hat{\lambda}_{1}(-z)\mathbf{B}\mathbf{K}_{1,2}(z)-\lambda_{2}(z)\mathbf{K}_{1,2}(z)\mathbf{B}\right).
\end{equation}
Here we used the fact that the pseudo-vacuum is a highest weight state.
Exploiting the symmetry property (\ref{eq:symProp}) and the exchange
relation (\ref{eq:felcs}), we obtain:
\begin{equation}
\langle\Psi|\mathbb{B}(\{z\},\emptyset^{\times N-2})=\alpha_{1}(z)\mathbf{B}\mathbf{K}_{1,1}^{-1}(z)\mathbf{K}_{1,2}(z)-\mathbf{K}_{1,1}^{-1}(z)\mathbf{K}_{1,2}(z)\mathbf{B}.\label{eq:oneroot}
\end{equation}
From this, the highest coefficients (HCs) can be expressed:
\begin{equation}
\bar{\mathbf{Z}}(\{z\},\emptyset^{\times N-2})=-\mathbf{K}_{1,1}^{-1}(z)\mathbf{K}_{1,2}(z),\quad\mathbf{Z}(\{z\},\emptyset^{\times N-2})=\mathbf{K}_{1,1}^{-1}(z)\mathbf{K}_{1,2}(z).
\end{equation}
Using the previously derived identity between the HCs (\ref{eq:ZZBcr}),
we obtain an equivalent formula for $\mathbf{Z}(\{z\},\emptyset^{\times N-2})$:
\begin{equation}
\mathbf{Z}(\{z\},\emptyset^{\times N-2})=-\left[\bar{\mathbf{Z}}^{\mathbf{K}^{T}}(\{-z-1\},\emptyset^{\times N-2})\right]^{t_{B}}=\mathbf{K}_{2,1}(-z-1)\mathbf{K}_{1,1}^{-1}(-z-1).
\end{equation}
Substituting back, we get:
\begin{multline}
\left\langle \Psi\right|\frac{\widehat{T}_{1,1}}{\lambda_{2}(z)}(-z)\mathbb{B}(\bar{t})\mathbf{K}_{1,2}(z)+\left\langle \Psi\right|\frac{\widehat{T}_{1,2}}{\lambda_{2}(z)}(-z)\mathbb{B}(\bar{t})\mathbf{K}_{2,2}(z)=\\
\sum_{\mathrm{part}}\bar{\mathbf{Z}}(\bar{t}_{\textsc{ii}}^{1},\{\bar{t}^{\nu}\}_{\nu=2}^{N-1})\mathbf{B}\left(\mathbf{K}_{2,1}(z)\mathbf{K}_{1,1}^{-1}(z)\mathbf{K}_{1,2}(z)-\mathbf{K}_{2,2}(z)\right)\frac{f(\bar{t}^{1},-z-1)}{f(\bar{t}^{2},-z-1)}\frac{f(\bar{t}_{\textsc{i}}^{1},\bar{t}_{\textsc{ii}}^{1})}{h(\bar{t}_{\textsc{i}}^{1},-z-1)}.
\end{multline}
We observe that on the right-hand side, the quantity $\mathbf{K}_{2,2}^{(2)}(z)=\mathbf{G}^{(2)}(z)$
appears, which commutes with $\mathbf{B}$.

Let us continue with the action of $T_{i,2}$. According to section
E.2 of \cite{Gombor:2021hmj}, the $\alpha$-independent terms are:
\begin{equation}
\left\langle \Psi\right|\frac{T_{i,2}(z)}{\lambda_{2}(z)}\mathbb{B}(\bar{t})\cong\sum_{\mathrm{part}}\bar{\mathbf{Z}}(\bar{w}_{\textsc{ii}}^{1},\{\bar{t}_{\textsc{ii}}^{\nu}\}_{\nu=2}^{i-1},\{\bar{t}^{\nu}\}_{\nu=i}^{N-1})\mathbf{B}\frac{f(\bar{w}_{\textsc{i}}^{1},\bar{w}_{\textsc{ii}}^{1})}{h(\bar{w}_{\textsc{i}}^{1},z)}\frac{f(\bar{t}_{\textsc{i}}^{2},\bar{t}_{\textsc{ii}}^{2})f(\bar{t}^{2},z)}{h(\bar{t}_{\textsc{i}}^{2},\bar{w}_{\textsc{i}}^{1})f(\bar{t}_{\textsc{i}}^{2},\bar{w}_{\textsc{ii}}^{1})}\prod_{s=3}^{i-1}\frac{f(\bar{t}_{\textsc{i}}^{s},\bar{t}_{\textsc{ii}}^{s})}{h(\bar{t}_{\textsc{i}}^{s},\bar{t}_{\textsc{i}}^{s-1})f(\bar{t}_{\textsc{i}}^{s},\bar{t}_{\textsc{ii}}^{s-1})},
\end{equation}
see equation (E.35) in \cite{Gombor:2021hmj}. The sum goes over all
the partitions of $\bar{w}^{1}\vdash\left\{ \bar{w}_{\textsc{i}}^{1},\bar{w}_{\textsc{ii}}^{1}\right\} $,
$\bar{t}^{\nu}\vdash\left\{ \bar{t}_{\textsc{i}}^{\nu},\bar{t}_{\textsc{ii}}^{\nu}\right\} $
where $\bar{w}^{1}=\{z,\bar{t}^{1}\}$ and $\#\bar{w}_{\textsc{i}}^{1}=\#\bar{t}_{\textsc{i}}^{\nu}=1$
for $\nu=2,\dots,i-1$. Substituting back to (\ref{eq:temp3}) we
obtain a recurrence formula (\ref{eq:rect1}) for the HC after the
simplification with $\mathbf{B}$.
\end{proof}
\begin{lem}
The uncrossed HC has the following recursions. The recursion for the
set $\bar{t}^{1}$ is
\begin{multline}
\bar{\mathbf{Z}}(\{z,\bar{t}^{1}\},\left\{ \bar{t}^{\nu}\right\} _{\nu=2}^{N-1})=\\
\sum_{\mathrm{part}}\mathbf{K}_{N,1}^{-1}(z)\bar{\mathbf{Z}}(\{\bar{\omega}_{\textsc{ii}}^{\nu}\}_{\nu=1}^{N-2},\bar{t}_{\textsc{ii}}^{N-1})\mathbf{G}^{(2)}(z)\prod_{s=1}^{N-2}\frac{f(\omega_{\textsc{i}}^{s},\bar{\omega}_{\textsc{ii}}^{s})}{h(\omega_{\textsc{i}}^{s},\omega_{\textsc{i}}^{s-1})f(\omega_{\textsc{i}}^{s},\bar{\omega}_{\textsc{ii}}^{s-1})}\frac{f(\bar{t}_{\textsc{i}}^{N-1},\bar{t}_{\textsc{ii}}^{N-1})f(\bar{t}^{N-1},-z)}{h(\bar{t}_{\textsc{i}}^{N-1},\omega_{\textsc{i}}^{N-2})f(\bar{t}_{\textsc{i}}^{N-1},\bar{\omega}_{\textsc{ii}}^{N-2})f(\bar{t}^{2},z)}-\\
-\sum_{i=2}^{N}\mathbf{K}_{N,1}^{-1}(z)\mathbf{K}_{N,i}(z)\sum_{\mathrm{part}}\bar{\mathbf{Z}}(\bar{w}_{\textsc{ii}}^{1},\{\bar{t}_{\textsc{ii}}^{\nu}\}_{\nu=2}^{i-1},\{\bar{t}^{\nu}\}_{\nu=i}^{N-1})\frac{f(\bar{w}_{\textsc{i}}^{1},\bar{w}_{\textsc{ii}}^{1})}{h(\bar{w}_{\textsc{i}}^{1},z)}\frac{f(\bar{t}_{\textsc{i}}^{2},\bar{t}_{\textsc{ii}}^{2})}{h(\bar{t}_{\textsc{i}}^{2},\bar{w}_{\textsc{i}}^{1})f(\bar{t}_{\textsc{i}}^{2},\bar{w}_{\textsc{ii}}^{1})}\prod_{s=3}^{i-1}\frac{f(\bar{t}_{\textsc{i}}^{s},\bar{t}_{\textsc{ii}}^{s})}{h(\bar{t}_{\textsc{i}}^{s},\bar{t}_{\textsc{i}}^{s-1})f(\bar{t}_{\textsc{i}}^{s},\bar{t}_{\textsc{ii}}^{s-1})},\label{eq:rect1-1}
\end{multline}
In the second line the sum goes for the partitions $\bar{t}^{N-1}\vdash\bar{t}_{\textsc{i}}^{N-1}\cup\bar{t}_{\textsc{ii}}^{N-1}$,$\bar{\omega}^{\nu}\vdash\bar{\omega}_{\textsc{i}}^{\nu}\cup\bar{\omega}_{\textsc{ii}}^{\nu}$,
for $\nu=1,\dots,N-2$ where $\bar{\omega}^{\nu}=\{-z,\bar{t}^{\nu}\}$
and $\#\bar{\omega}_{\textsc{i}}^{\nu}=\#\bar{t}_{\textsc{i}}^{N-1}=1$.
We also set $\bar{\omega}_{\textsc{i}}^{0}=\{-z\}$ and $\bar{\omega}_{\textsc{i}}^{0}=\emptyset$.
In the second line $\bar{w}^{1}=\{z,\bar{t}^{1}\}$ and the sum goes
over all the partitions of $\bar{w}^{1}\vdash\left\{ \bar{w}_{\textsc{i}}^{1},\bar{w}_{\textsc{ii}}^{1}\right\} $,
$\bar{t}^{\nu}\vdash\left\{ \bar{t}_{\textsc{i}}^{\nu},\bar{t}_{\textsc{ii}}^{\nu}\right\} $
where $\#\bar{w}_{\textsc{i}}^{1}=\#\bar{t}_{\textsc{i}}^{\nu}=1$
and $\nu=2,\dots,i-1$. 

The recursion for the set $\bar{t}^{N-1}$ is
\begin{multline}
\bar{\mathbf{Z}}(\left\{ \bar{t}^{s}\right\} _{s=1}^{N-2},\{z,\bar{t}^{N-1}\})=\sum_{i=1}^{N}\sum_{j=1}^{N-1}f(\bar{t}^{1},-z)\sum_{\mathrm{part}}\mathbf{K}_{j,i}(-z)\bar{\mathbf{Z}}(\left\{ \bar{t}_{\textsc{ii}}^{s}\right\} _{s=1}^{N-2},\bar{t}^{N-1})\mathbf{K}_{N,1}^{-1}(-z)\times\\
\prod_{s=1}^{i-1}\frac{f(\bar{t}_{\textsc{i}}^{s},\bar{t}_{\textsc{ii}}^{s})}{h(\bar{t}_{\textsc{i}}^{s},\bar{t}_{\textsc{i}}^{s-1})f(\bar{t}_{\textsc{i}}^{s},\bar{t}_{\textsc{ii}}^{s-1})}\frac{\prod_{\nu=j}^{N-2}g(\bar{t}_{\textsc{iii}}^{\nu+1},\bar{t}_{\textsc{iii}}^{\nu})f(\bar{t}_{\textsc{iii}}^{\nu},\bar{t}_{\textsc{ii}}^{\nu})f(\bar{t}_{\textsc{iii}}^{\nu},\bar{t}_{\textsc{i}}^{\nu})}{\prod_{\nu=j}^{N-1}f(\bar{t}_{\textsc{iii}}^{\nu},\bar{t}^{\nu-1})},\label{eq:recUTWN}
\end{multline}
where we sum up to the partitions $\bar{t}^{\nu}\vdash\bar{t}_{\textsc{i}}^{\nu}\cup\bar{t}_{\textsc{ii}}^{\nu}\cup\bar{t}_{\textsc{iii}}^{\nu}$
for $\nu=1,\dots,N-2$ where $\#\bar{t}_{\textsc{i}}^{\nu}=\Theta(i-1-\nu)$,
$\#\bar{t}_{\textsc{iii}}^{\nu}=\Theta(\nu-j)$ and $\bar{t}_{\textsc{iii}}^{N-1}=\{z\},\bar{t}_{\textsc{ii}}^{N-1}=\bar{t}^{N-1}$
and $\bar{t}_{\textsc{i}}^{0}=\{-z\}$. 
\end{lem}
\begin{proof}
We can see that this formula matches equations (4.11) and (4.12) from
the paper \cite{Gombor:2021hmj} for scalar quantities. In the lemma
concerning the crossed highest coefficient (HC), we saw that the proof
can be generalized to matrix-valued quantities with minimal modifications.
This also holds true in the uncrossed case. If we follow the derivation
in section E.1 of \cite{Gombor:2021hmj} and apply the substitution
\begin{equation}
\mathbf{G}^{(2)}(z)=\mathbf{K}_{N-1,2}^{(2)}(z)=\mathbf{K}_{N-1,2}(z)-\mathbf{K}_{N-1,1}(z)\mathbf{K}_{N,1}^{-1}(z)\mathbf{K}_{N,2}(z)
\end{equation}
using the commutation relation $[\mathbf{B},\mathbf{G}^{(2)}(z)]=0$,
we obtain the proof of the lemma.
\end{proof}
\begin{lem}
\label{lem:HC-pole-cr}The crossed HC has poles at $t_{l}^{1}\to-t_{k}^{1}-1$:
\begin{equation}
\bar{\mathbf{Z}}(\bar{t})\to\frac{1}{t_{l}^{1}+t_{k}^{1}+1}\frac{f(\bar{\tau}^{1},t_{l}^{1})f(\bar{\tau}^{1},t_{k}^{1})}{f(\bar{t}^{2},t_{l}^{1})f(\bar{t}^{2},t_{k}^{1})}\bar{\mathbf{Z}}(\bar{\tau})\mathbf{F}^{(1)}(t_{k}^{1})+reg.\label{eq:poleHC}
\end{equation}
where $\bar{\tau}=\bar{t}\backslash\{t_{k}^{1},t_{l}^{1}\}$ .
\end{lem}
\begin{proof}
We follow section F.1 $(\nu=1)$ of \cite{Gombor:2021hmj}, which
concerns the poles of the scalar HC. The proof proceeds by induction
on $r_{1}$. 

Let us calculate the pole for $r_{1}=2$. Substitute into (\ref{eq:rect1}).
At the pole $t^{1}+z+1\to0$, only the first term on the right-hand
side contains a contribution.
\begin{equation}
\bar{\mathbf{Z}}(\{z,t^{1}\},\left\{ \bar{t}^{\nu}\right\} _{\nu=2}^{N-1})\to-\frac{1}{t^{1}+z+1}\frac{1}{f(\bar{t}^{2},-z-1)f(\bar{t}^{2},z)}\mathbf{K}_{1,1}^{-1}(z)\bar{\mathbf{Z}}(\emptyset,\{\bar{t}^{\nu}\}_{\nu=2}^{N-1})\mathbf{G}^{(2)}(z)+reg.
\end{equation}
Since the HC $\bar{\mathbf{Z}}(\emptyset,\{\bar{t}^{\nu}\}_{\nu=2}^{N-1})$
is built from the components of second nested $K$-mátrix $\mathbf{K}^{(2)}$
which commute with $\mathbf{K}_{1,1}^{-1}$, we have
\begin{equation}
\left[\mathbf{K}_{1,1}^{-1}(z),\bar{\mathbf{Z}}(\emptyset,\{\bar{t}^{\nu}\}_{\nu=2}^{N-1})\right]=0,
\end{equation}
therefore we just proved (\ref{eq:poleHC}) for $r_{2}=2$. 

Now assume the statement holds for $r_{1}$ or fewer Bethe roots of
type $\bar{t}^{1}$. Use the recursive equation (\ref{eq:rect1})
for $r_{1}+1$ first-type roots:
\begin{multline}
\bar{\mathbf{Z}}(\{z,\bar{t}^{1}\},\left\{ \bar{t}^{\nu}\right\} _{\nu=2}^{N-1})=-\sum_{\mathrm{part}}\mathbf{K}_{1,1}^{-1}(z)\bar{\mathbf{Z}}(\bar{t}_{\textsc{ii}}^{1},\{\bar{t}^{\nu}\}_{\nu=2}^{N-1})\mathbf{G}^{(2)}(z)\frac{f(\bar{t}^{1},-z-1)}{f(\bar{t}^{2},-z-1)f(\bar{t}^{2},z)}\frac{f(\bar{t}_{\textsc{i}}^{1},\bar{t}_{\textsc{ii}}^{1})}{h(\bar{t}_{\textsc{i}}^{1},-z-1)}-\\
-\sum_{i=2}^{N}\mathbf{K}_{1,1}^{-1}(z)\mathbf{K}_{1,i}(z)\sum_{\mathrm{part}}\bar{\mathbf{Z}}(\bar{w}_{\textsc{ii}}^{1},\{\bar{t}_{\textsc{ii}}^{\nu}\}_{\nu=2}^{i-1},\{\bar{t}^{\nu}\}_{\nu=i}^{N-1})\frac{f(\bar{w}_{\textsc{i}}^{1},\bar{w}_{\textsc{ii}}^{1})}{h(\bar{w}_{\textsc{i}}^{1},z)}\frac{f(\bar{t}_{\textsc{i}}^{2},\bar{t}_{\textsc{ii}}^{2})}{h(\bar{t}_{\textsc{i}}^{2},\bar{w}_{\textsc{i}}^{1})f(\bar{t}_{\textsc{i}}^{2},\bar{w}_{\textsc{ii}}^{1})}\prod_{s=3}^{i-1}\frac{f(\bar{t}_{\textsc{i}}^{s},\bar{t}_{\textsc{ii}}^{s})}{h(\bar{t}_{\textsc{i}}^{s},\bar{t}_{\textsc{i}}^{s-1})f(\bar{t}_{\textsc{i}}^{s},\bar{t}_{\textsc{ii}}^{s-1})}.\label{eq:rect1-2}
\end{multline}
Now examine the pole at $t_{k}^{1}+t_{l}^{1}+1=0$. Since the HC is
a symmetric function of the first-type roots, it is sufficient to
prove the statement for this pair. On the right-hand side of the equation,
only the HCs contain such a pole, but in those HCs, there are $r_{1}$
or fewer first-type Bethe roots. We apply the induction hypothesis
to those. The pole appears in the first line only if $t_{k}^{1},t_{l}^{1}\in\bar{t}_{\textsc{ii}}^{1}$,
and in the second line if $t_{k}^{1},t_{l}^{1}\in\bar{w}_{\textsc{ii}}^{1}$.
The pole can be written in the following form
\begin{equation}
\bar{\mathbf{Z}}(\{z,\bar{t}^{1}\},\left\{ \bar{t}^{\nu}\right\} _{\nu=2}^{N-1})\to\frac{1}{t_{l}^{1}+t_{k}^{1}+1}\frac{f(\bar{\mathrm{w}}^{1},t_{l}^{1})f(\bar{\mathrm{w}}^{1},t_{k}^{1})}{f(\bar{t}^{2},t_{l}^{1})f(\bar{t}^{2},t_{k}^{1})}\mathbf{Q}(\bar{\mathrm{w}}^{1},\left\{ \bar{t}^{\nu}\right\} _{\nu=2}^{N-1})\mathbf{F}^{(1)}(t_{k}^{1})+reg.\label{eq:Zbw}
\end{equation}
where $\bar{\mathrm{w}}^{1}=\bar{w}^{1}\backslash\{t_{k}^{1},t_{l}^{1}\}=\{z,\bar{\tau}^{1}\}$.
Additionally, we introduce the notation $\bar{\mathrm{w}}_{\textsc{ii}}^{1}=\bar{w}_{\textsc{ii}}^{1}\backslash\{t_{k}^{1},t_{l}^{1}\}$.
Using the induction hypothesis, the right-hand side of (\ref{eq:rect1-2})
for each HC gives a $Q$-operator that can be expressed as follows
\begin{multline}
\mathbf{Q}(\{z,\bar{\tau}^{1}\},\left\{ \bar{t}^{\nu}\right\} _{\nu=2}^{N-1})=-\sum_{\mathrm{part}}\mathbf{K}_{1,1}^{-1}(z)\bar{\mathbf{Z}}(\bar{\tau}_{\textsc{ii}}^{1},\{\bar{t}^{\nu}\}_{\nu=2}^{N-1})\mathbf{G}^{(2)}(z)\frac{f(\bar{\tau}^{1},-z-1)}{f(\bar{t}^{2},-z-1)f(\bar{t}^{2},z)}\frac{f(\bar{\tau}_{\textsc{i}}^{1},\bar{\tau}_{\textsc{ii}}^{1})}{h(\bar{\tau}_{\textsc{i}}^{1},-z-1)}-\\
-\sum_{i=2}^{N}\mathbf{K}_{1,1}^{-1}(z)\mathbf{K}_{1,i}(z)\sum_{\mathrm{part}}\bar{\mathbf{Z}}(\bar{\mathrm{w}}_{\textsc{ii}}^{1},\{\bar{t}_{\textsc{ii}}^{\nu}\}_{\nu=2}^{i-1},\{\bar{t}^{\nu}\}_{\nu=i}^{N-1})\frac{f(\bar{\mathrm{w}}_{\textsc{i}}^{1},\bar{\mathrm{w}}_{\textsc{ii}}^{1})}{h(\bar{\mathrm{w}}_{\textsc{i}}^{1},z)}\frac{f(\bar{t}_{\textsc{i}}^{2},\bar{t}_{\textsc{ii}}^{2})}{h(\bar{t}_{\textsc{i}}^{2},\bar{\mathrm{w}}_{\textsc{i}}^{1})f(\bar{t}_{\textsc{i}}^{2},\bar{\mathrm{w}}_{\textsc{ii}}^{1})}\prod_{s=3}^{i-1}\frac{f(\bar{t}_{\textsc{i}}^{s},\bar{t}_{\textsc{ii}}^{s})}{h(\bar{t}_{\textsc{i}}^{s},\bar{t}_{\textsc{i}}^{s-1})f(\bar{t}_{\textsc{i}}^{s},\bar{t}_{\textsc{ii}}^{s-1})},\label{eq:Qeq}
\end{multline}
where we used the commutation relation
\begin{equation}
[\mathbf{F}^{(1)}(u),\mathbf{G}^{(2)}(v)]=0.
\end{equation}
We observe that the right-hand side of (\ref{eq:Qeq}) is exactly
the right-hand side of the recursive equation (\ref{eq:rect1}), i.e.,
\begin{equation}
\mathbf{Q}(\{z,\bar{\tau}^{1}\},\left\{ \bar{t}^{\nu}\right\} _{\nu=2}^{N-1})=\bar{\mathbf{Z}}(\{z,\bar{\tau}^{1}\},\left\{ \bar{t}^{\nu}\right\} _{\nu=2}^{N-1}).
\end{equation}
Substituting back into (\ref{eq:Zbw}), the lemma holds for $r_{1}+1$
first-type Bethe roots.
\end{proof}
\begin{lem}
\label{lem:HC-pole-uncr}The uncrossed HC-s have a pole at $t_{l}^{N-1}\to-t_{k}^{1}$:
\begin{equation}
\bar{\mathbf{Z}}(\bar{t})\to\frac{1}{t_{l}^{N-1}+t_{k}^{1}}\frac{f(\bar{\tau}^{1},t_{k}^{1})f(\bar{\tau}^{N-1},t_{l}^{N-1})}{f(\bar{t}^{2},t_{k}^{1})}\bar{\mathbf{Z}}(\bar{\tau})\mathbf{F}^{(1)}(t_{k}^{1})+reg.\label{eq:poleHC-2}
\end{equation}
where $\bar{\tau}=\bar{t}\backslash\{t_{k}^{1},t_{l}^{N-1}\}$.
\end{lem}
\begin{proof}
We follow section F.2 $(\nu=1)$ of \cite{Gombor:2021hmj}, which
concerns the poles of the scalar highest coefficient (HC). The proof
proceeds by induction on both $r_{1}$ and $r_{N-1}$. 

Let us calculate the pole for $r_{1}=r_{N-1}=1$:
\begin{multline}
\bar{\mathbf{Z}}(\{z\},\left\{ \bar{t}^{\nu}\right\} _{\nu=2}^{N-2},\{t^{N-1}\})\to\frac{1}{t^{N-1}+z}\times\\
\sum_{\mathrm{part}}\mathbf{K}_{N,1}^{-1}(z)\bar{\mathbf{Z}}(\emptyset,\{\bar{\omega}_{\textsc{ii}}^{\nu}\}_{\nu=2}^{N-2},\emptyset)\mathbf{G}^{(2)}(z)\prod_{s=2}^{N-2}\frac{f(\omega_{\textsc{i}}^{s},\bar{\omega}_{\textsc{ii}}^{s})}{h(\omega_{\textsc{i}}^{s},\omega_{\textsc{i}}^{s-1})f(\omega_{\textsc{i}}^{s},\bar{\omega}_{\textsc{ii}}^{s-1})}\frac{1}{h(-z,\omega_{\textsc{i}}^{N-2})f(-z,\bar{\omega}_{\textsc{ii}}^{N-2})f(\bar{t}^{2},z)}+reg.
\end{multline}
In the denominator there is a factor $f(-z,\bar{\omega}_{\textsc{ii}}^{N-2})$
therefore the residue is nonzero only when $-z\notin\bar{\omega}_{\textsc{ii}}^{N-2}\Rightarrow\bar{\omega}_{\textsc{i}}^{N-2}=\{-z\}$.
In an analogous way, the terms $f(\bar{\omega}_{\textsc{i}}^{s},\bar{\omega}_{\textsc{ii}}^{s-1})$
imply that $\bar{\omega}_{\textsc{i}}^{s}=\{-z\}$ for $s=2,\dots,N-3$
which means there is only one non-vanishing term in the sum
\begin{equation}
\bar{\mathbf{Z}}(\{z\},\left\{ \bar{t}^{\nu}\right\} _{\nu=2}^{N-2},\{t^{N-1}\})\to\frac{1}{t^{N-1}+z}\frac{1}{f(\bar{t}^{2},z)}\mathbf{K}_{N,1}^{-1}(z)\bar{\mathbf{Z}}(\emptyset,\{\bar{t}^{\nu}\}_{\nu=2}^{N-2},\emptyset)\mathbf{G}^{(2)}(z)+reg.
\end{equation}
Since the HC $\bar{\mathbf{Z}}(\emptyset,\{\bar{t}^{\nu}\}_{\nu=2}^{N-2},\emptyset)$
is build from the components of second nested $K$-mátrix $\mathbf{K}^{(2)}$
which commute with $\mathbf{K}_{N,1}^{-1}$, we have
\begin{equation}
\left[\mathbf{K}_{N,1}^{-1}(z),\bar{\mathbf{Z}}(\emptyset,\{\bar{t}^{\nu}\}_{\nu=2}^{N-2},\emptyset)\right]=0,
\end{equation}
therefore we just proved (\ref{eq:poleHC-2}) for $r_{1}=r_{N-1}=1$. 

The proof for $r_{1}+r_{N-1}>2$ can be done by induction. In the
lemma concerning the crossed HCs, we saw that the proof is entirely
analogous to the scalar case, provided we use the commutation relations
of the $F$- and $G$-operators. The proof here is also analogous
to the scalar case found in Appendix F.2 $(\nu=1)$ of \cite{Gombor:2021hmj},
if we use the relation
\begin{equation}
[\mathbf{F}^{(1)}(u),\mathbf{G}^{(1)}(v)]=[\mathbf{F}^{(1)}(u),\mathbf{G}^{(2)}(v)]=0.
\end{equation}
\end{proof}
We see that Lemmas \ref{lem:HC-pole-cr} and \ref{lem:HC-pole-uncr}
together prove Theorem \ref{thm:HC-poles} for $\nu=1$. Theorem \ref{thm:HC-poles}
can also be proven by another induction, which involves a very similar
computation to the $\nu=1$ case. The proof for the scalar HC in the
case $\nu>1$ is found in Appendices F.1 and F.2 of \cite{Gombor:2021hmj},
and this can be applied to the matrix-valued HCs without modification,
one only needs to use the relation
\begin{equation}
[\mathbf{F}^{(\nu)}(u),\mathbf{G}^{(\mu)}(v)]=0.
\end{equation}

\subsection{Proofs for the overlap functions $S^{(\ell)}$}

\subsubsection{Pair structure limit of the off-shell overlaps}

In this section, we prove Theorem \ref{thm:Spair}.
\begin{proof}
To prove the theorem, we need to take the limit of the off-shell overlap
as $t_{l}^{\tilde{\nu}}\to-t_{k}^{\nu}-\nu\tilde{c}$. Since the overlap
depends only on the $\alpha$-functions, the derivative terms $X_{k}^{\nu}$
can arise from the poles of the HCs. From the sum rule of the overlap,
we see that the HC $\mathbf{Z}(\bar{t}_{\textsc{i}})$ has a pole
if $t_{k}^{\nu},t_{l}^{\tilde{\nu}}\in\bar{t}_{\textsc{i}}$ and $\bar{\mathbf{Z}}(\bar{t}_{\textsc{i}})$
has a pole if $t_{k}^{\nu},t_{l}^{\tilde{\nu}}\in\bar{t}_{\textsc{ii}}$.
Introducing the usual notation $\tau=\bar{t}\backslash\{t_{k}^{\nu},t_{l}^{\tilde{\nu}}\}$
the part of the overlap formula proportional to the poles comes from
summing over the partitions of $\bar{\tau}$. For a given partition
$\bar{\tau}\vdash\bar{\tau}_{\textsc{i}}\cup\bar{\tau}_{\textsc{ii}}$
two terms contribute from the original sum: $\bar{t}_{\textsc{i}}=\bar{\tau}_{\textsc{i}}\cup\{t_{k}^{\nu},t_{l}^{\tilde{\nu}}\}$,
$\bar{t}_{\textsc{ii}}=\bar{\tau}_{\textsc{ii}}$ and $\bar{t}_{\textsc{i}}=\bar{\tau}_{\textsc{i}}$,
$\bar{t}_{\textsc{ii}}=\bar{\tau}_{\textsc{ii}}\cup\{t_{k}^{\nu},t_{l}^{\tilde{\nu}}\}$.
Based on this, the limit of the overlap as $t_{l}^{\tilde{\nu}}\to-t_{k}^{\nu}-\nu\tilde{c}$
can be written as follows
\begin{align}
\mathbf{S}_{\bar{\alpha},\mathbf{B}}(\bar{t}) & \to\frac{1}{t_{l}^{\tilde{\nu}}+t_{k}^{\nu}+\nu\tilde{c}}\left(1-\alpha_{\nu}(t_{k}^{\nu})\alpha_{\tilde{\nu}}(t_{l}^{\tilde{\nu}})\right)\sum_{\mathrm{part}(\bar{\tau})}\frac{\prod_{s=1}^{N-1}f(\bar{\tau}_{\textsc{ii}}^{s},\bar{\tau}_{\textsc{i}}^{s})}{\prod_{s=1}^{N-2}f(\bar{\tau}_{\textsc{ii}}^{s+1},\bar{\tau}_{\textsc{i}}^{s})}\bar{\mathbf{Z}}(\bar{\tau}_{\textsc{ii}})\mathbf{F}^{(\nu)}(t_{k}^{\nu})\mathbf{B}\mathbf{Z}(\bar{\tau}_{\textsc{i}})\times\nonumber \\
 & \times\frac{f(t_{k}^{\nu},\bar{\tau}_{\textsc{i}}^{\nu})}{f(t_{k}^{\nu},\bar{\tau}_{\textsc{i}}^{\nu-1})}\frac{f(\bar{\tau}_{\textsc{ii}}^{\nu},t_{k}^{\nu})}{f(\bar{\tau}_{\textsc{ii}}^{\nu+1},t_{k}^{\nu})}\frac{f(t_{l}^{\tilde{\nu}},\bar{\tau}_{\textsc{i}}^{\tilde{\nu}})}{f(t_{l}^{\tilde{\nu}},\bar{\tau}_{\textsc{i}}^{\tilde{\nu}-1})}\frac{f(\bar{\tau}_{\textsc{ii}}^{\tilde{\nu}},t_{l}^{\tilde{\nu}})}{f(\bar{\tau}_{\textsc{ii}}^{\tilde{\nu}+1},t_{l}^{\tilde{\nu}})}\prod_{s=1}^{N-1}\alpha_{s}(\bar{\tau}_{\textsc{i}}^{s})+reg.
\end{align}
Here we used the commutation relations (\ref{eq:felcs}) and (\ref{eq:felcsUcr}).
After rearrangement, we obtain:
\begin{align}
\mathbf{S}_{\bar{\alpha},\mathbf{B}}(\bar{t}) & \to\frac{1}{t_{l}^{\tilde{\nu}}+t_{k}^{\nu}+\nu\tilde{c}}\left(1-\alpha_{\nu}(t_{k}^{\nu})\alpha_{\tilde{\nu}}(t_{l}^{\tilde{\nu}})\right)\frac{f(\bar{\tau}^{\nu},t_{k}^{\nu})}{f(\bar{t}^{\nu+1},t_{k}^{\nu})}\frac{f(\bar{\tau}^{\tilde{\nu}},t_{l}^{\tilde{\nu}})}{f(\bar{t}^{\tilde{\nu}+1},t_{l}^{\tilde{\nu}})}\sum_{\mathrm{part}(\bar{\tau})}\frac{\prod_{s=1}^{N-1}f(\bar{\tau}_{\textsc{ii}}^{s},\bar{\tau}_{\textsc{i}}^{s})}{\prod_{s=1}^{N-2}f(\bar{\tau}_{\textsc{ii}}^{s+1},\bar{\tau}_{\textsc{i}}^{s})}\bar{\mathbf{Z}}(\bar{\tau}_{\textsc{ii}})\mathbf{F}^{(\nu)}(t_{k}^{\nu})\mathbf{B}\mathbf{Z}(\bar{\tau}_{\textsc{i}})\times\nonumber \\
 & \times\frac{f(t_{k}^{\nu},\bar{\tau}_{\textsc{i}}^{\nu})}{f(\bar{\tau}_{\textsc{i}}^{\nu},t_{k}^{\nu})}\frac{f(\bar{\tau}_{\textsc{i}}^{\nu+1},t_{k}^{\nu})}{f(t_{k}^{\nu},\bar{\tau}_{\textsc{i}}^{\nu-1})}\frac{f(t_{l}^{\tilde{\nu}},\bar{\tau}_{\textsc{i}}^{\tilde{\nu}})}{f(\bar{\tau}_{\textsc{i}}^{\tilde{\nu}},t_{l}^{\tilde{\nu}})}\frac{f(\bar{\tau}_{\textsc{i}}^{\tilde{\nu}+1},t_{l}^{\tilde{\nu}})}{f(t_{l}^{\tilde{\nu}},\bar{\tau}_{\textsc{i}}^{\tilde{\nu}-1})}\prod_{s=1}^{N-1}\alpha_{s}(\bar{\tau}_{\textsc{i}}^{s})+reg.
\end{align}
Using the definitions of $X_{k}^{\nu}$ (\ref{eq:Xdef}) and $\alpha_{s}^{mod}$
(\ref{eq:amod}), we get:
\begin{align}
\mathbf{S}_{\bar{\alpha},\mathbf{B}}(\bar{t}) & \to X_{k}^{\nu}\frac{f(\bar{\tau}^{\nu},t_{k}^{\nu})}{f(\bar{t}^{\nu+1},t_{k}^{\nu})}\frac{f(\bar{\tau}^{\tilde{\nu}},t_{l}^{\tilde{\nu}})}{f(\bar{t}^{\tilde{\nu}+1},t_{l}^{\tilde{\nu}})}\times\nonumber \\
 & \sum_{\mathrm{part}(\bar{\tau})}\frac{\prod_{s=1}^{N-1}f(\bar{\tau}_{\textsc{ii}}^{s},\bar{\tau}_{\textsc{i}}^{s})}{\prod_{s=1}^{N-2}f(\bar{\tau}_{\textsc{ii}}^{s+1},\bar{\tau}_{\textsc{i}}^{s})}\bar{\mathbf{Z}}(\bar{\tau}_{\textsc{ii}})\mathbf{F}^{(\nu)}(t_{k}^{\nu})\mathbf{B}\mathbf{Z}(\bar{\tau}_{\textsc{i}})\prod_{s=1}^{N-1}\alpha_{s}^{mod}(\bar{\tau}_{\textsc{i}}^{s})+\tilde{\mathbf{S}}_{\alpha,\mathbf{B}},
\end{align}
where $\tilde{\mathbf{S}}_{\alpha,\mathbf{B}}$ does not depend on
$X_{k}^{\nu}$. Taking the trace:
\begin{align}
\mathrm{tr}_{\mathcal{V}_{B}}\left[\mathbf{S}_{\bar{\alpha},\mathbf{B}}(\bar{t})\right] & \to X_{k}^{\nu}\frac{f(\bar{\tau}^{\nu},t_{k}^{\nu})}{f(\bar{t}^{\nu+1},t_{k}^{\nu})}\frac{f(\bar{\tau}^{\tilde{\nu}},t_{l}^{\tilde{\nu}})}{f(\bar{t}^{\tilde{\nu}+1},t_{l}^{\tilde{\nu}})}\times\nonumber \\
 & \sum_{\ell=1}^{d_{B}}\mathcal{F}_{\ell}^{(\nu)}(t_{k}^{\nu})\beta_{\ell}\sum_{\mathrm{part}(\bar{\tau})}\frac{\prod_{s=1}^{N-1}f(\bar{\tau}_{\textsc{ii}}^{s},\bar{\tau}_{\textsc{i}}^{s})}{\prod_{s=1}^{N-2}f(\bar{\tau}_{\textsc{ii}}^{s+1},\bar{\tau}_{\textsc{i}}^{s})}\left(\mathbf{Z}(\bar{\tau}_{\textsc{i}})\bar{\mathbf{Z}}(\bar{\tau}_{\textsc{ii}})\right)_{\ell,\ell}\prod_{s=1}^{N-1}\alpha_{s}^{mod}(\bar{\tau}_{\textsc{i}}^{s})+\tilde{\mathbf{S}}_{\alpha,\mathbf{B}},
\end{align}
i.e.,
\begin{equation}
\sum_{\ell=1}^{d_{B}}\beta_{\ell}S_{\bar{\alpha}}^{(\ell)}(\bar{t})\to X_{k}^{\nu}\frac{f(\bar{\tau}^{\nu},t_{k}^{\nu})}{f(\bar{t}^{\nu+1},t_{k}^{\nu})}\frac{f(\bar{\tau}^{\tilde{\nu}},t_{l}^{\tilde{\nu}})}{f(\bar{t}^{\tilde{\nu}+1},t_{l}^{\tilde{\nu}})}\sum_{\ell=1}^{d_{B}}\mathcal{F}_{\ell}^{(\nu)}(t_{k}^{\nu})\beta_{\ell}S_{\bar{\alpha}^{mod}}^{(\ell)}(\bar{\tau})+\mathrm{tr}\tilde{\mathbf{S}}_{\alpha,\mathbf{B}}.
\end{equation}
Assuming that the $\beta_{\ell}$ variables are linearly independent,
we obtain the desired formula.
\end{proof}

\subsubsection{On-shell limit}

In this section, we prove Theorem \ref{thm:onshell}. To prove the
theorem, we need to define the Korepin criterion and state a lemma
related to it.
\begin{defn}
Let $\mathcal{N}^{(\bar{r})}(\bar{t}|\bar{X})$ be a function of $2\sum_{s=1}^{n}r_{s}$
variables, where $\bar{t}=\left\{ \bar{t}^{s}\right\} _{s=1}^{n}$
and $\bar{X}=\left\{ \bar{X}^{s}\right\} _{s=1}^{n}$ with cardinalities
$\#\bar{t}^{s}=\#\bar{X}^{s}=r_{s}$. The Korepin criterion is defined
as follows:
\end{defn}
\begin{enumerate}[label=(\roman*)]
\item \label{enum:prop1}The function $\mathcal{N}^{(\bar{r})}(\bar{t}|\bar{X})$
is symmetric over the replacement of the pairs $(X_{j}^{\mu},t_{j}^{\mu})\leftrightarrow(X_{k}^{\mu},t_{k}^{\mu})$.
\item \label{enum:prop2}It is linear function of each $X_{j}^{\mu}$.
\item \label{enum:prop3}$\mathcal{N}^{(\bar{1}_{\nu})}(\dots,\emptyset,\{t^{\nu}\},\emptyset,\dots|\dots,\emptyset,\{X^{\nu}\},\emptyset,\dots)=X^{\nu}$,
where $\bar{1}_{\nu}\in\mathbb{N}^{n}$ is an $n$-component vector
whose components are defined as $(\mathbf{1}_{\nu})_{k}=\delta_{k,\nu}$
for $k=1,\dots,n$.
\item \label{enum:prop4}The coefficient of $X_{j}^{\mu}$ is given by the
function $\mathcal{N}^{(\bar{r}-\bar{1}_{\mu})}$ with modified parameters
$X_{k}^{\nu}$
\begin{equation}
\frac{\partial\mathcal{N}^{(\bar{r})}(\bar{t}|\bar{X})}{\partial X_{j}^{\mu}}=\mathcal{N}^{(\bar{r}-\bar{1}_{\mu})}(\bar{\tau},\bar{\mathcal{X}}^{mod}),\label{eq:derivFX}
\end{equation}
where $\bar{\tau}=\bar{t}\backslash\{t_{k}^{\mu}\}$, $\bar{\mathcal{X}}=\bar{X}\backslash\{X_{k}^{\mu}\}$
and the original variables $X_{k}^{\nu}$ should be replaced by $X_{k}^{\nu,mod}$
which are defined with (\ref{eq:Xdef}) and (\ref{eq:amod}).
\item \label{enum:prop5}$\mathcal{N}^{(\bar{r})}(\bar{t}|\bar{X})=0$,
if all $X_{j}^{\mu}=0$.
\end{enumerate}
\begin{lem}
\label{lem:Korepin}If a set of functions $\mathcal{N}^{(\bar{r}^{+})}(\bar{t}^{+}|\bar{X}^{+})$
satisfies the Korepin criterion, then
\begin{equation}
\mathcal{N}^{(\bar{r}^{+})}(\bar{t}^{+}|\bar{X}^{+})=\det G_{+}.
\end{equation}
\end{lem}
\begin{proof}
The proof proceeds recursively in the total number of variables $\mathbf{r}^{+}=\sum_{s=1}^{n}r_{s}^{+}$.
This is the same as Proposition 4.1 in \cite{Hutsalyuk:2017way}.
\end{proof}
Using the above lemma and the earlier Theorem \ref{thm:Spair} we
can now easily prove Theorem \ref{thm:onshell} for on-shell overlaps.
\begin{proof}
The derivation of Theorem \ref{thm:onshell} follows Appendix H of
\cite{Gombor:2021hmj}. We begin with the crossed case and introduce
the normalized overlap functions.
\begin{equation}
\mathcal{N}^{(\bar{r}^{+})}(\bar{t}^{+}|\bar{X}^{+})=\frac{S^{(\ell)}(\bar{t}^{+}|\bar{X}^{+})}{\prod_{\nu=1}^{N-1}\mathcal{F}_{\ell}^{(\nu)}(\bar{t}^{+,\nu})\prod_{k\neq l}f(t_{l}^{+,\nu},t_{k}^{+,\nu})\prod_{k<l}f(t_{l}^{+,\nu},-t_{k}^{+,\nu}-\nu)f(-t_{k}^{+,\nu}-\nu,t_{l}^{+,\nu})}.
\end{equation}
According to Lemma \ref{lem:Korepin}, it is sufficient to show that
the functions $\mathcal{N}^{(\bar{r}^{+})}$ satisfy the Korepin conditions.
Property \ref{enum:prop1} follows from the definition of the overlaps.
Properties \ref{enum:prop2} and \ref{enum:prop4} follow from Theorem
\ref{thm:Spair}. Property \ref{enum:prop5} follows from the fact
that for a Bethe vector that does not satisfy the pair structures,
the on-shell overlap is zero, i.e., the overlap
\begin{equation}
\langle\mathrm{MPS}|\mathbb{B}(\bar{t})=\sum_{\ell=1}^{d_{B}}\beta_{\ell}S_{\bar{\alpha}}^{(\ell)}(\bar{t})
\end{equation}
vanishes in the on-shell limit. In the generalized model, the variables
$\alpha_{k}^{\nu}\equiv\alpha_{\nu}(t_{k}^{\nu})$ are algebraically
independent of the $t_{k}^{\nu}$, so any set of Bethe roots can be
on-shell with appropriately chosen $\alpha$'s. Thus, in the generalized
model, the on-shell limit is equivalent to:
\begin{equation}
\alpha_{k}^{\nu}\to\mathcal{A}_{k}^{\nu}(\bar{t})\equiv\frac{f(t_{k}^{\mu},\bar{t}_{k}^{\mu})}{f(\bar{t}_{k}^{\mu},t_{k}^{\mu})}\frac{f(\bar{t}^{\mu+1},t_{k}^{\mu})}{f(t_{k}^{\mu},\bar{t}^{\mu-1})}.
\end{equation}
The on-shell limit of a general overlap is zero.
\begin{equation}
\lim_{\alpha_{k}^{\nu}\to\mathcal{A}_{k}^{\nu}(\bar{t})}\langle\mathrm{MPS}|\mathbb{B}(\bar{t})=\sum_{\ell=1}^{d_{B}}\beta_{\ell}\left(\lim_{\alpha_{k}^{\nu}\to\mathcal{A}_{k}^{\nu}(\bar{t})}S_{\bar{\alpha}}^{(\ell)}(\bar{t})\right)=0.
\end{equation}
Assuming the $\beta_{\ell}$ are also independent variables, the overlap
functions vanish in the on-shell limit.
\begin{equation}
\lim_{\alpha_{k}^{\nu}\to\mathcal{A}_{k}^{\nu}(\bar{t})}S_{\bar{\alpha}}^{(\ell)}(\bar{t})=0.
\end{equation}
The definition of $\mathcal{N}^{(\bar{r}^{+})}$ involves two limits:
the pair structure limit and the on-shell limit. These limits are
not interchangeable, if we take the on-shell limit first, we get zero;
however, if we start with the pair structure limit, we obtain $\mathcal{N}^{(\bar{r}^{+})}$,
i.e.,
\begin{equation}
\lim_{\alpha_{k}^{\nu}\to\mathcal{A}_{k}^{\nu}(\bar{t})}\lim_{t^{-}\to\pi(\bar{t}^{+})}S_{\bar{\alpha}}^{(\ell)}(\bar{t})=S^{(\ell)}(\bar{t}^{+}|\bar{X}^{+}).
\end{equation}
The reason the two limits are not interchangeable is that the overlap
function contains formal first-order poles, and in the pair structure
limit, the first-order terms of the expressions $\alpha_{\nu}(t_{k}^{-,\nu})$
in $(t_{k}^{-,\nu}+t_{k}^{+,\nu}+\nu)$ are needed:
\begin{equation}
\alpha_{\nu}(t_{k}^{-,\nu})=\alpha_{\nu}(t_{k}^{+,\nu})^{-1}+(t_{k}^{-,\nu}+t_{k}^{+,\nu}+\nu)\alpha_{\nu}(t_{k}^{+,\nu})^{-2}\alpha'_{\nu}(t_{k}^{+,\nu})+\mathcal{O}((t_{k}^{-,\nu}+t_{k}^{+,\nu}+\nu)^{2}).
\end{equation}
It is evident that the limits are not interchangeable because if we
take the pair structure limit first, we get terms proportional to
$\alpha'_{\nu}(t_{k}^{+,\nu})$ which we do not get if we start with
the on-shell limit. If $X_{k}^{\nu}=0$, then $\alpha'_{\nu}(t_{k}^{+,\nu})=0$,
and in this case, the two limits are interchangeable, i.e., 
\begin{equation}
S^{(\ell)}(\bar{t}^{+}|\bar{0})=0,
\end{equation}
meaning that \ref{enum:prop5} also holds. Property \ref{enum:prop3}
follows from the Theorem \ref{thm:Spair} and \ref{enum:prop5}. The
proof in the non-crossed case is completely analogous.
\end{proof}
\bibliographystyle{elsarticle-num}
\bibliography{refs}

\end{document}